\documentclass[11pt]{article}
\usepackage{float}
\usepackage{amsmath,url}
\usepackage{amsthm}
\usepackage{amssymb}
\usepackage{algorithm}
\usepackage{algorithmic}
\usepackage{epsfig}
\usepackage{multirow}
\usepackage{pdflscape}
\usepackage{verbatim}
\usepackage[english]{babel}

\usepackage{graphicx}
\usepackage{caption}
\usepackage{scalerel}
\usepackage[usenames, dvipsnames]{color}
\definecolor{mypink1}{rgb}{0.858, 0.188, 0.478}
\definecolor{mypink2}{RGB}{219, 48, 122}
\definecolor{mypink3}{cmyk}{0, 0.7808, 0.4429, 0.1412}
\definecolor{mygray}{gray}{0.6}

\theoremstyle{plain}
\newtheorem{theorem}{Theorem}[section]
\newtheorem{lemma}[theorem]{Lemma}
\newtheorem{corollary}[theorem]{Corollary}
\newtheorem{proposition}[theorem]{Proposition}

\theoremstyle{definition}
\newtheorem{definition}[theorem]{Definition}

\theoremstyle{remark}
\newtheorem{remark}[theorem]{Remark}

\newcommand{\tail}[3]{\widehat{#1}_{#2,#3}}
\numberwithin{equation}{section}

\def\P{\mathbb P}

\def\cal{\mathcal}

\begin{document}


\title{Trimming the Hill estimator: robustness, optimality and adaptivity\thanks{MK and SB were supported by National Science Foundation grant CNS-1422078; SB and SS
	were partially supported by the NSF grant DMS-1462368.}}



\author{Shrijita Bhattacharya\thanks{Department of Statistics, University of Michigan, 311 West Hall, 1085 S. University Ann Arbor, MI 48109-1107, {\tt \{shrijita, sstoev\}@umich.edu}}\quad Michael Kallitsis\thanks{Merit Network, Inc., 1000 Oakbrook Drive, Suite 200, Ann Arbor, MI 48104, {\tt mgkallit@merit.edu}}\quad Stilian Stoev$^\dagger$}




\maketitle

\begin{abstract} We introduce a trimmed version of the Hill estimator for the index of a heavy-tailed distribution, which is robust to perturbations in the extreme order statistics. 
In the ideal Pareto setting, the estimator is essentially finite-sample efficient among all unbiased estimators with a given strict upper break-down point.  For general heavy-tailed models, 
we establish the asymptotic normality of the estimator under second order conditions and discuss its minimax optimal rate in the Hall class.  We introduce the so-called trimmed Hill plot, which can be used 
to select the number of top order statistics to trim.  We also develop an automatic, data-driven procedure for the choice of trimming.  This results in a new type of robust estimator that can {\em adapt} to the 
unknown level of contamination in the extremes.  As a by-product we also obtain a methodology for identifying extreme outliers in heavy tailed data. The competitive performance of the trimmed Hill 
and adaptive trimmed Hill estimators is illustrated with simulations.  
\end{abstract}

\section{Introduction}

 The estimation of the tail index for heavy-tailed distributions is perhaps one of the most studied problems in
 extreme value theory. Since the seminal works of \cite{Hill:1975}, \cite{pickands:1975}, \cite{hall:1982} and others,
 numerous aspects of this problem and its applications have been explored (see for example the monographs of \cite{embrechts:kluppelberg:mikosch:1997}, 
 \cite{beirlant:goegebeur:teugels:segers:2004}, \cite{dehaan:ferreira:2006}, \cite{resnick:2007} and the references therein).
 
 Given the extensive work on the subject, it may appear naive to hope to say something new.  Nevertheless, some curious 
 aspects of this fundamental problem have remained unexplored. 
 
 Suppose that $X_1,\cdots,X_n$ is an i.i.d. sample from a heavy tailed distribution $F$.  Namely,
 \begin{equation}
 \label{e:heavy-tail}
 \P(X_1> x) \equiv 1-F(x) \sim \ell(x) x^{-1/\xi},\ \ \mbox{ as }x\to\infty,
 \end{equation}
  for some $\xi>0$ and a slowly varying function $\ell:(0,\infty)\rightarrow (0,\infty)$, i.e.,  $\ell(\lambda x)/\ell(x)\to 1,\ x\to\infty,$ for all $\lambda>0$. 
 The parameter $\xi$ will be referred to as the {\em tail index} of $F$.  Its estimation
 is of fundamental importance to the applications of extreme value theory.
 
 The fact that  $\xi$ governs the asymptotic tail-behavior of $F$ means that, in practice, one should estimate it by 
 focusing on the most extreme values of the sample.  In many applications, one quickly runs out of data since only the largest
 few order statistics are utilized.  In this case, every extreme data-point matters.   In practice, however, the largest 
 few order statistics may be {\em corrupted}. This may lead to a severe bias in the estimation of $\xi$ (see, Tables \ref{tab:ARE_exp_0} and \ref{tab:ARE_scl_0}, below).
 In fact, the computed estimate of $\xi$ may be entirely based on these corrupted observations. In such contexts, it is important to have a robust 
 estimator of $\xi$, which does not necessarily use the most extreme order statistics, perhaps 
 puts less weight on them, or indicates to what extent the most extreme data can be trusted to come from the same distribution. 
 
 At first sight, this appears to be an ill-posed problem. Since the tail index $\xi$ is an asymptotic quantity, one {\em has to focus} on the largest 
 order statistics and if these statistics are corrupted, then there little or no information left to estimate $\xi$.  Nevertheless, using the joint 
 asymptotic behavior of the extreme order statistics, one can detect statistically significant anomalies in the most extreme order statics. 
  
The problem of robust estimation of the tail index has already received some attention (see for example \cite{guillou1}, \cite{guillou2}, \cite{Knight_asimple}, 
\cite{MR1856199}, \cite{Peng2001}, \cite{Brzezinski2016}). However, there are still open questions on the optimality and adaptivity of robust estimators to the potentially 
unknown proportion of extreme outliers. In this paper, we address these two issues. 

Recall the classic {\em Hill estimator}
\begin{equation}
\label{e:hill}
\widehat \xi_k(n):= \frac{1}{k} \sum_{i=1}^{k} \log \Bigg(\frac{X_{(n-i+1,n)}}{X_{(n-k,n)}} \Bigg),
\end{equation}
where $1\leq k\leq n-1$ and $X_{(n,n)} \geq X_{(n-1,n)}\geq  \cdots \geq  X_{(1,n)}$ are the order statistics of the sample $X_i,\ i=1,\cdots,n$. In Section \ref{sec:trim-hill}, we introduce the {\em trimmed Hill estimator}:
\begin{equation}\label{e:xi-trimmed}
\widehat{\xi}^{\rm trim}_{k_0,k}(n):= \sum_{i=k_0+1}^{k} c_{k_0,k}(i) \log \Bigg(\frac{X_{(n-i+1,n)}}{X_{(n-k,n)}} \Bigg),\hspace{5mm}  0\le k_0<k<n.
\end{equation}
Under the Pareto model \eqref{e:Pareto-model}, we obtain the {\em optimal} weights, $c_{k_0,k}(i)$ such that $\widehat{\xi}^{\rm trim}_{k_0,k}$ is the best linear unbiased estimator for 
$\xi$ (see Proposition \ref{prop:xi-opt}, below).  

Although the idea trimming has been considered before by Brazauskas and Serfling \cite{MR1856199}, and most recently by 
Zou {\em et al} \cite{zou:davis:samorodnitsky:2017}, the optimal trimmed Hill estimator has not been derived before.  These two works use equal weights in \eqref{e:xi-trimmed}, thereby producing either suboptimal or biased estimators respectively. Inference for the {\em truncated} Pareto model has been developed in the seminal work of Aban {\em et al} \cite{aban} and recently by Beirlant {\em et al}
\cite{bier_truncated}.  This should be distinguished from the approach of {\em trimming} the data in order to achieve robustness, which is the main focus of our work.  

Note that the trimmed estimators in \eqref{e:xi-trimmed} {\em do not depend} on the top $k_0$ order statistics. Therefore, they have a {\em strong upper break-down point} (see Definition 
\ref{def:ubp-est}). In the ideal Pareto setting, it turns out that our trimmed Hill estimator is essentially finite--sample optimal among the class of all unbiased estimators of 
$\xi$ with a fixed {\em strong upper break-down point}  (see Theorem \ref{thm:umvue-p}). In Section \ref{sec:heavy-tail}, we establish the asymptotic normality of the trimmed Hill estimator
in the semi parametric regime \eqref{e:heavy-tail}, under second order conditions on the regularly varying function $\ell$ as in Beirlant {\em et al} \cite{bier}. 
The rate of convergence of these estimators is the same as  that of the classic Hill as long as $k_0=o(k)$ (see Theorem \ref{prop:E-conv}).  The minimax
rate--optimality of the trimmed Hill estimators is established in Section \ref{sec:optimal}.

These theoretical results though encouraging, are not practically useful unless one has a data-adaptive method for the choice of the trimming parameter $k_0$. This problem is addressed in Section
 \ref{sec:aut-trim}. There, we start by introducing \textit{trimmed Hill plot} which can be used to visually determine $k_0$. Then, by exploiting the elegant joint distribution structure of the optimal trimmed Hill estimators, we devise a weighted sequential testing method for the identification of $k_0$. The devised sequential testing can be shown to be asymptotically consistent for the general heavy tailed regime (see
 \cite{bhattacharya:kallitsis:stoev:2017_arxiv}). This leads to a new {\em adaptive trimmed Hill} estimator, which works well even if the degree of contamination in the top order statistics is largely 
 unknown.  This novel adaptive robustness property is not present in the existing robust estimators. 
 
 In Section \ref{sec:comp}, we demonstrate the need for adaptive robustness and the advantages of our estimator in comparison with
 established robust estimators in the literature. The finite--sample performance of the trimmed Hill estimator is studied in the context of various heavy tailed models, tail indices, and contamination scenarios
 in Section \ref{sec:simulate}. We also propose a unified approach which can jointly estimate $k_0$ along with $k$ so that the method is more suited to practical applications. In Section \ref{sec:discussion}, we finally summarize our contributions and outline some future problems and practical challenges.

\section{The Trimmed Hill Estimator}
\label{sec:trim-hill}

In this section, we shall focus on the fundamental ${\rm Pareto}(\sigma,\xi)$ model and assume that
\begin{equation}\label{e:Pareto-model}
\P(X>x) = (x/\sigma)^{-1/\xi},\ x\ge \sigma,
\end{equation}
for some $\sigma>0$ and a tail index $\xi>0$.

Motivated by the goal to provide a robust estimate of the tail index $\xi$ and in view of the classical Hill estimator in  \eqref{e:hill}, we consider the class of statistics, 
$\widehat{\xi}^{\rm trim}_{k_0,k}(n)$ defined in \eqref{e:xi-trimmed}. Proposition \ref{prop:xi-opt} below finds the weights, $c_{k_0,k}(i)$ for which the estimator in
\eqref{e:xi-trimmed} is unbiased for $\xi$ and also has the minimum variance. Their optimality and robustness are discussed in Section \ref{sec:optimal}.

The following result gives the form of the best linear unbiased trimmed Hill estimator.  Its proof is given in Section \ref{sec:trim-hill}.
\vspace{2mm}
\begin{proposition} 
	\label{prop:xi-opt}
	Suppose $X_1, \cdots, X_n$ are i.i.d. observations from the distribution ${\rm Pareto}(\sigma,\xi)$ as in \eqref{e:Pareto-model}. Then among the general class of estimators  given by \eqref{e:xi-trimmed}, the minimum variance linear unbiased estimator of $\xi$ is given by
	\begin{equation}
	\label{e:xi-opt}
	\widehat{\xi}_{k_0,k} (n)=\frac{k_0+1}{k-k_0} \log \Bigg(\frac{X_{(n-k_0,n)}}{X_{(n-k,n)}} \Bigg)+\frac{1}{k-k_0} \sum_{i=k_0+2}^{k} \log \Bigg(\frac{X_{(n-i+1,n)}}{X_{(n-k,n)}} \Bigg),  \hspace{5mm} 0\leq k_0<k<n.
	\end{equation}
\end{proposition}

The choice of the trimming parameter $k_0$ is of key importance in practice. In Section \ref{sec:aut-trim}, we propose an automatic data driven methodology 
for the selection of $k_0$, which is motivated by the following result.

\begin{proposition}
	\label{prop:xi-exp}
	The joint distribution of $\widehat{\xi}_{k_0,k}(n)$ can be expressed in terms of gamma distributed random variables;
	\begin{equation}
	\label{e:xi-jt-big}
	{\Big\{\widehat{\xi}_{k_0,k}(n),\ k_0=0,\ldots,k-1 \Big\}}
	\stackrel{d}{=} {\Big\{\xi \frac{\Gamma_{k-k_0}}{k-k_0},\ k_0=0,\ldots,k-1 \Big\}},
	\end{equation}
	where the $\Gamma_i$'s are as in \eqref{e:gam-dist}. Consequently, we have that
	\begin{equation}
	\label{e:covar}\\
	{\rm Cov}(\widehat{\xi}_{i,k}(n),\widehat{\xi}_{j,k}(n))=\frac{\xi^2}{k-i\wedge j}, \hspace{5mm} i,j=0, 1, \cdots k-1
	\end{equation} where $\wedge$ denotes the min operator. Moreover, as $k-k_0\rightarrow \infty$,
	\begin{equation}
	\label{e:xi-normal}\\
	\sqrt{k-k_0}(\widehat{\xi}_{k_0,k}(n)-\xi)\stackrel{d}{\implies}N(0,\xi^2)
	\end{equation}
\end{proposition}

\noindent The proof is given in Section \ref{sec:trim-hill}.

\section{Optimality and Asymptotic Properties}
\label{sec:optimal}


\subsection{Optimality in the ideal Pareto case}
\label{sec:pareto}

The trimmed Hill estimators in \eqref{e:xi-trimmed} possess a strict upper breakdown point in the following sense.

\begin{definition}\label{def:ubp-est}
	A statistic $\widehat{\theta}$ is said to have a strict upper breakdown point $\beta$, $0\leq\beta<1$, if $\widehat{\theta}=T(X_{(n-[n\beta],n)},\cdots,X_{(1,n)})$ where $X_{(n,n)}\geq\cdots\geq X_{(1,n)}$ are the order statistics of the sample. That is, $\widehat{\theta}$ is unaffected by the values of the top $[n\beta]$ order statistics.
\end{definition}

Assuming that all observations are generated from ${\rm Pareto}(\sigma,\xi)$, the following theorem describes the optimality properties of the trimmed Hill estimator
for both the asymptotic and finite sample regimes for a given value of strict upper break down point.

\begin{theorem}
	\label{thm:umvue-p}
	Consider the class of statistics given by 
	\begin{equation*}
	{\cal{U}}_{k_0}=\left\{T=T(X_{(n-k_0,n)},\cdots,X_{(1,n)}): \: \mathbb{E}(T)=\xi, \: X_1, \cdots, X_n \stackrel{i.i.d.}{\sim} {\rm Pareto}(\sigma,{\xi})\right\}
	\end{equation*}
	which are all unbiased estimators of $\xi$ with strong upper breakdown point $\beta=k_0/n$. Then for $\widehat{\xi}_{k_0,n-1}(n)$ as in \eqref{e:xi-opt}, we have
	\begin{equation}
	\label{e:opt-var}
	\frac{\xi^2}{n-k_0} \leq  \inf_{T\in	{\cal{U}}_{k_0}} Var(T) \leq Var(\widehat{\xi}_{k_0,n-1})=\frac{\xi^2}{n-k_0-1}.
	\end{equation}
	In particular, $\widehat{\xi}_{k_0,n-1}$ is asymptotically  minimum variance unbiased estimator (MVUE) of $\xi$ among the class of estimators described by ${\cal{U}}_{k_0}$.
\end{theorem}
The proof is given in Section \ref{sec:proofs-sec3}.

\subsection{Asymptotic normality}
\label{sec:heavy-tail}

Here, we shall establish the asymptotic normality of $\tail \xi {k_0} k$ under the general semi-parametric regime \eqref{e:heavy-tail}. We shall also briefly discuss the minimax rate optimality of the
trimmed Hill estimator.

Following \cite{bier}, consider the tail quantile function
\begin{equation}\label{e:tail-qnt}Q(t)=\inf\{x:F(x)\geq 1-1/t\}=F^{-1}(1-1/t), \:\: t> 1\end{equation}
where $F^{-1}$ is the generalized inverse of the distribution function $F$. As in \cite{bier}, we assume
\begin{equation}
\label{e:L-def}
Q(t)=t^{\xi}L(t)
\end{equation}
where $L$ is a slowly varying function at $\infty$, which is equivalent to \eqref{e:heavy-tail} (see, e.g., p.\ 29 in \cite{bingham1989regular}).


Observe that 
$$X_i=Q(Y_i), \quad i=1,\cdots, n,$$
where $Y_i$, $i=1,\cdots,n$ are i.i.d ${\rm Pareto}(1,1)$. Thus in view of \eqref{e:xi-opt} and \eqref{e:L-def}, straightforward algebra yields: 
\begin{eqnarray}
\label{e:tail-3}
\tail \xi {k_0} k(n)&=&\frac{k_0+1}{k-k_0}\log\left( \frac{Y_{(n-k_0,n)}^\xi}{Y_{(n-k,n)}^\xi}\right)+\frac{1}{k-k_0} \sum_{i=k_0+2}^{k}\log\left( \frac{Y_{(n-i+1,n)}^\xi}{Y_{(n-k,n)}^\xi}\right)+R_{k_0,k}(n)\\\nonumber
&=:&\widehat{\xi}^{*}_{{k_0},k}(n)+R_{k_0,k}(n),
\end{eqnarray}
where $Y_{(i,n)}$'s are the order statistics for the $Y_i$'s and where
the remainder $R_{k_0,k}(n)$ is:
\begin{equation}
\label{e:r-def}
R_{k_0,k}(n)=\frac{1}{k-k_0}\Big((k_0+1)\log \frac{L(Y_{(n-k_0,n)})}{L(Y_{(n-k,n)})} +\sum_{i=k_0+2}^{k}\log \frac{L(Y_{(n-i+1,n})}{L(Y_{(n-k,n)})}\Big).
\end{equation}

Observe that, the $X_i^*:= Y_i^{\xi}$'s follow ${\rm Pareto}(1,\xi)$ and thus the  statistic $\widehat{\xi}^{*}_{{k_0},k}(n)$ in \eqref{e:tail-3} is nothing but the trimmed Hill estimator in the ideal Pareto data
$X_i^*,\ i=1,\dots,n$. We shall show that under suitable assumptions on the function $L$,  $\sqrt{k-k_0}R_{k_0,k}(n)$ converges to a constant in probability. This in view of \eqref{e:tail-3}, naturally leads to an asymptotic normality result for $\widehat{\xi}_{k_0,k}(n)$ (see \eqref{e:xi-normal}).

To this end, following \cite{bier}, we adopt the second order condition:
\begin{equation} \label{e:SR2} \forall x>1: \frac{L(tx)}{L(t)}=1+cg(t) \int_{1}^x  \nu^{-\rho-1}d \nu +o(g(t)), \hspace{5mm} t \rightarrow \infty \end{equation}
such that $g:(0,\infty) \rightarrow (0,\infty)$ is a $-\rho$ varying function with $\rho \geq 0$. It can be shown that \eqref{e:SR2} implies 
\begin{equation}\label{e:L-behav}
\sup_{t\ge t_\varepsilon}\Big| \log \frac{L(tx)}{L(t)}-cg(t)\int_{1}^x  \nu^{-\rho-1}d \nu  \Big| 
\leq \Bigg\{ 
\begin{array}{ll}
\varepsilon g(t) & \mbox{ if }\rho>0\\
\varepsilon g(t) x^\varepsilon  &  \mbox{ if }\rho=0.
\end{array}
\end{equation}
for all $\varepsilon>0$ and some $t_\varepsilon$ dependent on $\varepsilon$ and $g$ (see Lemma A.2 in  \cite{bier} for more details.)


\vspace{2mm}

\begin{theorem}
	\label{prop:E-conv}
	Suppose \eqref{e:SR2} holds and let $k\rightarrow\infty$, $n\rightarrow \infty$ and $k/n \rightarrow 0$ be such that	for some $\delta>0$,
	\begin{equation}\label{e:A-def}k^\delta g(n/k) \rightarrow A \end{equation}for a constant $A$. Then,
	\begin{equation}
	\label{e:E-conv}
	k^\delta\max_{0\leq k_0 <h(k)}\Bigg|\widehat{\xi}_{k_0,k}(n)-\widehat{\xi}^*_{k_0,k}(n)-\frac{cA k^{-\delta}}{1+\rho}\Bigg|\stackrel{P}{\longrightarrow}0,
	\end{equation}
	where $h(k)=o(k)$ and $\widehat{\xi}_{k_0,k}(n)$ and $\widehat{\xi}^*_{k_0,k}(n)$ are defined in \eqref{e:tail-3}.
\end{theorem}
\noindent The proof is given in Section \ref{sec:proofs-sec3}.

\begin{corollary}
	If $k_0=o(k)$ and $\sqrt{k}g(n/k) \rightarrow A$,
	$$\sqrt{k}(\widehat{\xi}_{k_0,k}(n)-\xi)\stackrel{d}{\implies}N\left(\frac{cA}{1+\rho}, \xi^2\right)$$
\end{corollary}
\noindent The proof is a direct consequence of Theorem \ref{prop:E-conv} for $\delta=1/2$ and result \eqref{e:xi-normal}.

\subsection{On the minimax rate--optimality}\label{sec:minmax}

We end this section with a brief discussion of the rate-optimality of the trimmed Hill estimators
in the context of the Hall class. Namely, consider the class of distributions ${\cal D}: = {\cal D}_\xi(B,\rho)$
with tail index $\xi>0$, such that \eqref{e:L-def} holds, where
\begin{equation}
\label{e:D-new-def}
L(x) = 1+r(x),\ \quad\mbox{ with  }\quad |r(x)|\le B x^{-\rho},\ (x>0)
\end{equation}
for some {\em fixed} constants $B>0$ and $\rho>0$ (see also (2.7) in \cite{boucheron:thomas:2015}).

\begin{theorem}[uniform consistency] \label{t:uniform-consistency} Suppose that $k=k(n) \propto n^{2\rho/(2\rho+1)}$ and $h(k) = o(k)$, as $n\to\infty$.
	
	Then, for every sequence $a(n)\downarrow 0$, such that  $a(n) \sqrt{k(n)} \to \infty$, we have
	\begin{equation}\label{e:t:uniform-consistency}
	\liminf_{n\to\infty} \inf_{F\in {\cal D}_\xi(B,\rho)} \P_F\left( \max_{0\le k_0< h(k)} |\widehat \xi_{k_0,k}(n) - \xi| \le a(n) \right) =1.
	\end{equation}
	where by $\P_F$, we understand that $\widehat \xi_{k_0,k}(n)$ was built using
	independent realizations from $F$.
\end{theorem}

The proof of this result is given in Section \ref{sec:supple}.
Relation \eqref{e:t:uniform-consistency} reads as follows. The estimator $\widehat \xi_{k_0,k}(n)$ is {\em uniformly consistent} 
(at the rate $a(n)$) in {\em both} the family of possible distributions ${\cal D}$ and in the choice of the trimming parameter $k_0$, 
so long as $k_0 = o(k)$.  This remarkable property shows that $\widehat \xi_{k_0,k}(n)$ are minimax rate-optimal 
in the sense of Hall and Welsh \cite{HallWelsh}. Indeed, Theorem 1 in Hall and Welsh implies the following.

\begin{theorem}[rate optimality]
	Let $\widehat \xi_n$ be any estimator of $\xi$ based on an independent sample from a distribution $F\in {\cal D}_\xi(B,\rho)$. 
	If we have
	\begin{equation}
	\label{e:opt-rate}
	\liminf_{n \rightarrow \infty} \inf_{F \in {\cal{D}_\xi(B,\rho)}} \P_F(|\widehat \xi_n-\xi |\leq a(n))=1
	\end{equation}
	then $n^{\rho/(2\rho+1)}a(n)=\infty$. 
\end{theorem}

This result shows that no estimator can be uniformly consistent over the Hall class of distributions ${\cal D}$ at a rate better than
$n^{\rho/(2\rho+1)}$. This is the {\em minimax} optimal rate that one could possibly hope to achieve.    
Observe that this result applies also to the trimmed Hill estimators.  As seen in Theorem \ref{t:uniform-consistency} above 
the {\em trimmed Hill} estimators attain this minimax optimal rates uniformly in $k_0 \in [0,h(k)]$, for any $h(k) = o(n^{2\rho/(2\rho+1)})$.  

\section{Data driven parameter selection}
\label{sec:aut-trim}

\subsection{Choice of $k_0$}

Suppose $X_i$, $i=1,2,\cdots,n$ are generated from the distribution $F$ of the form \eqref{e:heavy-tail}, then the optimal trimmed Hill statistic, $\widehat{\xi}_{k_0,k}(n)$ is asymptotically an unbiased estimator for the tail index $\xi$ (see Theorem \ref{prop:E-conv}) as long as the parameters $k_0$ and $k$ satisfy \eqref{e:E-conv}. However, this result breaks down in the presence of outliers, i.e. $\widehat{\xi}_{k_0,k}(n)$ may be biased estimate of $\xi$  for some $1\leq k_0 \leq k-1$. The intuition to this end is illustrated via trimmed Hill plots explained below.

For a fixed value of $k$, trimmed Hill plot is a plot of the values of $\widehat{\xi}_{k_0,k}(n)$ for varying values of $k_0$ (see Figure \ref{fig:knee1}). The vertical lines correspond to $\widehat{\xi}_{k_0,k}(n)\underline{+}\widehat{\sigma}_{k_0,k}(n)$ where $\widehat{\sigma}_{k_0,k}(n)=\widehat{\xi}_{k_0,k}(n)/\sqrt{k-k_0}$ denotes the plug in estimate of the standard error of $\widehat{\xi}_{k_0,k}(n)$ (see Proposition \ref{prop:xi-exp}). In the presence of outliers, a change-point  in the form of a knee occurs in the values of $\widehat{\xi}_{k_0,k}(n)$, when $k_0$ is close to true number of outliers, $k_0^*$. In order to obtain a robust estimate of the tail index $\xi$, it is essential to obtain an adaptive estimate of the $k_0^*$. This can be achieved by estimating the location of the knee, which serves as close approximation to the true number of outliers $k^*_0$. The plug in statistic, $\widehat{\xi}_{\widehat{k}_0,k}(n)$ based on the so-obtained $\widehat{k}_0$ serves as a robust estimate of the tail index, $\xi$. 


\begin{figure}[H]
	\includegraphics[width=0.5\textwidth]{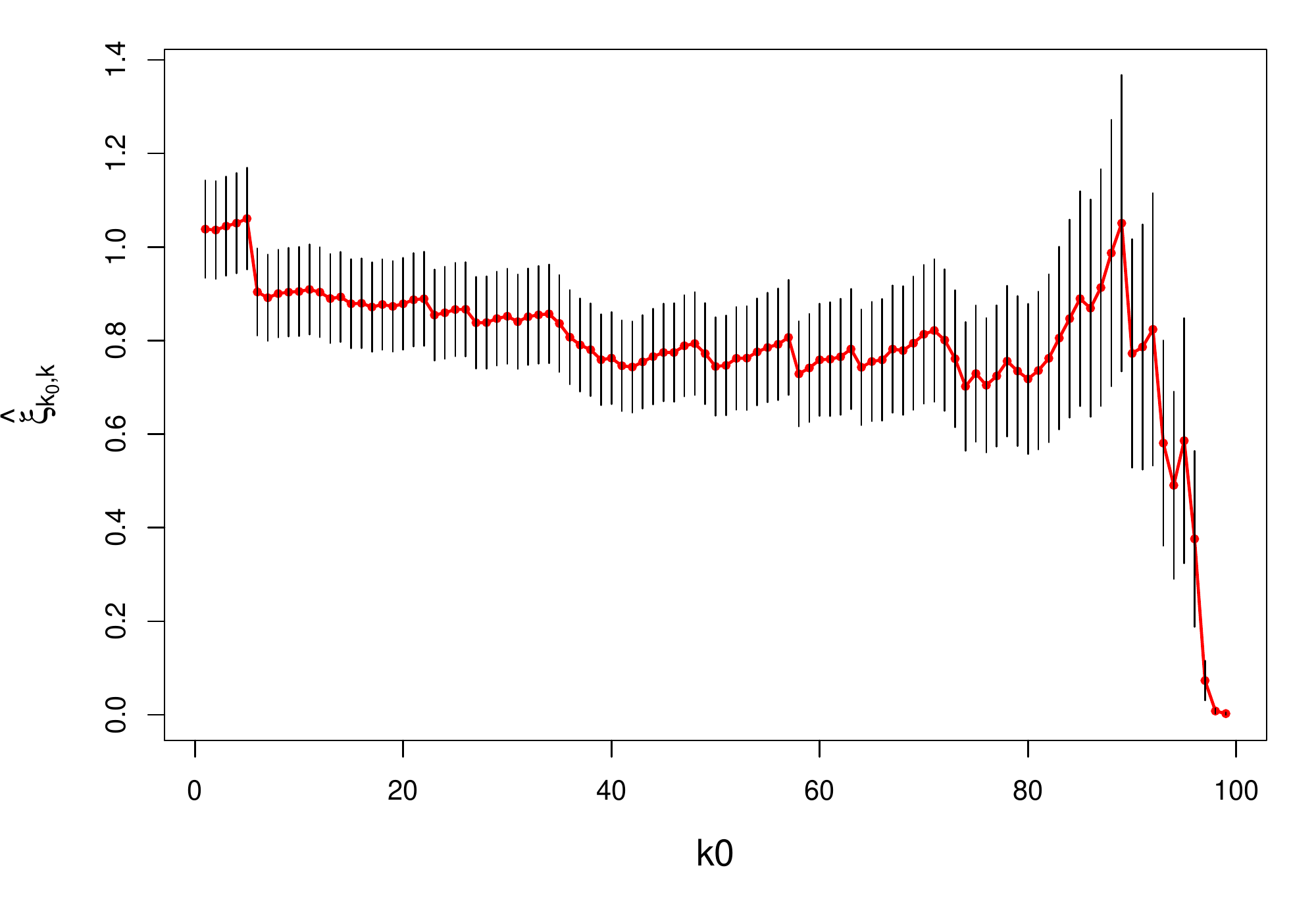}
	\includegraphics[width=0.5\textwidth]{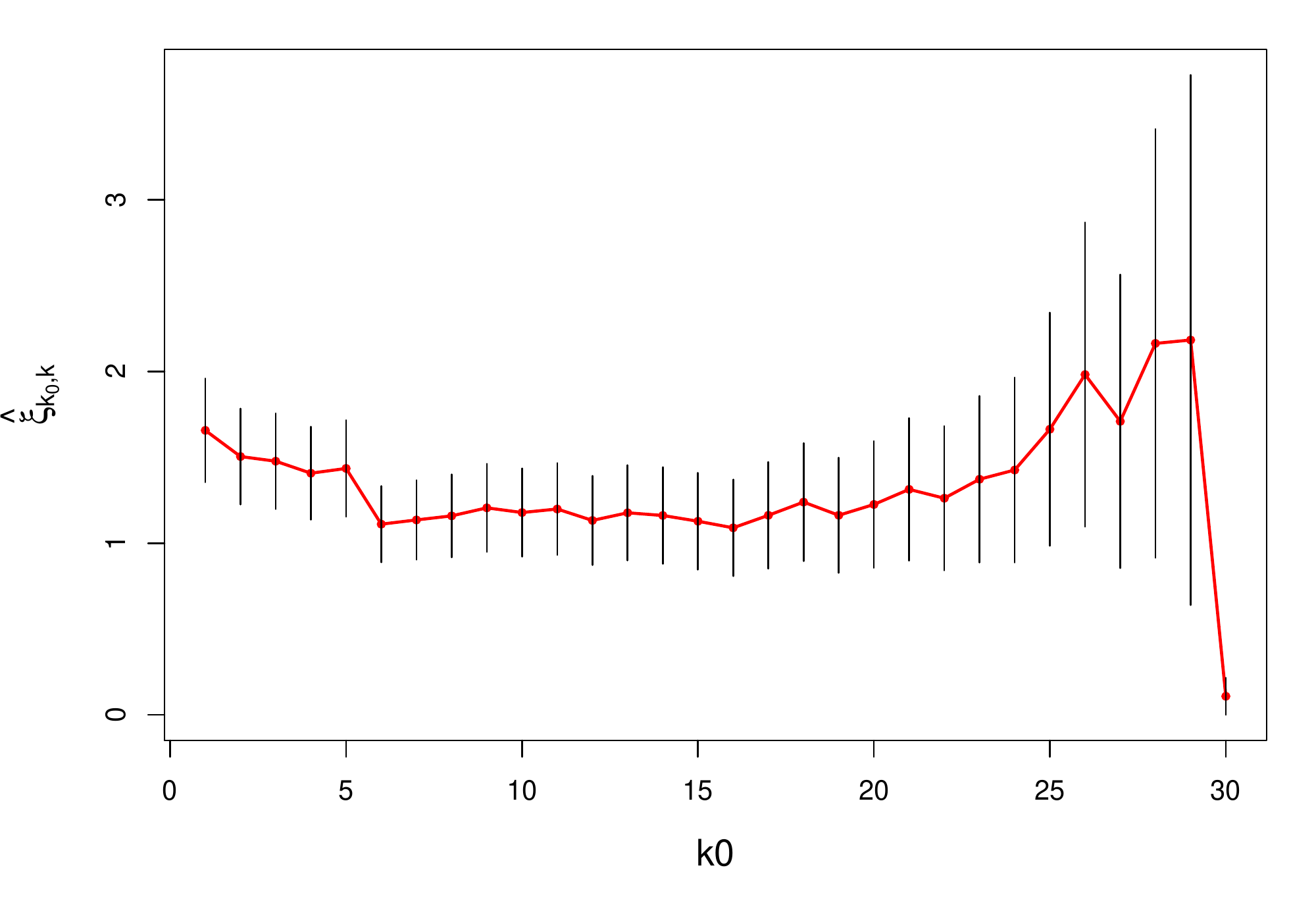}
	\caption{Trimmed Hill Plot for 10 outliers and sample size 100. Left: Pareto(1,1) with $k=99$ Burr(1,0.5,1) with $k=24$ (see \eqref{e:heavy-dist}).}
	\label{fig:knee1}
\end{figure}

In order to obtain an accurate estimate for $\xi$, it is an important task to get an estimate of the parameters $k_0$ and $k$. In the first section, we describe the methodology for the estimation of $k_0$ when $k$ is fixed.  Next we describe an iterative algorithm which allows for the estimation of the parameters $k_0$ and $k$ simultaneously.

\vspace{5mm}
\begin{proposition}
	\label{T-def} 
	Suppose all the $X_i$'s are generated from ${\rm Pareto}(\sigma, \xi)$, then consider the following class of statistics
	\begin{equation}
	\label{e:T-i-k}
	T_{k_0,k}(n):=\frac{(k-k_0-1)\widehat{\xi}_{k_0+1,k}(n)}{(k-k_0)\widehat{\xi}_{k_0,k}(n)}, \hspace{5mm} k_0=0,1,\cdots, k-2.
	\end{equation} the $T_{k_0,k}(n)$'s are independent and follow  ${\rm Beta}(k-k_0-1,1)$ distribution for $k_0=0,1,\cdots,k-2$. 
	
\end{proposition}
\begin{proof}
	By \eqref{e:T-i-k} and Proposition \ref{prop:xi-exp}, we have
	\begin{equation}
	\label{e:T-joint}
	\Big(T_{0,k},\cdots,T_{k-2,k}\Big)\stackrel{d}{=}\Big(\frac{\Gamma_{k-1}}{\Gamma_k},\cdots,\frac{\Gamma_1}{\Gamma_2}\Big),
	\end{equation}
	which implies
	$$T_{k_0,k}\stackrel{d}{=}\frac{\Gamma_{k-k_0-1}}{\Gamma_{k-k_0}}\sim{\rm Beta}(k-k_0-1,1),\hspace{5mm}i=0,\cdots,k-2.$$ 
	To show the independence of the $T_{k_0,k}$'s, from Relation \eqref{e:gam-ind} in Lemma \ref{lem:u-ind} observe that $\Gamma_m$ and  $\{\Gamma_i/\Gamma_{m}, i=1,\cdots,m\}$  are independent for all $1\leq m \leq k-2$. This in turn implies that $$\Big(\frac{\Gamma_1}{\Gamma_{2}},\frac{\Gamma_2}{\Gamma_{3}}, \cdots,\frac{\Gamma_{m-1}}{\Gamma_m}\Big)\textmd{ and } \Gamma_m\textmd{ are independent}.$$
	Since $\Gamma_i$, $i=1, \cdots, m$ and $(E_{m+1}, \cdots, E_k)$ are independent, for all $m=1, \cdots, k-2$, we have
	\begin{equation}
	\label{e:T-i-k-joint}
	\Big(\frac{\Gamma_1}{\Gamma_2},\cdots,\frac{\Gamma_{m-1}}{\Gamma_m}\Big)\textmd{ and } (\Gamma_m, E_{m+1}, \cdots, E_k)\textmd{ are independent }.
	\end{equation}
	The independence of the $T_{k_0,k}$'s follows from \eqref{e:T-i-k-joint} by observing that for all $1\leq m \leq k-2$,  $\Big(\frac{\Gamma_{m}}{\Gamma_{m+1}},\cdots,\frac{\Gamma_{k-1}}{\Gamma_k}\Big)$ is a function of $(\Gamma_m, E_{m+1}, \cdots, E_k)$.
\end{proof}

\begin{remark}
	\label{rem:T-dist}
	Observe that the distribution of $T_{k_0,k}(n)$ depends only on $X_{(n-k_0,n)}, \cdots, X_{(n-k,n)}$. Therefore the joint distribution of $T_{k_0,k}(n)$'s and hence that of $U_{k_0,k}(n)$'s remains unchanged as long as $$(X_{(n-k_0,n)}, \cdots, X_{(n-k,n)})\stackrel{d}{=}(Y_{(n-k_0,n)}, \cdots, Y_{(n-k,n)})$$ where $Y_{(n,n)}>\cdots>Y_{(1,n)}$ are the order statistics for a sample of $n$ i.i.d. observations from ${\rm Pareto}(\sigma, \xi)$. In other words, Proposition \ref{T-def} goes through for all $k_0\geq k^*_0$ provided that the top $k^*_0$ outliers do not perturb the nature of the order statistics $X_{(n-k_0+1,n)}$, $k_0 \geq k^*_0$. This motivates the sequential testing methodology discussed in the next section.
\end{remark}
\vspace{3mm}

\begin{theorem}
	\label{prop:T-conv}
	Suppose \eqref{prop:E-conv} in Theorem \ref{prop:E-conv} holds for some $\delta>0$. Then,
	\begin{equation}
	\label{e:T-conv}
	k^\delta\max_{0\leq k_0 <h(k)}\Bigg|T_{k_0,k}(n)-T^*_{k_0,k}(n)\Bigg|\stackrel{P}{\longrightarrow}0,
	\end{equation}
	where $T_{k_0,k}(n)$ and $T^*_{k_0,k}(n)$ are based on $\widehat{\xi}_{k_0,k}(n)$ and $\widehat{\xi}^*_{k_0,k}(n)$ respectively (see \eqref{e:tail-3} and \eqref{e:T-i-k} for explicit expressions).
\end{theorem}
The proof of this is described in Section \ref{sec:proofs-sec4}.
\begin{remark}
	\label{rem:T-heavy}
	Observe that by Theorem \ref{prop:T-conv}, the  asymptotic distribution of $T_{k_0,k}(n)$ and also that of $U_{k_0,k}(n)$	is same as described in Proposition \ref{T-def} as long as the number of outliers, $k_0=o(k)$. This allows us to use the algorithm described below (see Algorithm\ref{algo:seq}) for the estimation of $k_0$ in general heavy tailed models  \eqref{e:heavy-tail}. 
\end{remark}

\subsection{Exponentially Weighted Sequential Testing, EWST}
\label{subsec:seq}

Whereas the trimmed Hill plot provides an illustrative estimate of the number of outliers $k^*_0$, we discuss the weighted sequential testing algorithm for the estimation of $k^*_0$ in a principled manner. One strategy to estimate the true number of outliers, $k^*_0$, is to look for the presence of outliers among the set of values, $T_{k_0,k}(n)$. In this context, we define the following statistic
\begin{equation}
\label{e:U-i-k}
U_{k_0,k}(n):=2|(T_{k_0,k}(n))^{k-k_0-1}-0.5|, \hspace{5mm} k_0=0,1,\cdots, k-2.
\end{equation}
For i.i.d. observations from Pareto($\sigma,\xi$), $U_{k_0,k}(n)$ are i.i.d. $U(0,1)$ random variables (see Proposition \ref{T-def}). An estimate of $k^*_0$ is obtained by identifying the largest value of $k_0$ for which the hypothesis that $U_{k_0,k}(n)$ follows $U(0,1)$ gets rejected. 

In this direction, we begin with a large value of $k_0=f(k)$ and test the hypothesis: $U_{k_0,k}(n) \sim N(0,1)$. If rejected, we stop our search and declare $\widehat{k}_0=k_0$. Otherwise, we decrease the value of $k_0$ by 1 and proceed until the hypothesis $U_{k_0,k}(n) \sim U(0,1)$ gets rejected or $k_0=0$. The resulting value of $k_0$ then gives an estimate of $k_0^*$. The level for these tests increases exponentially with decrease in $k_0$. This is done in order to guard against large values $k_0$ close to $k$.

The methodology is formally described in the following algorithm.

\begin{algorithm}[H]
	\caption{\label{algo:seq}}
	\small
	\begin{algorithmic}[1]
		\label{algo:seq}
		\STATE Let $q \in (0,1)$ be the significance level.
		\STATE Choose a constant $a>1$ and set $c=1/\sum_{i=1}^{k-1}a^i$.
		\STATE Set $k_0=f(k)$.
		\STATE Compute $T_{k_0,k}(n)=((k-k_0-1)\widehat{\xi}_{k_0+1,k}(n)/((k-k_0)\widehat{\xi}_{k_0,k}(n))$.
		\STATE Compute $U_{k_0,k}(n)=2|(T_{k_0,k}(n))^{k-k_0-1}-0.5|$ as defined in \eqref{e:U-i-k}.
		\STATE If $\log (U_{k_0,k}(n))<ca^{k-k_0-1}\log (1-q)$, set $k_0=k_0-1$ else goto step 6.
		\STATE If $k_0\geq 0$, goto step 3 else $k_0=k_0+1$.
		\STATE Return $\widehat{k}_0=k_0$.
	\end{algorithmic}
\end{algorithm}

\begin{proposition} 
	\label{prop:type-I-seq}
	For i.i.d. observations from Pareto($\sigma,\xi$) and $q \in (0,1)$, let $\widehat{k}_0(q)$ be the estimate of $k^*_0$ based on Algorithm \ref{algo:seq} with $f(k)=k-2$, then under the null hypothesis $H_0:k^*_0=0$, we have $P_{H_0}[\widehat{k}_0>0]=q$.
\end{proposition}

\begin{proof} The type I error for   Algorithm \ref{algo:seq}, 	$P_{H_0}[\widehat{k}_0>0]$ is given by
	\begin{eqnarray} 
	\label{e:type-I-I}
	1-P_{H_0}[\widehat{k}_0=0]
	&=&1-P_{H_0}\Big[\log(U_{0,k})<ca^{k-1}\log(1-q),\cdots,\log(U_{k-2,k})<ca\log(1-q)\Big]\\\nonumber
	&=&1-\Pi_{i=0}^{k-2}P_{H_0}\Big[U_{i,k}<(1-q)^{ca^{k-i-1}}\Big]\\\nonumber
	&=&1-\Pi_{i=0}^{k-2}(1-q)^{ca^{k-i-1}}\\\nonumber
	&=&1-(1-q)^{c\sum_{i=0}^{k-2}a^{k-i-1}}=q.
	\end{eqnarray}
	where the last equality follows since $c=\sum_{i=1}^{k-1}a^i=\sum_{i=0}^{k-2}a^{k-i-1}$.
\end{proof}

\begin{remark}
	For Pareto case we attain the the exact bound of type I error. The bound is also attained asymptotically for the general heavy tailed distribution in \eqref{e:heavy-tail} but requires additional 
	assumptions. The following theorem sheds light on the reason behind the consistency of EWST for the more general heavy tailed setup.
\end{remark}

\begin{theorem} 
	\label{prop:U-conv}
	If \eqref{e:T-conv} holds for some $1<\delta<2$, then	\begin{equation}
	\label{e:U-conv}
	k^{(\delta-1)}\max_{0\leq k_0 <h(k)}\Bigg|U_{k_0,k}(n)-U^*_{k_0,k}(n)\Bigg|\stackrel{P}{\longrightarrow}0,
	\end{equation}
	with $U_{k_0,k}(n)$ and $U^*_{k_0,k}(n)$ are defined in \eqref{e:U-i-k}. Moreover, if $f(k)=O(k^{\delta-1})$, \begin{equation}\label{e:type-I-gen}
	P_{H_0}[\widehat{k}_0>0]\stackrel{P}{\longrightarrow}q.
	\end{equation}
\end{theorem}

The proof is described in Section \ref{sec:supple}.

\section{Simulations}
\label{sec:simulate}

In this section, we evaluate the performance of the adaptive trimmed Hill estimator, $\widehat{\xi}_{\hat{k}_0,\hat{k}}(n)$, in terms of the mean squared error, MSE as
\begin{equation}
{\rm MSE}(\widehat{\xi}_{\hat{k}_0,\hat{k}}(n))=\mathbb{E}(\widehat{\xi}_{\hat{k}_0,\hat{k}}(n)-\xi)^2
\end{equation}
For comparison, we compute the asymptotic relative efficiency, ${\rm ARE}$ with respect to both the  trimmed Hill estimator, $\widehat{\xi}_{k_0,k}(n)$ and the classic Hill, $\widehat{\xi}_{k}(n)$. The formulas are given by
\begin{eqnarray}
\label{e:are}
{\rm ARE}_{\rm TRIM}&=&{\rm MSE}(\widehat{\xi}_{k_0,k}(n))/{\rm MSE}(\widehat{\xi}_{\hat{k}_0,\hat{k}}(n))\hspace{5mm}\\\nonumber
{\rm ARE}_{\rm HILL}&=&{\rm MSE}(\widehat{\xi}_{k}(n))/{\rm MSE}\widehat{\xi}_{\hat{k}_0,\hat{k}}(n))
\end{eqnarray}
respectively, where $k_0$ is the true trimming parameter, and $k$ is replaced by its optimal choice as: \begin{equation}
\label{e:k-sim}
k^*_{n,k_0}=\underset{k=k_0+1,\cdots,n-1}{\rm arg\,min}{\rm MSE}(\widehat{\xi}_{k_0,k}(n)).
\end{equation}
We first explore the performance of exponentially weighted sequential testing algorithm, EWST as described in Section \ref{subsec:seq} as an estimator of the trimming parameter $k_0$.  In Sections \ref{subsec:h0}, \ref{subsec:ha_large} and \ref{subsec:ha_small}, we replace $\hat{k}$ in \eqref{e:are} by the optimal values $k^{*}_{n,k_0}$ in \eqref{e:k-sim}.

In Section \ref{subsec:k0-k-jt}, we will address the performance of the adaptive trimmed Hill, $\widehat{\xi}_{\hat{k}_0,\hat{k}}(n)$ where $k$ is unknown and estimated from the data as described in Section \ref{subsec:k0-k-jt}.

The efficacy of the proposed algorithms have been explored in the light of the following heavy-tailed distributions.
\begin{eqnarray}
\label{e:heavy-dist}
{\rm Pareto}(\sigma,\alpha)&:& 1-F(x)=\sigma^\alpha x^{-\alpha};\:\: x>1,\alpha>0;\:\: \xi=1/\alpha\\\nonumber
{\rm Frechet}(\alpha)&:&1-F(x):1-\exp(-x^{-\alpha});\:\: x>0, \alpha>0;\:\: \xi=1/\alpha\\\nonumber
{\rm Burr} (\eta,\lambda,\tau)&:&1-F(x)=1-\left(\frac{\eta}{\eta+x^{-\tau}}\right)^{-\lambda};\:\: x>0, \eta>0, \lambda>0, \tau>0; \:\: \xi=1/\tau\\\nonumber
{\rm |{\rm T}|}(t)&:&1-F(x)=\int_{x}^{\infty}\frac{2\Gamma(\frac{t+1}{2})}{\sqrt{n\pi}\Gamma(\frac{t}{2})}\left(1+\frac{w^2}{t}\right)^{-\frac{t+1}{2}}dw; \:\: x>0,t>0;\:\: \xi=1/t
\end{eqnarray}
In Sections \ref{subsec:ha_large} and \ref{subsec:ha_small}, the number of outliers $k_0$ and the tail index $\xi$ are kept fixed. Varying values of $\xi$ and $k_0$ are studied in Section \ref{subsec:alpha-k0}.
\subsection{Performance under $H_0$ ($k_0=0$)} 
\label{subsec:h0}
In this section, we let $X_1, \cdots, X_n$ be i.i.d. generated from one of the four distributions in \eqref{e:heavy-dist}. The tail index $\xi$ is fixed at 1. We assume that there are no outliers, i.e. $k_0=0$ which in turn implies that the trimmed Hill coincides with the classic Hill estimator.

Assuming $k=\hat{k}=k^*_{n,0}$, we evaluate the performance of the adaptive trimmed Hill, $\widehat{\xi}_{\hat{k}_0,\hat{k}}(n)$ with respect to the classic Hill, $\widehat{\xi}_{k}(n)$ in terms of ARE using \eqref{e:are}. The trimming parameter estimate $\hat{k}_0$ is obtained using EWST as in Section \ref{subsec:seq}. The ${\rm ARE}$'s are based on 5000 independent Monte Carlo realizations. For EWST , the significance level, $q$ and the exponentiation parameter $a$  are fixed at 0.05 and 1.2 respectively.

\begin{table}[H]
	\centering
	\begin{tabular}{|c|c|c|c|c|c|}
		\hline
		$n$ & Pareto(1,1) & Frechet(1) & Burr(1,0.5,1) & T(1)\\\hline
		100 &99.17  & 97.19 &  86.25 &  97.22\\ 
		200 &99.53  & 99.33 &  96.64 &  99.83\\
		500 &99.85  & 99.88 &  98.27 &  99.85\\\hline
	\end{tabular}
	\caption{{\rm ARE} of the adaptive trimmed Hill with respect to the classic Hill, $k_0=0$ and $\xi=1$.}
	\label{tab:ARE_ho}
\end{table}

As seen in Table \ref{tab:ARE_ho}, apart from the Burr distribution, we have fairly large ARE values (almost 100\%) even at sample size $n=100$. This indicates that the EWST algorithm picks up the true $k_0=0$ in almost all of the cases. As the sample size grows ($n=500$), the behavior is more uniform across different distribution and we achieve nearly 100\% asymptotic relative efficiency even for the Burr case.  This may be explained by the asymptotic Pareto-like behavior of the heavy tailed distributions (see \eqref{e:heavy-tail}).

In the following section, we explore the behavior of adaptive trimmed Hill when there are non zero outliers in the data, i.e. $k_0>0$.
\subsection{Inflated outliers ($k_0>0$)}
\label{subsec:ha_large}

We simulate from one of the distributions in \eqref{e:heavy-dist} with $\xi=1$. We introduce $k_0$ outliers by perturbing the top-$k_0$ order statistics using one of the following two approaches
\begin{eqnarray}
\label{e:exp-trans-0}
X_{(n-i+1,n)}:= X_{(n-k_0,n)}+(X_{(n-i+1,n)}-X_{(n-k_0,n)}))^L, \hspace{5mm} i=1,\cdots,k_0, \hspace{5mm} L>1\\
\label{e:scl-trans-0}
X_{(n-i+1,n)}:= X_{(n-k_0,n)}+C(X_{(n-i+1,n)}-X_{(n-k_0,n)})), \hspace{5mm} i=1,\cdots,k_0, \hspace{5mm} C>1
\end{eqnarray}
For $L,C>1$, the transformations \eqref{e:exp-trans-0} and \eqref{e:scl-trans-0} lead to inflation of the top-$k_0$ order statistics while still preserving their order.

We first fix $k_0=10$ and assume that $k=\hat{k}=k^*_{n,10}$. We then obtain the trimming parameter estimate, $\hat{k}_0$ and the corresponding adaptive trimmed  Hill estimator, ${\widehat{\xi}}_{\hat{k}_0,\hat{k}}$ by using the EWST algorithm in Section \ref{subsec:seq}. The performance is evaluated in terms of the  ${\rm ARE}$ relative to the  trimmed Hill and ${\widehat{\xi}}_{k_0,k}(n)$  and the classic Hill ${\widehat{\xi}}_{k}(n)$, as in \eqref{e:are}. Tables \ref{tab:ARE_exp_0} and \ref{tab:ARE_scl_0}  show the performance of the adaptive trimmed Hill for varying values of $L$ and $C$ respectively.

\begin{table}[H]
	\centering
	\begin{tabular}{|c|cccc|cccc|cccc|}
		\hline
		n   &  & 100 &   &    &     & 200 &    &   &   & 500 &   & \\\hline
		L &  1.2 & 1.5 & 5 & 20 & 1.2 & 1.5 & 5 & 20 & 1.2 & 1.5 & 5 & 20\\\hline
		{\rm Pareto}(1,1)	& 	0.92 & 0.54 & 0.94   & 0.99 & 0.94&  0.77& 0.98 & 1.00& 0.95 & 0.94& 1.00&  1.00\\	
		& 1.03  & 2.95 & 76.5 & 1808 & 1.02 &  3.58 & 644.9 & 1608 & 1.02 & 3.21&  45.7& 1114\\\hline
		Frechet(1) & 0.82 & 0.26 & 0.69 & 0.96 & 0.74 & 0.37 & 0.89 & 0.99 & 0.71 & 0.56 & 0.97 & 1.00 \\    
		& 11.6 & 3.66 & 10.4 & 13.9 & 21.7 & 10.8 & 26.5 & 28.9 & 45.4 & 35.7 & 59.8 & 62.2\\\hline    
		Burr(1,0.5,1)	& 1.13 &	0.33 &	0.21 &	0.94 &	0.87 &	0.26 &	0.54 & 0.96 &	0.74 &	0.37 &	0.88 &	0.99\\
		& 5.19  &	1.48 &	0.88 &	4.24 &	9.47 &	2.92 &	10.3 &	19.7 & 5.87 & 	10.3 &	23.8 &	25.8\\\hline
		$|{\rm T}|$(1)	& 0.87 & 0.29 & 0.56 & 0.96 & 0.79 & 0.36 & 0.85 & 0.98 & 0.75 & 0.62 & 0.95 & 1.00\\
		& 15.6 & 5.11 & 10.0 & 17.0 & 31.7 & 14.5 & 33.3 & 38.0  & 71.6 & 58.4 & 88.6 & 94.7
		\\\hline
	\end{tabular}
	\caption{{\rm ARE} of the adaptive trimmed Hill $k_0=10$, $\xi=1$ and $L>1$. For each distribution, top row corresponds to ${\rm ARE}_{\rm TRIM}$ and bottom row indicates ${\rm ARE}_{\rm HILL}$.}
	\label{tab:ARE_exp_0}
\end{table}

\vspace{-3mm}
\begin{table}[H]
	
	\centering
	\begin{tabular}{|c|cccc|cccc|cccc|}
		\hline
		n   &  & 100 &   &    &     & 200 &    &   &   & 500 &   & \\\hline
		C &  2 & 10 & 20 & 100 & 2 & 10 & 20 & 100 & 2 & 10 & 20 & 100\\\hline
		{\rm Pareto}(1,1)	& 	0.93 & 0.57 & 0.64   & 0.95 & 0.95&  0.73& 0.80 & 0.98& 0.98 & 0.87& 0.93&  0.99\\	
		& 1.01  & 1.85 & 3.33 & 12.7 & 1.01 &  1.57 & 2.57 & 7.22 & 1.00 & 1.29&  1.79& 3.52\\\hline
		Frechet(1) & 0.84 & 0.29 & 0.37 & 0.70 & 0.81 & 0.39 & 0.45 & 0.83 & 0.81 & 0.51 & 0.60 & 0.92 \\    
		& 11.8 & 3.94 & 5.44 & 10.3 & 22.7 & 11.8 & 13.9 & 23.7 & 50.2 & 31.6 & 37.7 & 59.6\\\hline    
		Burr(1,0.5,1) & 0.98 & 0.33 & 0.27 & 0.37 & 0.86 & 0.33 & 0.33 & 0.62 & 0.8 & 0.44 & 0.47 & 0.83\\
		& 4.52 & 1.44 & 1.20 & 1.58 & 9.17 & 3.58 & 3.65 & 6.82 & 21.3 & 12.0 & 12.7 & 22.3\\\hline
		$|{\rm T}|$(1) &0.80 &0.28 &0.30 &0.58 &0.77 &0.38 &0.47 &0.86 &0.83 &0.54 &0.65 &0.93\\
		&14.5 &4.97 &5.34 &10.6 &31.4 &15.2 &18.9 &31.9 & 80.5 &53.0 &61.4 &87.4\\\hline
	\end{tabular}
	\caption{{\rm ARE} of the adaptive trimmed Hill $k_0=10$, $\xi=1$ and $C>1$. For each distribution, top row corresponds to ${\rm ARE}_{\rm TRIM}$ and bottom row indicates ${\rm ARE}_{\rm HILL}$.}
	\label{tab:ARE_scl_0}
\end{table}

We first observe the ARE values compared to the oracle trimmed Hill statistic are relatively stable and improve considerably with the increase in sample size $n$. For outliers of small magnitude, i.e  $L=1.2$ and $C=2$, the ARE values are relatively higher as compared to the case of moderate outliers, i.e. $L=2,5$ or $C=10,20$. This is natural, since small values of $L$ and $C$ are indicative of lower levels of contamination and thus the estimation of $\xi$ is accurate even if $k_0$ is underestimated. For $L=2,5$ and $C=10,20$, we have for estimation accuracy for the trimming parameter $\hat{k}_0$ (observed in histograms of $\hat{k}_0$ not reported here). However, the increase in severity of outliers  produces a greater error in the estimation of $\xi$. Outliers of large magnitude, i.e.  $L=20$ and $C=100$, allow for nearly perfect detection accuracy for the trimming parameter $k_0$ and hence the ARE values close to 100\%.

The estimation of $k_0$ is best under the Pareto setting followed by Frechet and the T-distribution. Of all cases, the Burr distribution is most challenging. This is explained by the slow rate of convergence of Burr tails to Pareto tails and hence the relatively lower efficiency of the adaptive trimmed Hill. For large sample sizes $n=500$, sensitivity of the adaptive trimmed Hill to underlying distribution structure decreases and  we attain nearly 100\% accuracy uniformly across all distributions when $L>2$ and $C>20$. 

Finally, we observe the unusually large ARE values relative to the classic Hill. It is remarkable that even small perturbations in the top order statistics ($L=1.2$ and $C=2$) lead to an unacceptable bias of the classic Hill estimator. The MSE deteriorates by a factor of 14 or 15 in case of the T-distribution for $n=100$ and it could be as bad as 80 when $n=500$. This highlights the importance of considering adaptive robust estimators of $\xi$ in real data problems where the observations could be contaminated. For the remaining section, we shall thus consider the ARE values relative to the trimmed Hill only.

\subsection{Role of $\xi$ and $k_0$}
\label{subsec:alpha-k0}

In this section, we explore the influence of the tail exponent $\xi$ and the extent of contamination  $k_0$ on the EWST algorithm of Section \ref{subsec:seq}. Figures \ref{fig:alpha-rol} and \ref{fig:k0-rol} display the ARE values of the adaptive trimmed Hill,  ${\widehat{\xi}}_{\hat{k}_0,\hat{k}}(n)$ for varying values of $\xi$ and $k_0$ respectively.

\vspace{-2mm}
\begin{figure}[H]
	\centering
	\includegraphics[width=0.45\textwidth]{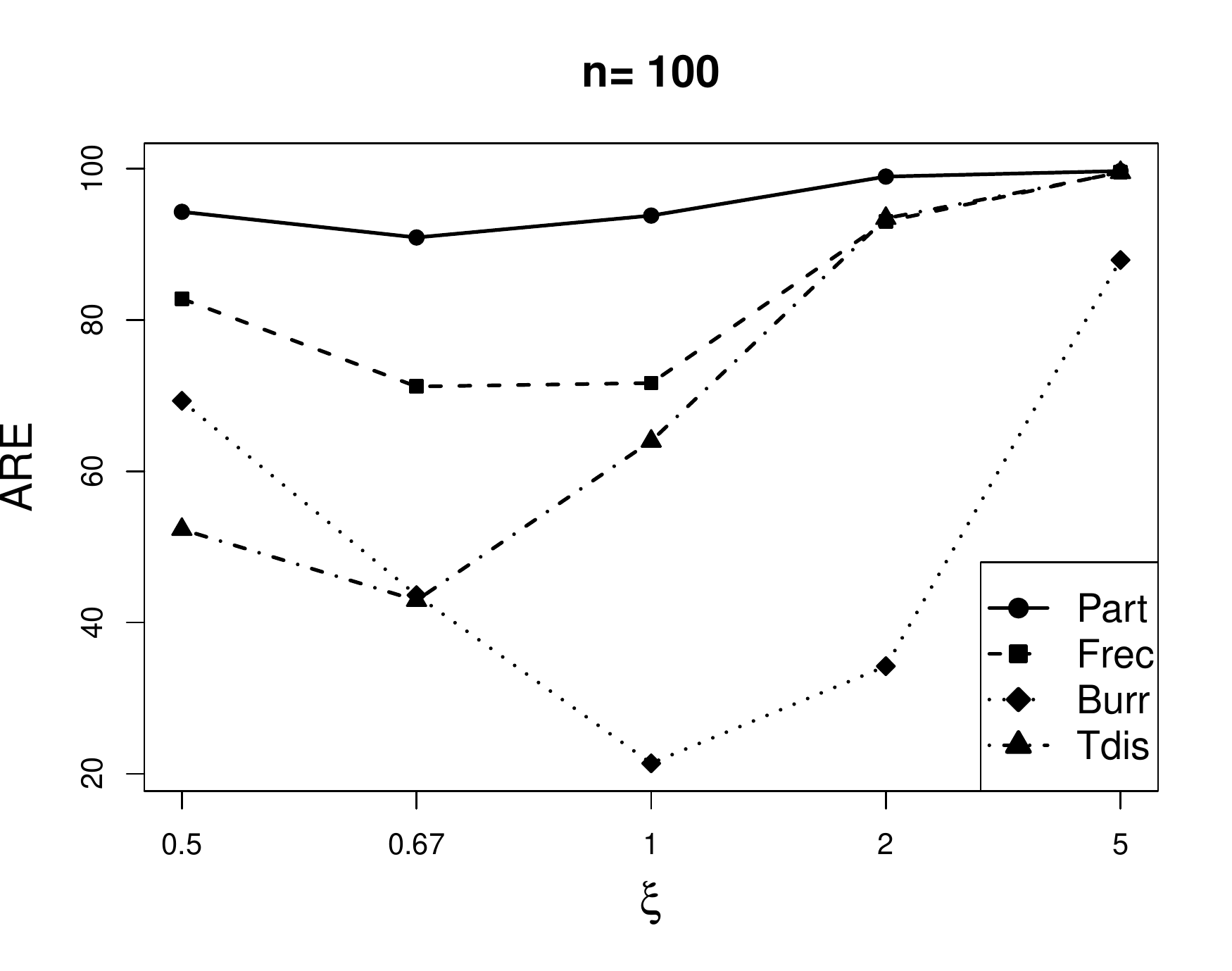}
	\includegraphics[width=0.45\textwidth]{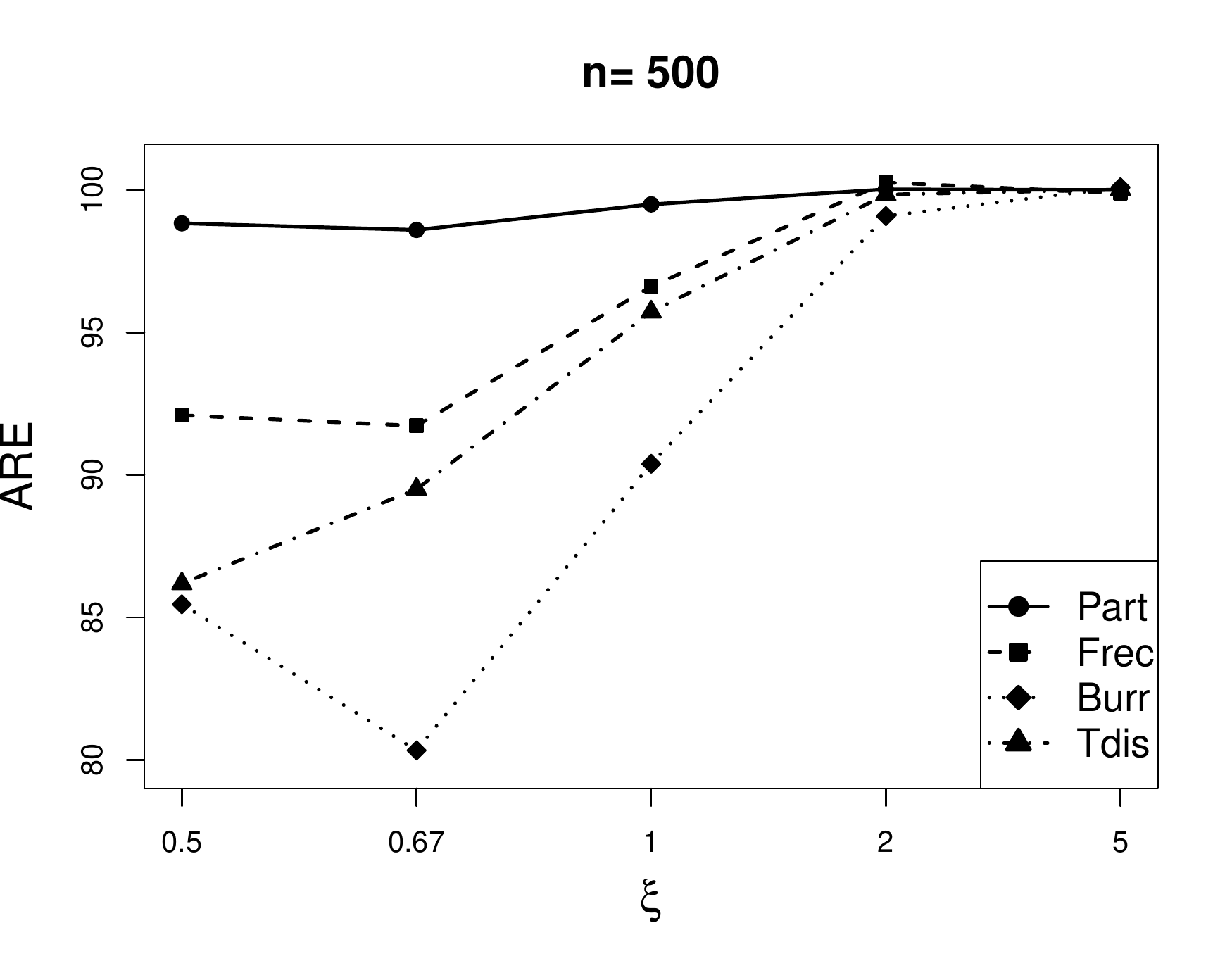}
	\caption{Performance for varying values of tail exponent $\xi$. Left and right panels correspond to $n=100$ and $n=500$, respectively.}
	\label{fig:alpha-rol}
\end{figure}

We first inject $k_0=10$ top outlier statistics as in \eqref{e:exp-trans-0} with $L=5$. The underlying distributions from which the data is generated correspond to ${\rm Pareto}(1,1/\xi)$, ${\rm Frechet}(1/\xi)$, ${\rm Burr}(1,0.5,1/\xi)$ and $|T|(1/\xi)$ with $\xi$ in the range $\{0.5,0.67,1,2,5\}$. With $k=\hat{k}=k^*_{n,10}$, we consider the ARE values of the adaptive trimmed Hill,  ${\widehat{\xi}}_{\widehat{k}_0,\hat{k}}(n)$ relative to the trimmed Hill, ${\widehat{\xi}}_{k_0,k}$ for varying values of $\xi$. Figure \ref{fig:alpha-rol} shows this behavior for varying sample sizes.

We observe that  the efficiency of the proposed adaptive trimmed Hill  approaches 100\% for $\xi>0.67$ for all distributions with increase in sample size $n$.  This is expected as the heavy tailed distributions in \eqref{e:heavy-tail} get closer to the Pareto distribution asymptotically. Whereas the Pareto distribution is more or less robust to the change in the tail exponent $\xi$, the other three heavy tailed distributions suffer from a mild loss in efficiency in the range $\xi <2$.  The superior performance at higher values of $\xi$ indicates easy identification of outliers in heavier tails. The performance of our estimator improves on both sides of $\xi=0.67$ for all distributions apart from the Burr. The Burr distribution is the most challenging in terms of identification of the trimming parameter $k_0$ and has relatively low efficiency for $0.67<\xi<2$ especially for smaller sample sizes.

\vspace{-2mm}
\begin{figure}[H]
	\centering
	\includegraphics[width=0.45\textwidth]{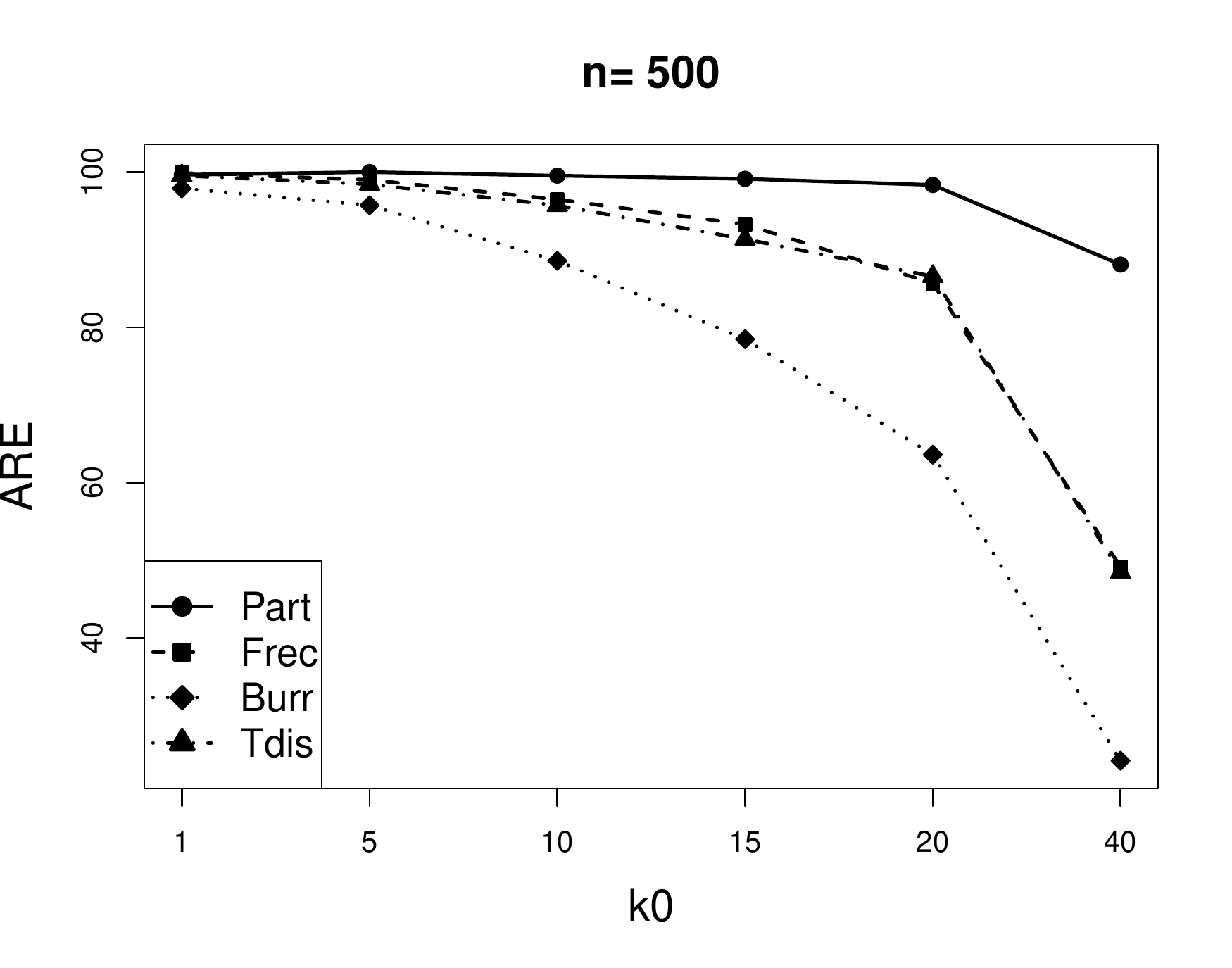}
	\includegraphics[width=0.45\textwidth]{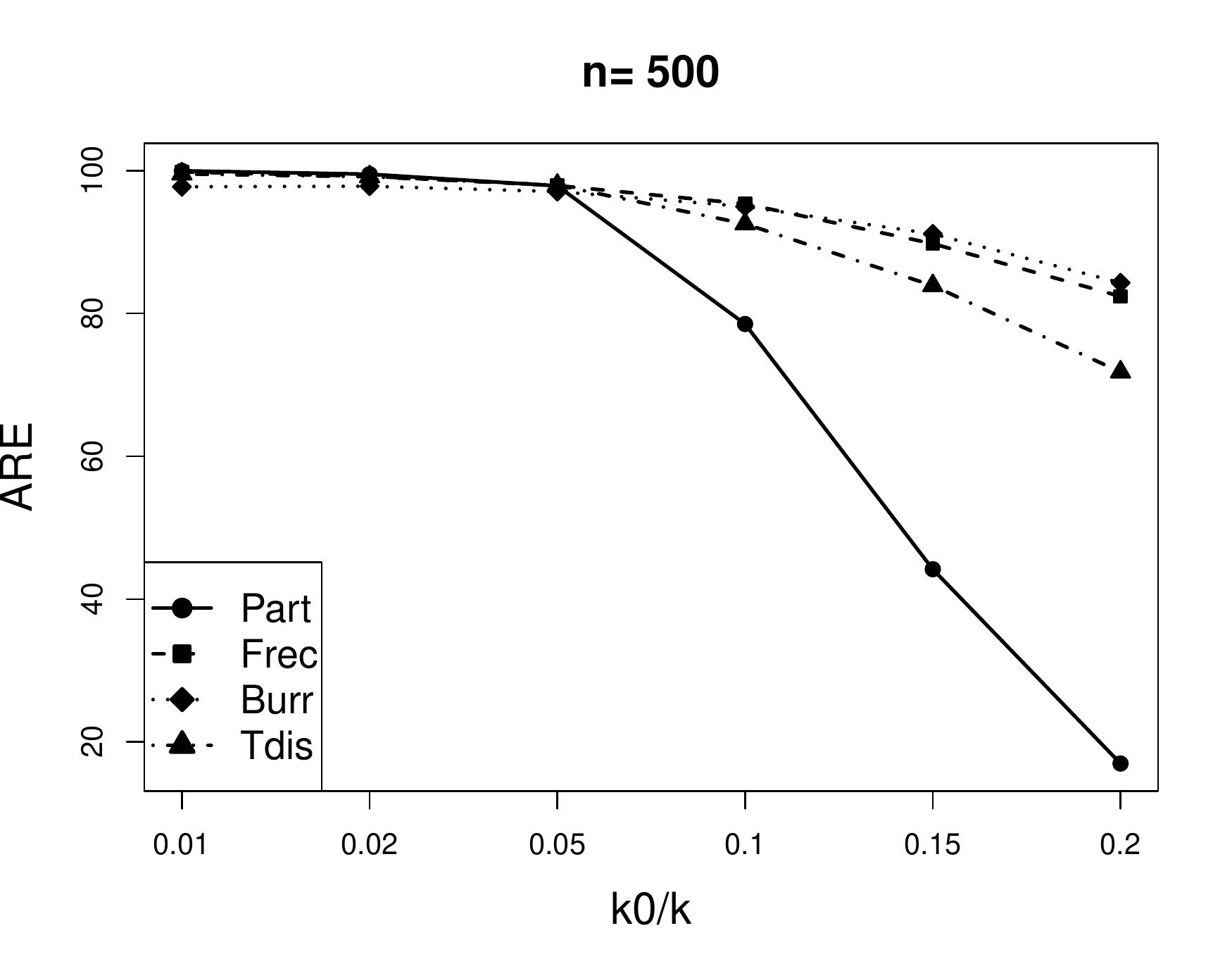}
	\caption{Performance for varying amount of outliers. Left and right panels correspond to the number $k_0$ and the proportion $k_0/k$ of outliers respectively.}
	\label{fig:k0-rol}
\end{figure}

We next inject $k_0$ top outlier statistics as in \eqref{e:exp-trans-0} with $L=5$ and distribution in \eqref{e:heavy-dist} with $\xi=1$. The underlying distributions from which the data is generated correspond to ${\rm Pareto}(1,1)$, ${\rm Frechet}(1)$, ${\rm Burr}(1,0.5,1)$ and $|T|(1)$. We consider two different scenarios, one where  $k_0 \in \{1,5,10,15,20,40\}$ and the other where $k_0/k \in \{0.01,0.02,0.05,0.1,0.2\}$. For scenario 1 (varying $k_0$), we let $k=\hat{k}=k^*_{n,k_0}$ and for scenario 2 (varying $k_0/k$), we let $k=\hat{k}=k^*_{n,0}$. We then apply the EWST algorithm for estimation of $k_0$.  Figure  \ref{fig:k0-rol} displays the ARE values of the adaptive trimmed Hill,  ${\widehat{\xi}}_{\hat{k}_0,k}(n)$ relative to the  trimmed Hill, ${\widehat{\xi}}_{\widehat{k}_0,k}$.

We observe that with increase in both number ($k_0$) and proportion of outliers ($k_0/k$), naturally the efficiency of the proposed adaptive trimmed Hill decreases. This may be attributed to the fact that the detection accuracy of  $k_0$ becomes increasingly difficult with increase in both number and proportion of outliers. From Figure \ref{fig:k0-rol}, we observe that when the number of outliers, $k_0$ is kept constant, the ARE of Pareto is the greatest while that of Burr is the least. On the other when the proportion $k_0/k$ is kept constant, the performance under Burr is the best while that under Pareto is the worst. This unusual phenomenon can be explained as follows. For same sample size, the optimal $k=k^*_{n,k_0}$ is the largest for Pareto followed by T, Frechet and Burr. Since a large effective sample size allows for better estimation of $k_0$, therefore the highest ARE values are obtained corresponding to the Pareto distribution in \ref{fig:k0-rol} left. For $k_0/k$ constant, large $k=k^*{n,k_0}$ implies large number of outliers, $k_0$. Since $k^*_{n,0}$ is smallest in case of the Burr distribution, we record largest ARE values for the Burr distribution in Figure \ref{fig:k0-rol} right.
\subsection{Deflated Outliers, $k_0>0$}
\label{subsec:ha_small}

We simulate from one of the distributions in \eqref{e:heavy-dist} with $\xi=1$. We introduce $k_0$ outliers by perturbing the top order statistics as in \eqref{e:exp-trans-0} and \eqref{e:scl-trans-0} where now $L, C<1$. For $L,C<1$, the transformations \eqref{e:exp-trans-0}  and \eqref{e:scl-trans-0} lead to the deflation of the top-$k_0$ order statistics while still preserving their order.

\begin{table}[H]
	\centering
	\begin{tabular}{|c|ccc|ccc|ccc|}
		\hline
		n     & 100 &   &        & 200 &    &   & 500 &   & \\\hline
		L &  0.005 & 0.05 & 0.5 & 0.005 & 0.05 & 0.5 & 0.005 & 0.05 & 0.5\\\hline
		Pareto(1,1)& 0.99 &0.82 &0.90 &1.00 &0.88 &0.90 &1.00 &0.96 &0.93\\
		Frechet(1) &1.35 &1.64 &2.87 &1.20 &1.54 &1.92 &1.12 &1.23 &1.68\\
		Burr(1,0.5,1)&1.59 &3.04 &5.14 &1.38 &1.97 &2.91 &1.21 &1.46 &2.06\\
		$|{\rm T}|(1)$ & 1.31 &1.59 &2.33 &1.17 &1.22 &1.59 &1.10 &1.13 &1.42\\\hline
	\end{tabular}
	\caption{{\rm ARE} of the adaptive trimmed Hill relative to  trimmed Hill for $k_0=10$, $\xi=1$ and $L<1$.}
	\label{tab:ARE_exp_1}
\end{table}
With  $k_0=10$ and $k=\hat{k}=k^*_{n,10}$, we obtain the trimming parameter estimate, $\hat{k}_0$ and the corresponding adaptive trimmed  Hill estimator, ${\widehat{\xi}}_{\hat{k}_0,\hat{k}}$ using the EWST algorithm in Section \ref{subsec:seq}. Their performance is evaluated in terms of the  ${\rm ARE}$ relative to the  trimmed Hill ${\widehat{\xi}}_{k_0,k}(n)$ in Tables  \ref{tab:ARE_exp_1} and \ref{tab:ARE_scl_1}. 

\begin{table}[H]
	\centering
	\begin{tabular}{|c|ccc|ccc|ccc|}
		\hline
		n     & 100 &   &        & 200 &    &   & 500 &   & \\\hline
		L &  0.005 & 0.05 & 0.5 & 0.005 & 0.05 & 0.5 & 0.005 & 0.05 & 0.5\\\hline
		Pareto(1,1)& 0.97 & 0.75 & 1.05 & 0.98 & 0.87 & 1.02 & 1.00 & 0.94 & 1.00 \\
		Frechet(1) &1.09 & 1.68 & 2.22 & 1.06 & 1.97 & 1.71 & 1.07 & 1.75 & 1.46\\
		Burr(1,0.5,1)&1.11 & 3.46 & 3.51 & 1.05 & 2.95 & 2.48 & 1.08 & 1.99 & 1.78 \\
		$|{\rm T}|(1)$ & 1.06 & 1.50 & 1.94 & 1.06 & 1.40 & 1.52 & 1.03 & 1.32 & 1.29\\\hline
	\end{tabular}
	\caption{{\rm ARE} of the adaptive trimmed Hill relative to trimmed Hill for $k_0=10$, $\xi=1$ and $C<1$.}
	\label{tab:ARE_scl_1}
\end{table}

For the Pareto distribution, the ARE is higher in for more severe outliers $L=0.005$, $C=0.001$ than the case of moderate outliers $L=0.05$, $C=0.1$. This is because more extreme outliers facilitate easier estimation of $k_0$ and hence the large ARE. However observe that for $L=0.5, C=0.5$, we have greater than the case of $L=0.05$, $C=0.1$. This is because values of $L$ and $C$  close to 1 under estimation of $k_0$ does not have a huge impact on the MSE of the adaptive trimmed Hill. For the distributions apart from Pareto, we obtain ARE values which are greater than 100\%. The detection accuracy of EWST in determining $k_0$ has the exact same trend as that for the Pareto case. However, for other heavy tailed distributions, a few downscaled outliers sometimes helps in improving the MSE value of the adaptive trimmed Hill. As a result, the adaptive trimmed Hill outperforms the oracle trimmed Hill benchmark based on the true value of $k_0$.

\subsection{Joint estimation of $k$ and $k_0$}
\label{subsec:k0-k-jt}
From Relation  \eqref{e:t:uniform-consistency} in Theorem \ref{t:uniform-consistency} and Corollary \ref{cor:trim-rate} , we observe that if $k_0=o(n^{2\rho/(2\rho+1)})$,  the asymptotic mean squared error (AMSE) of the trimmed Hill estimator is same as that of the classic Hill. Therefore following \cite{HallWelsh-1}, the value of $k$ which minimizes the AMSE of the trimmed Hill is
$$k_{n}^{\rm opt}\sim \left(\frac{C^2\rho (\rho+1) ^2}{2D^2\rho^3}\right)^{1/(2\rho+1)}n^{2\rho/(2\rho+1)}.$$

The finite sample equivalent of $k_n^{\rm opt}$ is given by $k^*_{n,k_0}$ as in \eqref{e:k-sim}. Drees and Kaufmann in \cite{drees} provide a methodology for the estimation of $k_n^{\rm opt}$ for the classic Hill. Motivated by their approach, we propose a method for the joint estimation of $k_0$ and $k$ under the following assumptions
\begin{eqnarray}
\label{e:cond-k}
k_0&=&o(n^{2\rho/(2\rho+1)})\nonumber\\
1-F(x)&=&Cx^{-1/\xi}(1+Dx^{-\rho/\xi}+o(x^{-\rho/\xi}))\\
F^{-1}(1-t)&=&ct^{-\xi}\exp\left(\int_{t}^1 \frac{\varepsilon(s)}{s}ds\right)\nonumber
\end{eqnarray}
The last two assumptions correspond to Eqs (2) and (5) in \cite{drees} respectively.

Suppose the trimming parameter, $k_0$ is known. We define the modified version of Eq (4) in \cite{drees} as
\begin{equation}
\label{e:kn-def}
\bar{k}_{n,k_0}(r_n)=\min \left\{k \in \{k_0+1,\cdots,n-1\}\Big |\max_{k_0+1\leq i \leq k}(i-k_0+1)^{1/2}|\widehat{\xi}_{k_0,i}(n)-\widehat{\xi}_{k_0,k}(n)|>r_n\right\}
\end{equation}
where $\hat{\xi}_{k_0,i}$ is the trimmed Hill based on $i-k_0$ observations. We conjecture a modified version of Theorem 1 in \cite{drees}, where the classic Hill estimator gets replaced by its corresponding trimmed version as follows:
\vspace{2mm}

\begin{proposition}
	\label{thm:k-est}
	Suppose $r_n=o((n-k_0)^{1/2})$, $\log \log (n-k_0)=o(r_n)$ and \eqref{e:cond-k} holds. Then
	if $\hat{\rho}_n$ is any consistent estimator of $\rho$ and $\tilde{\xi}_{n}$ is a consistent initial estimator of $\xi$, then for $\epsilon \in (0,1)$ and $(\log \log (n-k_0))^{1/2\epsilon}=o(r_n)$, we have
	\begin{equation}
	\label{e:k-est}
	\hat{k}_{n,k_0}^{\rm opt}=(2\hat{\rho}_n+1)^{-1/\hat{\rho}_n}(2\tilde{\xi}_{n}\hat{\rho}_n)^{1/(2\hat{\rho}_n+1)}\Bigg(\frac{\bar{k}_{n,k_0}(r_n^\epsilon)}{(\bar{k}_{n,k_0}(r_n))^\epsilon}\Bigg)^{1/(1-\epsilon)}
	\end{equation}
	is a consistent estimator of $k_{n,k_0}^{*}$ in the sense that $\hat{k}_{n,k_0}/k_{n,k_0}^*$ converges in probability to 1. In particular, $\widehat{\xi}_{k_0,\hat{k}_{n,k_0}^{\rm opt}}(n)$ has the same asymptotic efficiency as  $\widehat{\xi}_{k_0,k_{n,k_0}^{*}}(n)$.
\end{proposition}

The trimmed estimator, $\hat{\xi}_{k_0,2\sqrt{n}}$ can be used as an initial consistent estimator of $\xi$ for a wide range distributions from \eqref{e:heavy-tail}. As in \cite{drees}, it can be shown that for  $\lambda \in (0,1)$, a consistent estimator of $\rho$ is given by
\begin{equation}
\label{e:rho}
\hat{\rho}^{(1)}_{n,k_0,\lambda}(r_n)=\log_{\lambda} \frac{\max_{k_0+1 \leq i \leq [\lambda\bar{k}_{n,k_0}(r_n)]} (i-k_0+1)^{1/2}|\widehat{\xi}_{k_0,i}(n)-\widehat{\xi}_{k_0,[\lambda\bar{k}_{n,k_0}(r_n)]}(n)|}{\max_{k_0+1 \leq i \leq [\bar{k}_{n,k_0}(r_n)]} (i-k_0+1)^{1/2}|\widehat{\xi}_{k_0,i}(n)-\widehat{\xi}_{k_0,[\bar{k}_{n,k_0}(r_n)]}(n)|}-\frac{1}{2}
\end{equation}
The detailed presentation of the proof of Proposition \ref{thm:k-est} shall be the subject of another work. Here, we shall  only demonstrate its application in practice (see Tables \ref{tab:ARE_k0k_1} and \ref{tab:ARE_k0k_2}). 

We next describe a methodology which allows for the estimation of $k$ when the trimming parameter $k_0$ is unknown. In this direction, we start with an initial choice of the parameter $k$. From this initial choice of $k$, we estimate the trimming parameter, $k_0$ using EWST Algorithm \ref{algo:seq}. With this choice of $\hat{k}_0$, we obtain an estimate for $k$ by using Proposition \ref{thm:k-est}. We iterate between the values of $k$ and $k_0$, unless convergence is obtained. Next, we describe the methodology more formally:

\begin{algorithm}[H]
	\caption{ \label{algo:seq-k}}
	\small
	\begin{algorithmic}[1]
		\STATE Set a threshold $\tau$ and $i=1$.
		\STATE Choose $k$ as a function of $n$. Let $\hat{k}^{(0)}$ be this initial choice.
		\STATE Let $i=i+1$.
		\STATE With $k=\hat{k}^{(i)}$, obtain $\hat{k}_0^{(i)}$ using Algorithm \ref{algo:seq}
		\STATE With $k_0=\hat{k}_0^{(i)}$, obtain $\hat{k}_{(i+1)}$ using \eqref{e:k-est} in Proposition \ref{thm:k-est}.
		\STATE If $|\hat{k}^{(i+1)}-\hat{k}^{(i)}|>\tau$, goto step 4 else goto step 7
		\STATE Return $\hat{k}= \hat{k}^{(i)}$ and $\hat{k}_0=\hat{k}_0^{(i)}$.
	\end{algorithmic}
\end{algorithm}

In order to evaluate the performance of Algorithm \ref{algo:seq-k}, we first consider the ARE of the adaptive trimmed Hill, $\widehat{\xi}_{\hat{k}_0,\hat{k}}$ relative to the trimmed Hill, $\widehat{\xi}_{k_0,k^*_{n,k_0}}$ where $k^*_{n,k_0}$ is obtained as in \eqref{e:k-sim}. Table \ref{tab:ARE_k0k_1} shows the ARE values of for Frechet and T distributions with varying tail indices (see \eqref{e:heavy-dist}). The number of outliers $k_0$ is fixed at 10 and two values of $L=5, 20$ are chosen. The two columns correspond to the case where $\rho$ is either fixed at constant 1 or estimated using \eqref{e:rho} for $\lambda=0.6$.
\begin{table}[H]
	\centering
	\begin{tabular}{|c|c|cc|cc|cc|cc|cc|}\hline
		L  &n    & Frechet(5) & & Frechet(2)& &Frechet(1)& &$|{\rm T}|(4)$ & &	$|{\rm T}|(10)$ &\\\hline
		&    &   1 & $\hat{\rho}^{(1)}$ & 1 & $\hat{\rho}^{(1)}$ &  1 & $\hat{\rho}^{(1)}$ & 1 & $\hat{\rho}^{(1)}$ & 1 & $\hat{\rho}^{(1)}$ \\\hline
		&100  &  0.74 &  0.63 &  0.31 &  0.18 &  0.27 &  0.24 &  0.86 &  0.63 &  0.88 &  0.64\\
		5  &200  &  0.71 &  0.64 &  0.37 &  0.51 &  0.72 &  0.49 &  0.87 &  0.66 &  0.86 &  0.64
		\\
		&500  &  0.78 &  0.60 &  0.66 &  0.53 &  0.71 &  0.46 &  0.83 &  0.62 &  0.80 &  0.69
		\\\hline
		&100  & 0.76 & 0.61 &  0.87 &  0.68 &  0.50 &  0.55 &  0.94 &  0.74 &  0.90 &  0.77
		\\
		20 &200  &  0.74 & 0.71 & 0.86 & 0.60 & 0.74 & 0.47 & 0.89 & 0.66 & 0.85 & 0.69
		\\
		&500  &  0.79 &  0.63 &  0.88 &  0.51 &  0.75 &  0.44 &  0.82 &  0.80 & 0.63 &	0.68
		\\\hline
	\end{tabular}
	\caption{{\rm ARE} of the adaptive trimmed Hill relative to trimmed Hill for $k_0=10$.}
	\label{tab:ARE_k0k_1}
\end{table}

We observe that the ARE values are nearly 75\% for the Frechet and become as large as 90\% for the T distribution. The performance is better when $\rho=1$ rather than estimated from the data using \eqref{e:rho}. Large values of $\alpha=1/\xi$ lead to greater ARE values for both T and Frechet. This behavior is similar to that observed in the case of known $k$ (see Figure \ref{fig:alpha-rol}). Increase in the severity of outliers, $L$ leads to overall improvement in the efficiency, a phenomenon also seen previously in Section \ref{subsec:ha_large}.

In order to allow for a comparative baseline to our results in Table \ref{tab:ARE_k0k_1}, we replicate the settings of Tables 3 and 6 in \cite{drees}. We  consider the ratio of root mean squared error of the adaptive trimmed to that of the trimmed as:
\begin{equation*}
\label{e:R-def}
R=\sqrt{{\rm MSE}(\widehat{\xi}_{\hat{k}_0,\hat{k}})}\Big/\sqrt{{\rm MSE}(\widehat{\xi}_{k_0,k^*_{n,k_0}})}
\end{equation*} 
The results in \cite{drees} correpond to $k_0=0$ and $k^*_{n,k_0}=k_n^{\rm opt, sim}$. As can been from the Table, our results  nearly match the ones obtained from \cite{drees}. This further indicates the efficiency of the proposed Algorithm \ref{algo:seq-k} in the joint estimation of $k_0$ and $k$.

\begin{table}[H]
	\centering
	\begin{tabular}{|c|c|cc|cc|cc|cc|cc|}\hline
		n  &L    & & Frechet(5) & & Frechet(2)& & Frechet(1)& & $|{\rm T}|(4)$&	 &$|{\rm T}|(10)$\\\hline
		&    &   1 & $\hat{\rho}^{(1)}$ & 1 & $\hat{\rho}^{(1)}$ &  1 & $\hat{\rho}^{(1)}$ & 1 & $\hat{\rho}^{(1)}$ & 1 & $\hat{\rho}^{(1)}$ \\\hline
		& 5     &  1.16 &  1.26 &  1.80 &  2.34 &  1.93 &  2.04 &  1.08 &  1.26 &  1.06 &  1.25\\
		100  & 20    &  1.15 &  1.28 &  1.07 &  1.22 &  1.41 &  1.35 &  1.03 &  1.16 &  1.05 &  1.14\\
		& drees & 1.29& 	1.22 & 1.08 & 1.24 & 1.28 & 1.12 & 1.36 & 1.15 & 1.24 & 1.48\\\hline
		& 5     &  1.18 &  1.25 &  1.65 &  1.40 &  1.18 &  1.44 &  1.07 &  1.23 &  1.08 &  1.25\\
		200  & 20    &  1.16 &  1.19 &  1.08 &  1.30 &  1.16 &  1.46 &  1.06 &  1.23 &  1.09 &  1.20\\
		& drees &  1.19 & 1.21 & 1.08 & 1.23 & 1.34 & 1.14 & 1.28 & 1.14 & 1.28 & 1.46\\\hline
		& 5     &  1.14 &  1.29 &  1.23 &  1.38 &  1.18 &  1.48 &  1.10 &  1.27 &  1.12 &  1.20\\
		500  & 20    &  1.12 &  1.26 &  1.07 &  1.40 &  1.16 &  1.51 &  1.11 &  1.12 &  1.26 & 1.21\\
		& drees & 1.12 & 1.18 & 1.05 & 1.26 & 1.30 & 1.12 & 1.27 & 1.14 & 1.3 & 1.41
		\\\hline
	\end{tabular}
	\caption{Ratio of mean squared errors: adaptive trimmed Hill  to  trimmed Hill for $k_0=10$.}
	\label{tab:ARE_k0k_2}
\end{table}

\section{Comparisons with existing estimators and adaptivity}
\label{sec:comp}

\subsection{Comparison with other robust estimators}
\label{subsec:comp-prev}

In this section, we present a comparative analysis of the performance of our proposed trimmed Hill estimator, $\hat{\xi}_{k_0,k}$ with respect to the already existing robust tail estimation procedures in the literature. For observations from the Pareto distribution, a robust estimator of $\alpha$ based on the trimmed Hill estimator, $\hat{\xi}_{\rm k_0,n-1}$ is given by
\begin{equation}
\label{e:alpha-trim}
\hat{\alpha}_{\rm TRIM}=\left(1-\frac{2}{n}\right)\frac{1}{\hat{\xi}_{k_0,n-1}}
\end{equation}
where $(1-2/n)$ is the correction factor for $\hat{\alpha}_{\rm MLE}$ as in \cite{Brzezinski2016}. 
\vspace{-5mm}
\begin{figure}[H]
	\centering
	\includegraphics[width=0.45\textwidth]{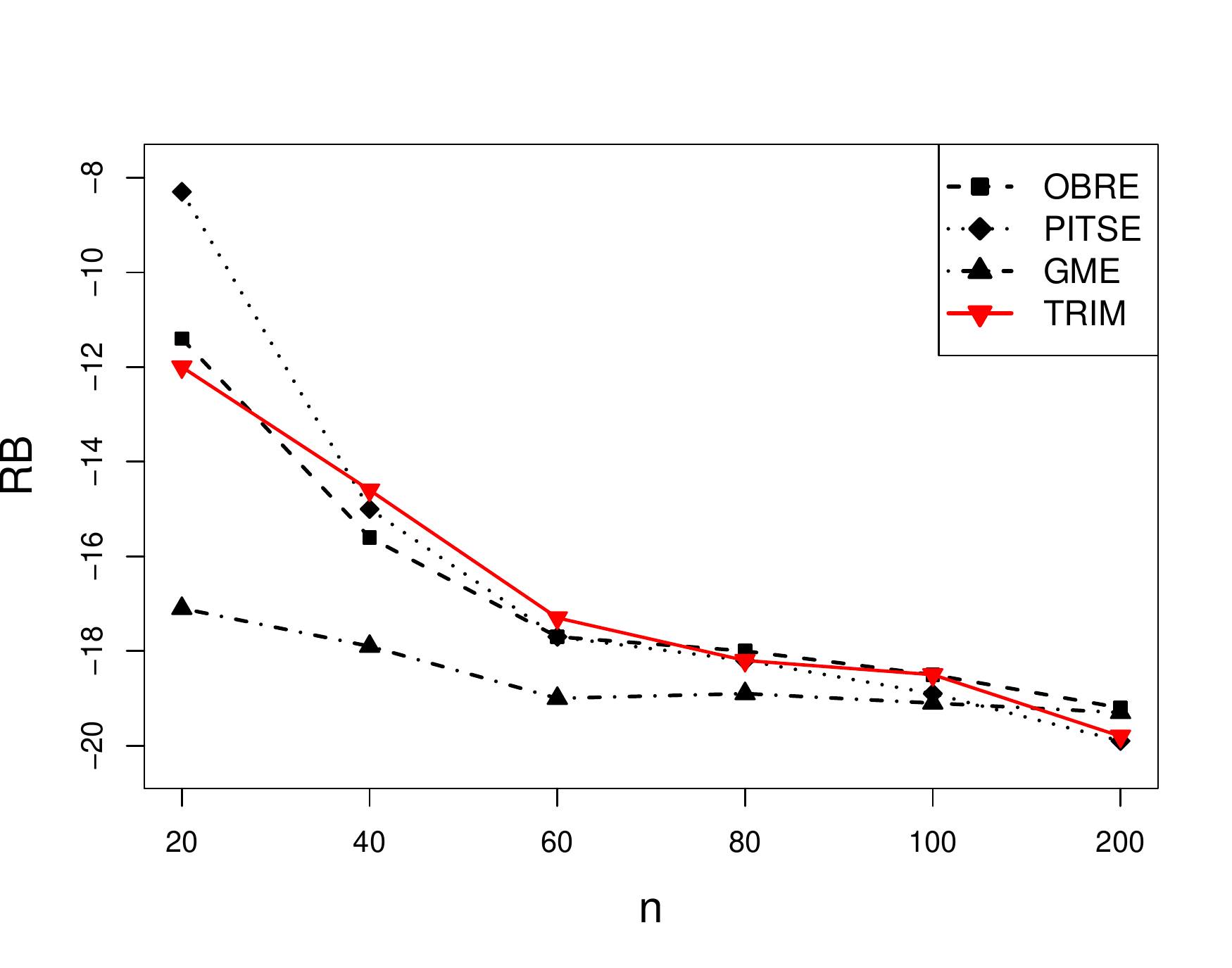}
	\includegraphics[width=0.45\textwidth]{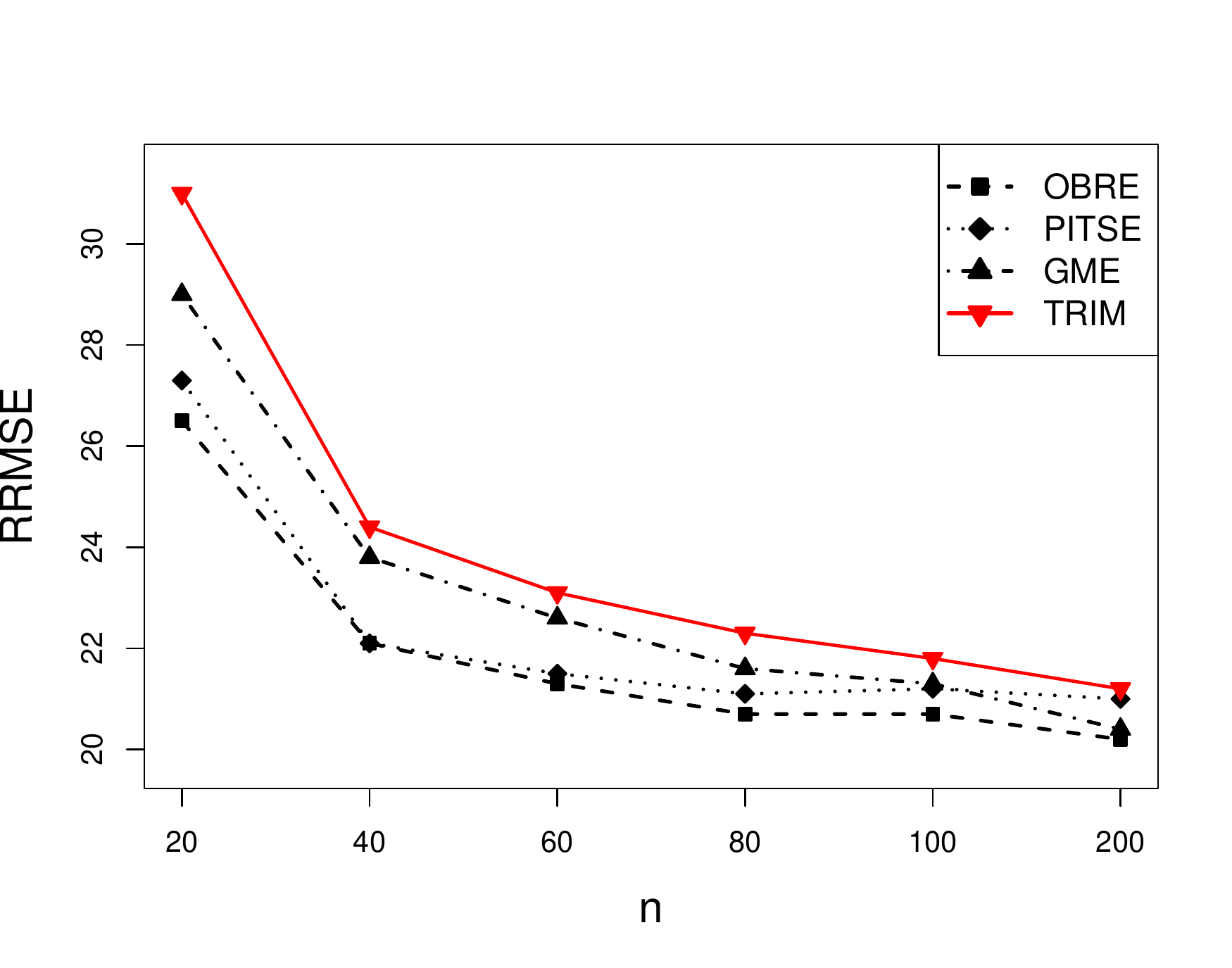}
\end{figure}
\vspace{-15mm}
\begin{figure}[H]
	\centering
	\includegraphics[width=0.45\textwidth]{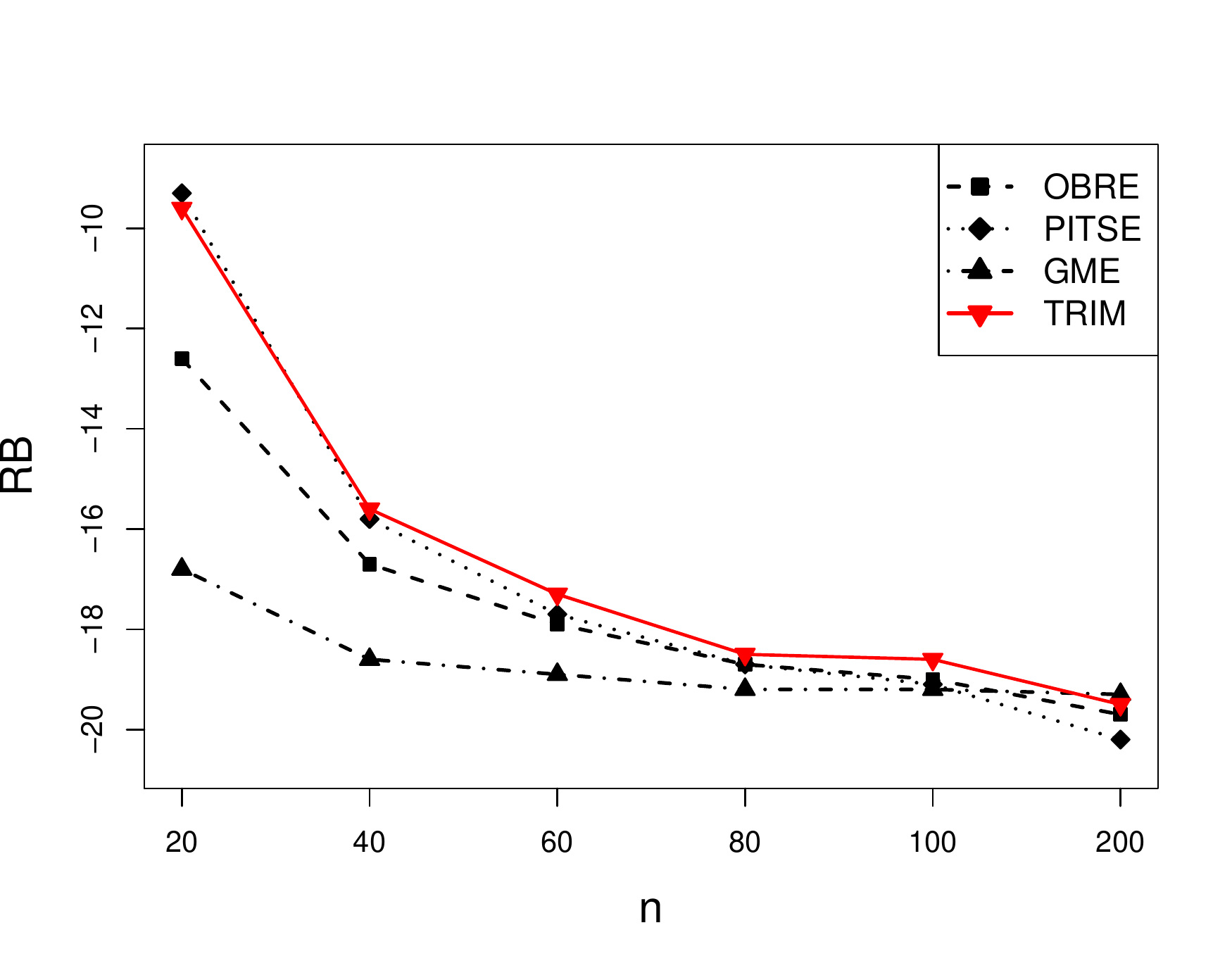}
	\includegraphics[width=0.45\textwidth]{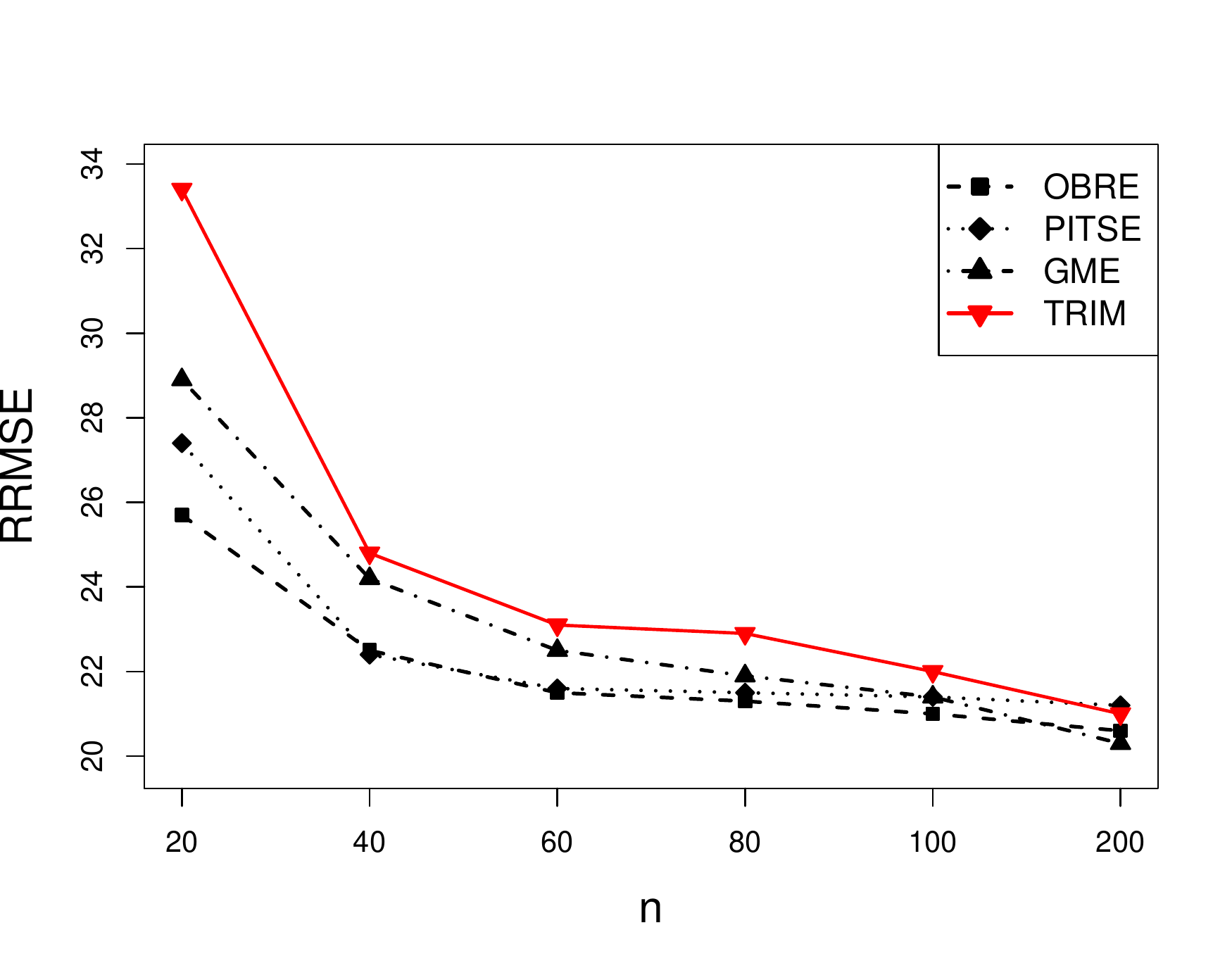}
	\caption{Performance of robust estimators for  $0.9 P(\alpha,1)+0.1 P(\alpha,1000)$ at ARE=78\%. Top left and right correspond to RB and RRMSE values  for $\alpha=1$. Bottom left and right correspond to RB and RRMSE values  for $\alpha=3$.}
	\label{fig:mix-78}
\end{figure}

The comprehensive comparative analysis in \cite{Brzezinski2016} evaluates many robust estimators of the exponent $\alpha=1/\xi$ with respect to the maximum likelihood estimator $\hat{\alpha}_{\rm MLE}$ for i.i.d. Pareto observations.  
The  class of estimators used in \cite{Brzezinski2016} include the optimal B-robust estimator, (OBRE) proposed in \cite{CJS:CJS247}, the weighted maximum likelihood estimator (WMLE)  introduced in \cite{wmle}, the generalized median estimator (GME) of \cite{MR1856199}, the partial density component estimator (PDCE) proposed in \cite{Vandewalle:2007:RET:1280299.1280640} and the probability integral transform statistic estimator (PITSE) of \cite{pitse}. Among these estimators of $\alpha$, the OBRE, PITSE and GME exhibit a superior performance in comparison to the rest and shall be used as the comparative baseline.

\vspace{-5mm}
\begin{figure}[H]
	\centering
	\includegraphics[width=0.45\textwidth]{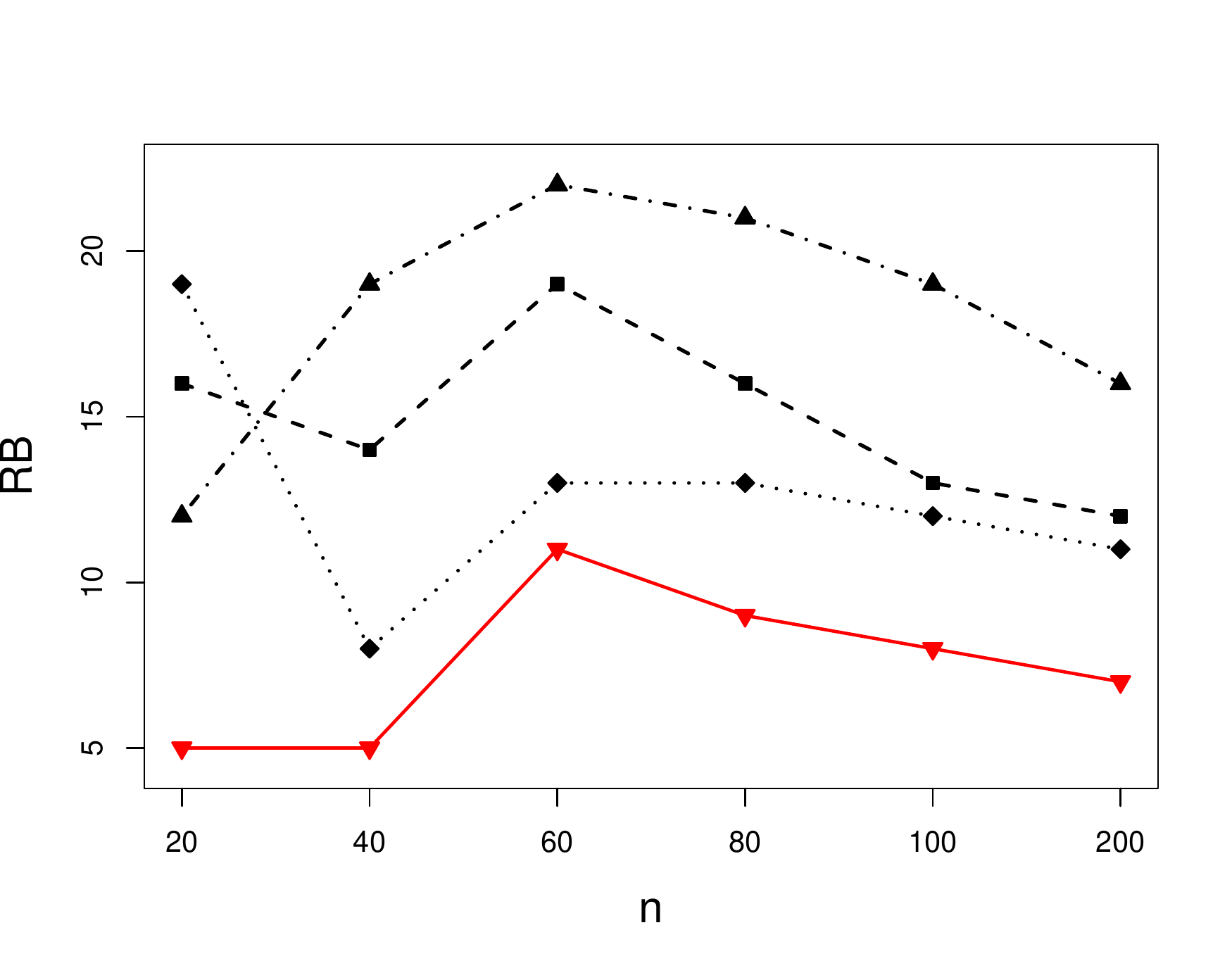}
	\includegraphics[width=0.45\textwidth]{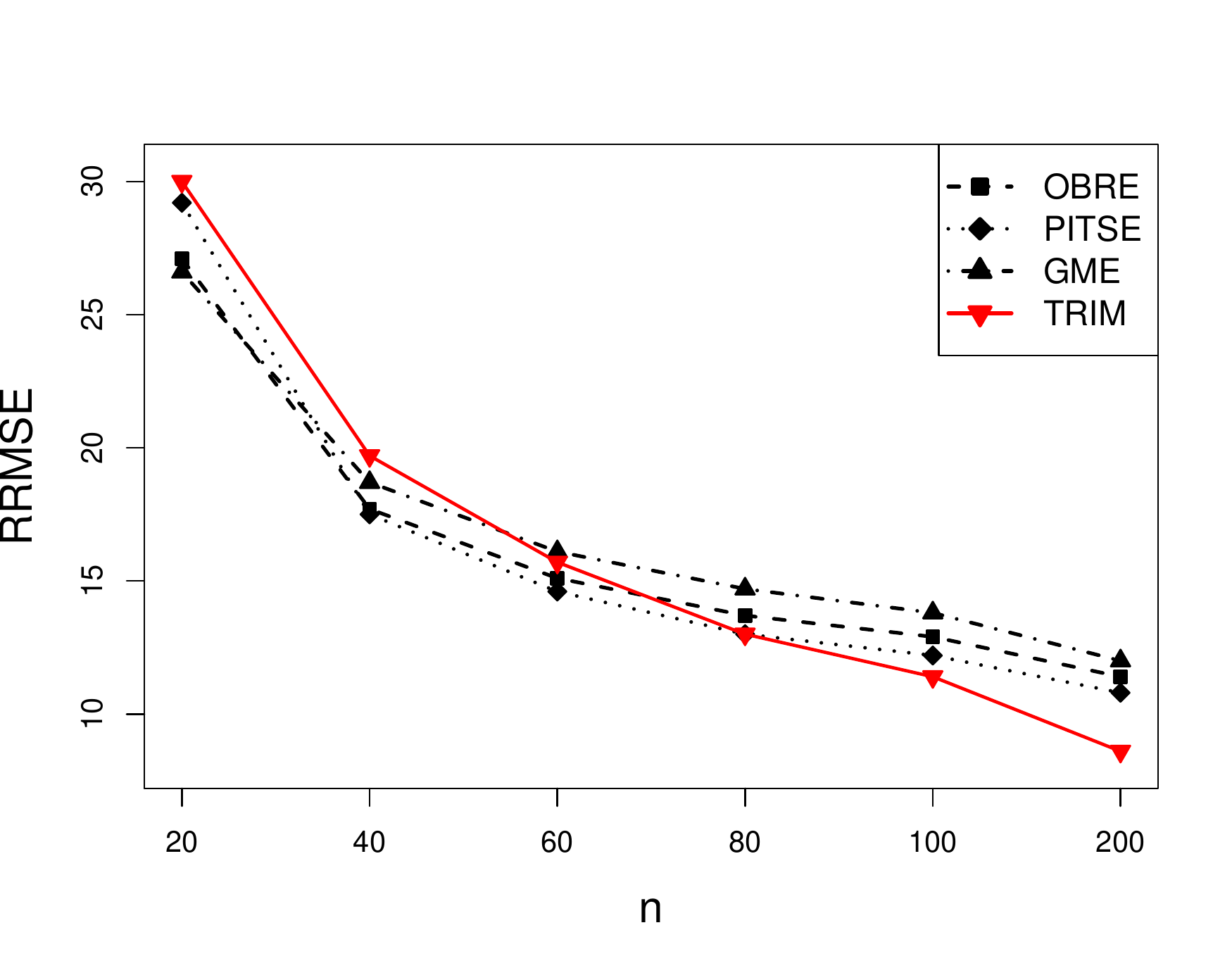}
\end{figure}
\vspace{-15mm}
\begin{figure}[H]
	\centering
	\includegraphics[width=0.45\textwidth]{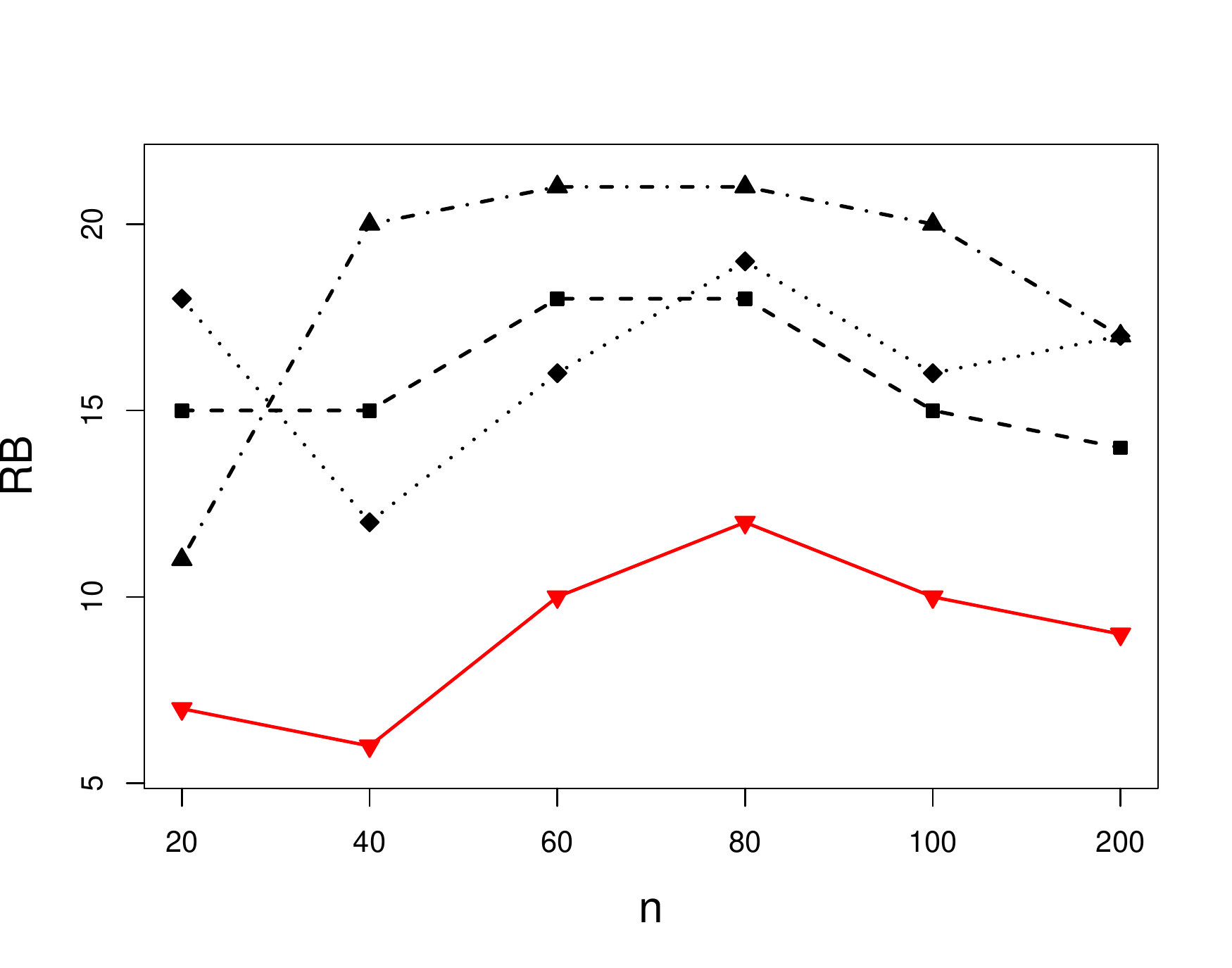}
	\includegraphics[width=0.45\textwidth]{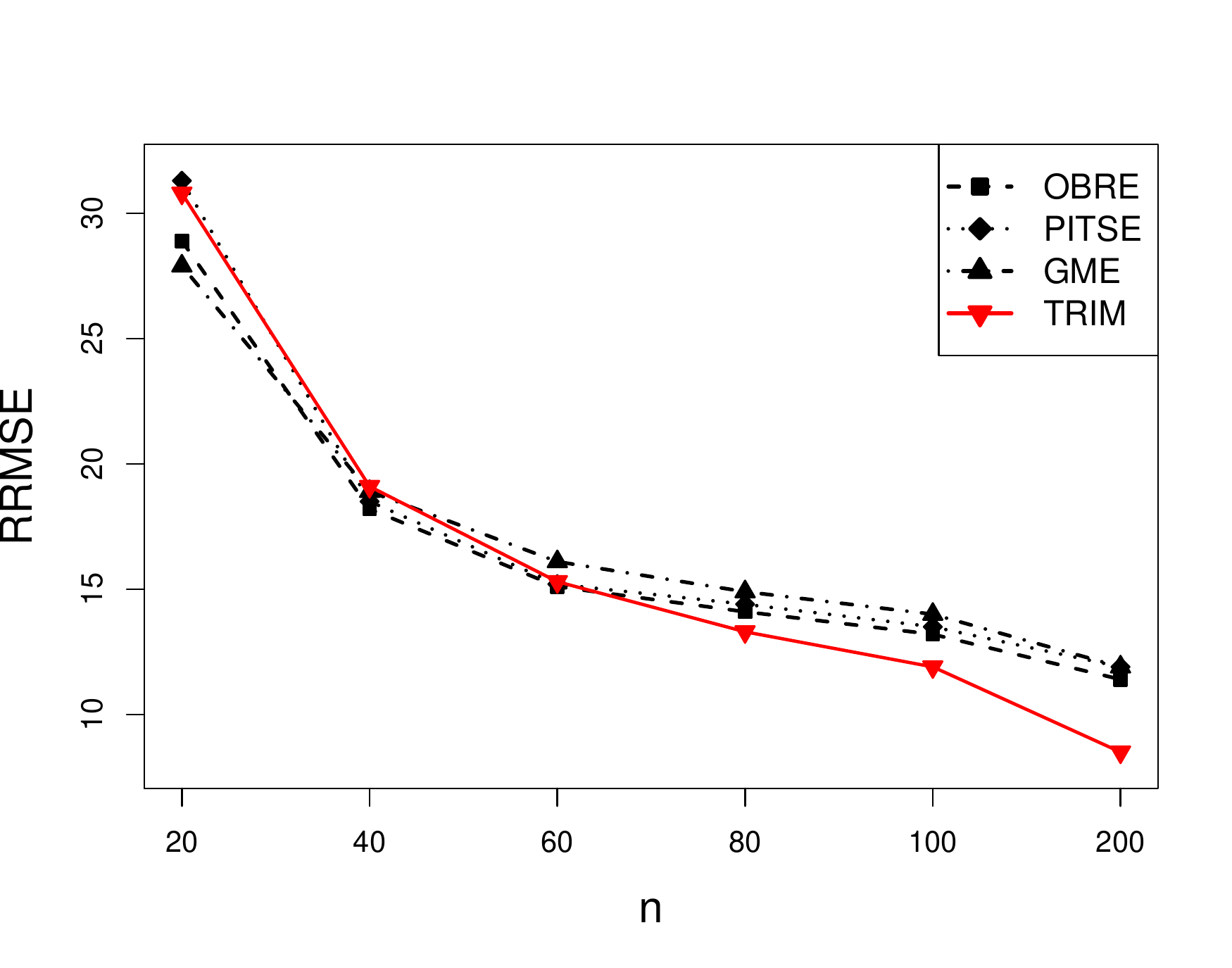}
	\caption{Performance of robust estimators where 5\% observations of $P(\alpha,1)$ are inflated by 10 at ARE=78\%. Top left and right correspond to RB and RRMSE values for $\alpha=1$. Bottom left and right correspond to RB and RRMSE values for $\alpha=3$.}
	\label{fig:scl-78}
\end{figure}
The comparison criterion chosen is the relative bias, RB and relative mean squared error, RRMSE as in \cite{Brzezinski2016}. The explicit formulas for RB and RRMSE are given by
\begin{eqnarray}
\label{e:b-mse}
{\rm RB}(\alpha)&=&\frac{1}{\alpha}\Big(\frac{1}{m}\sum_{i=1}^m (\widehat{\alpha}_i-\alpha)\Big)\times 100\%\\\nonumber
{\rm RRMSE}(\alpha)&=&\frac{1}{\alpha}\Big(\frac{1}{m}\sum_{i=1}^m (\widehat{\alpha}_i-\alpha)^2\Big)^{1/2} \times 100\%
\end{eqnarray} 
where the $\widehat{\alpha}_i$'s are independent realizations of a particular estimator of $\alpha=1/\xi$. 

To be able to compare with  \cite{Brzezinski2016}, we need to determine $k_0$ in \eqref{e:alpha-trim} so as to match the target ARE (Asymptotic Relative Efficiency) of the estimators considered therein. By relation  \eqref{e:covar} in Proposition \ref{prop:xi-exp} it is easy to see that
\begin{equation}
\label{e:are-trim}
{\rm ARE}(\hat{\alpha}_{\rm TRIM})=\frac{{\rm Var}(\hat{\alpha}_{\rm MLE})}{{\rm Var}(\hat{\alpha}_{\rm TRIM})}\approx \frac{1/n}{1/(n-1-k_0)}
\end{equation}

where the last asymptotic equivalence follows by a simple application of delta method to the function form of  $\hat{\alpha}_{\rm TRIM}$ in terms of the statistic $\hat{\xi}_{k_0,n-1}$. Given $n$, to achieve a target ARE, we use \eqref{e:are-trim} to solve for $k_0$.

\vspace{-5mm}
\begin{figure}[H]
	\includegraphics[width=0.45\textwidth]{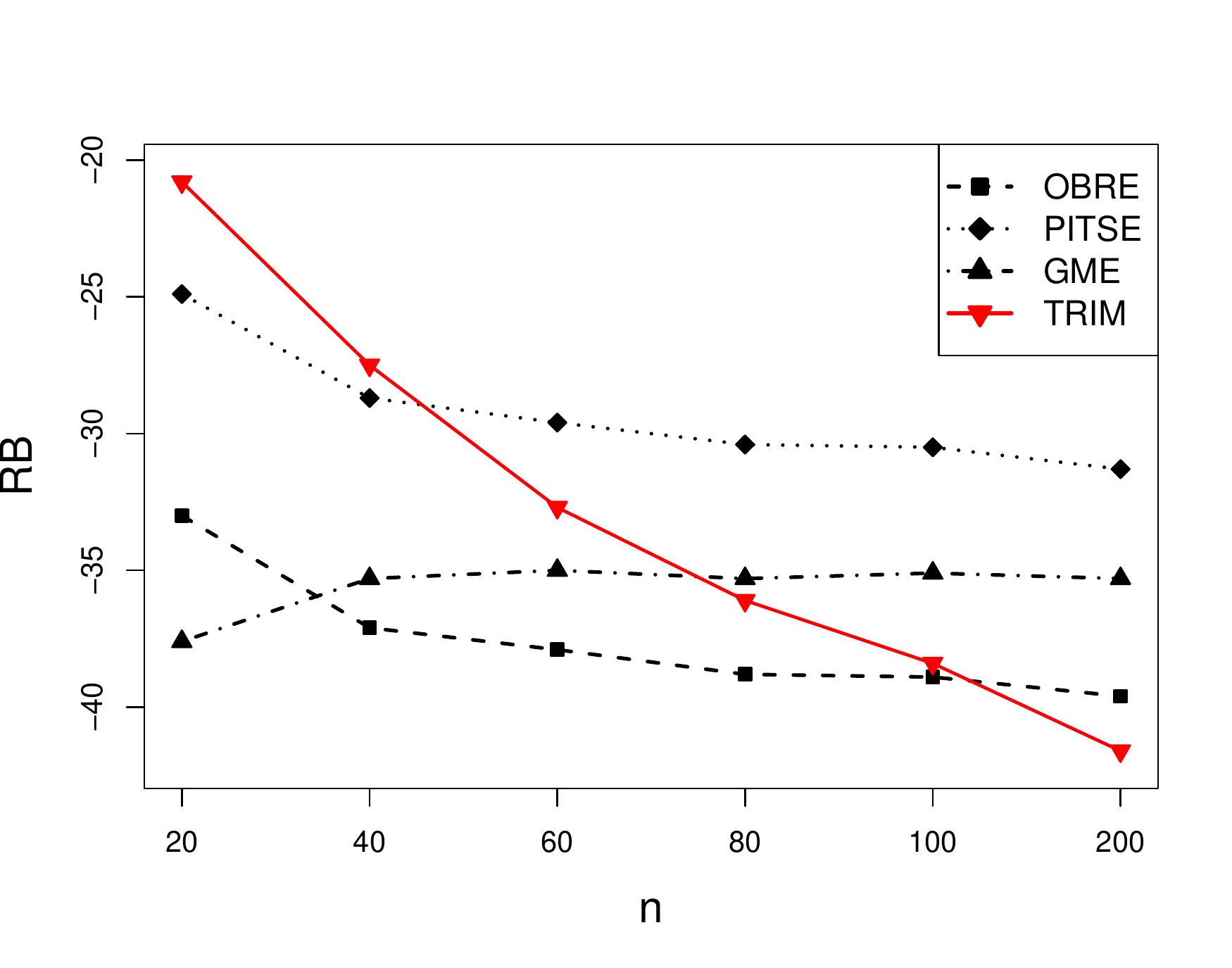}
	\includegraphics[width=0.45\textwidth]{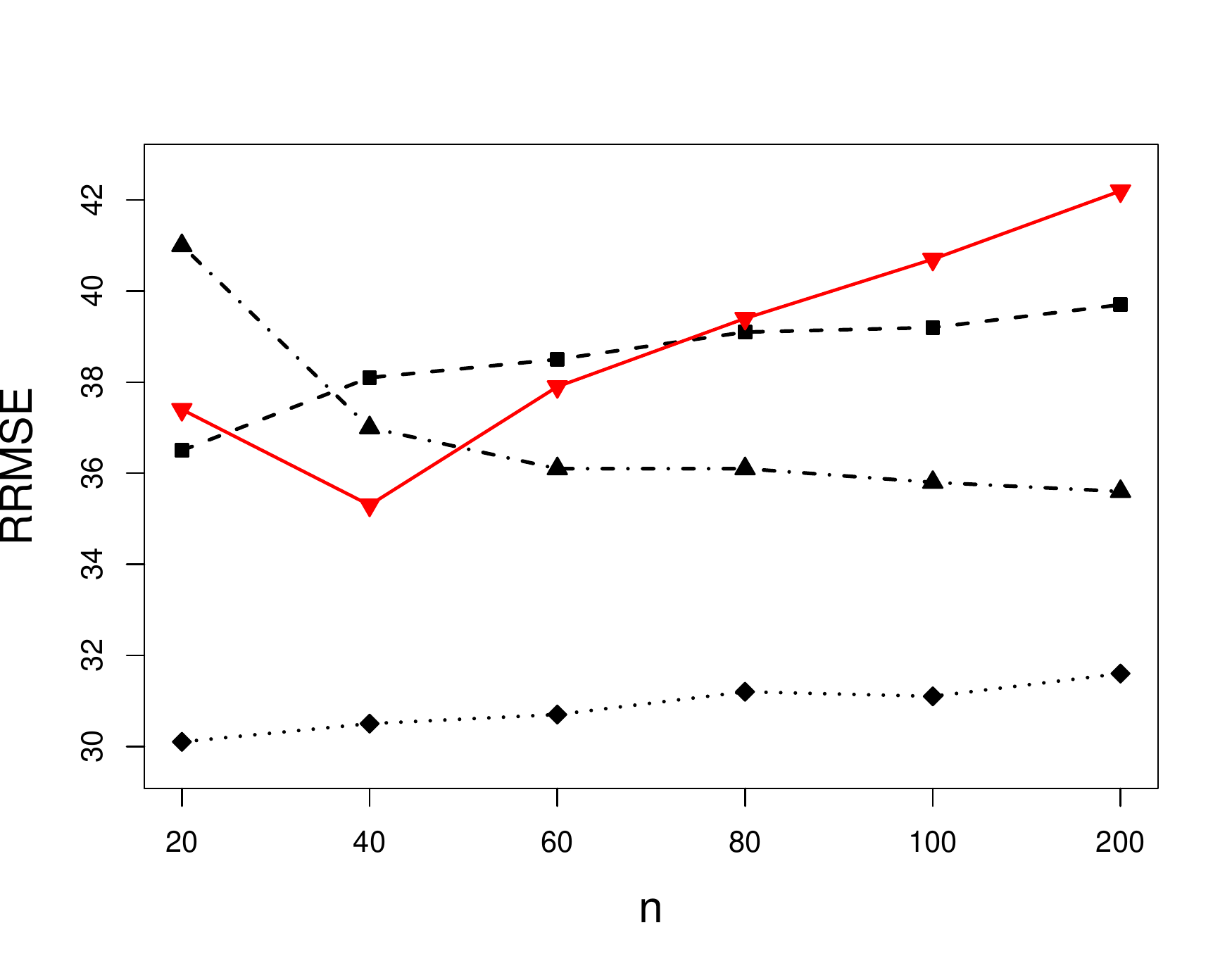}
\end{figure}
\vspace{-15mm}
\begin{figure}[H]
	\includegraphics[width=0.45\textwidth]{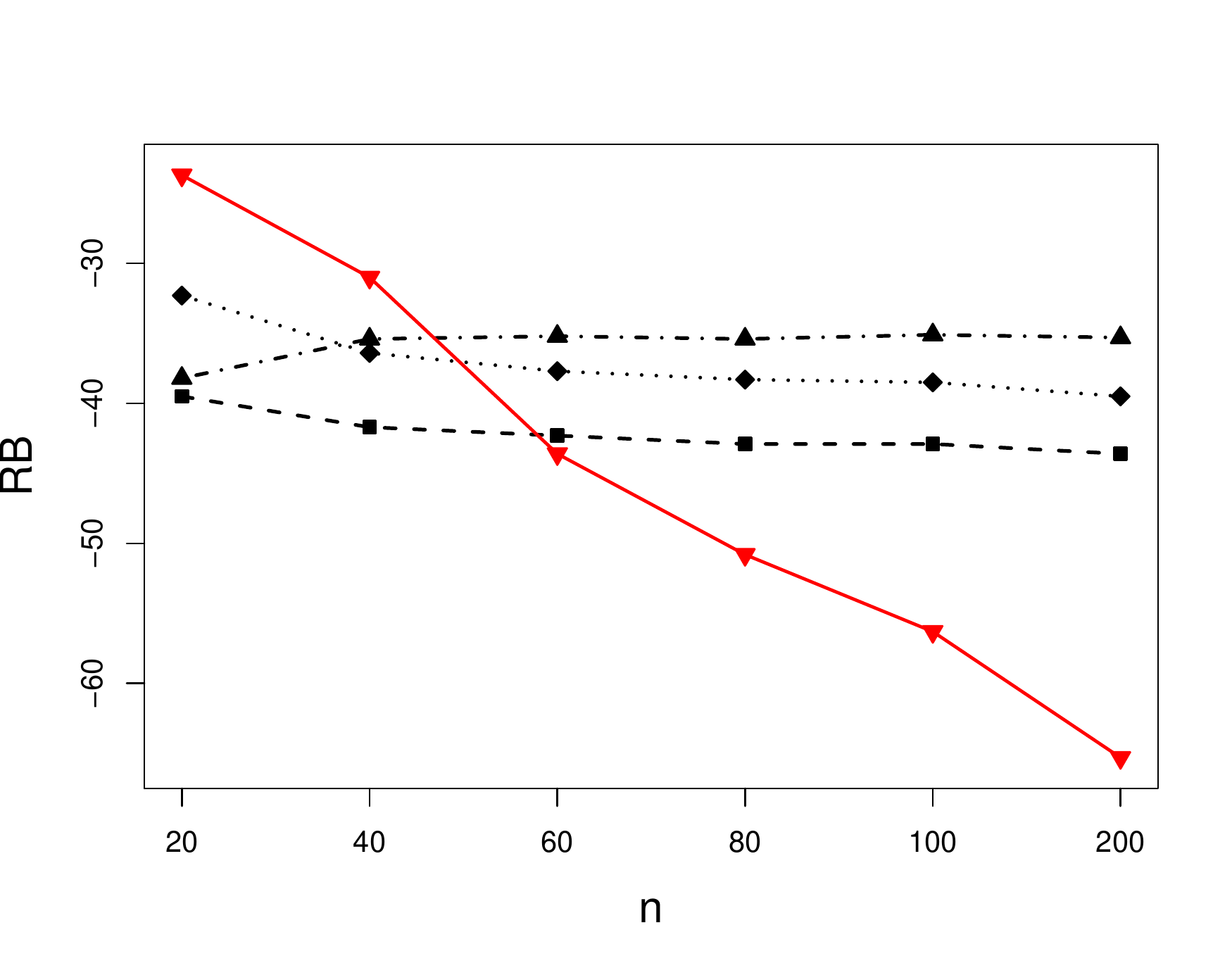}
	\includegraphics[width=0.45\textwidth]{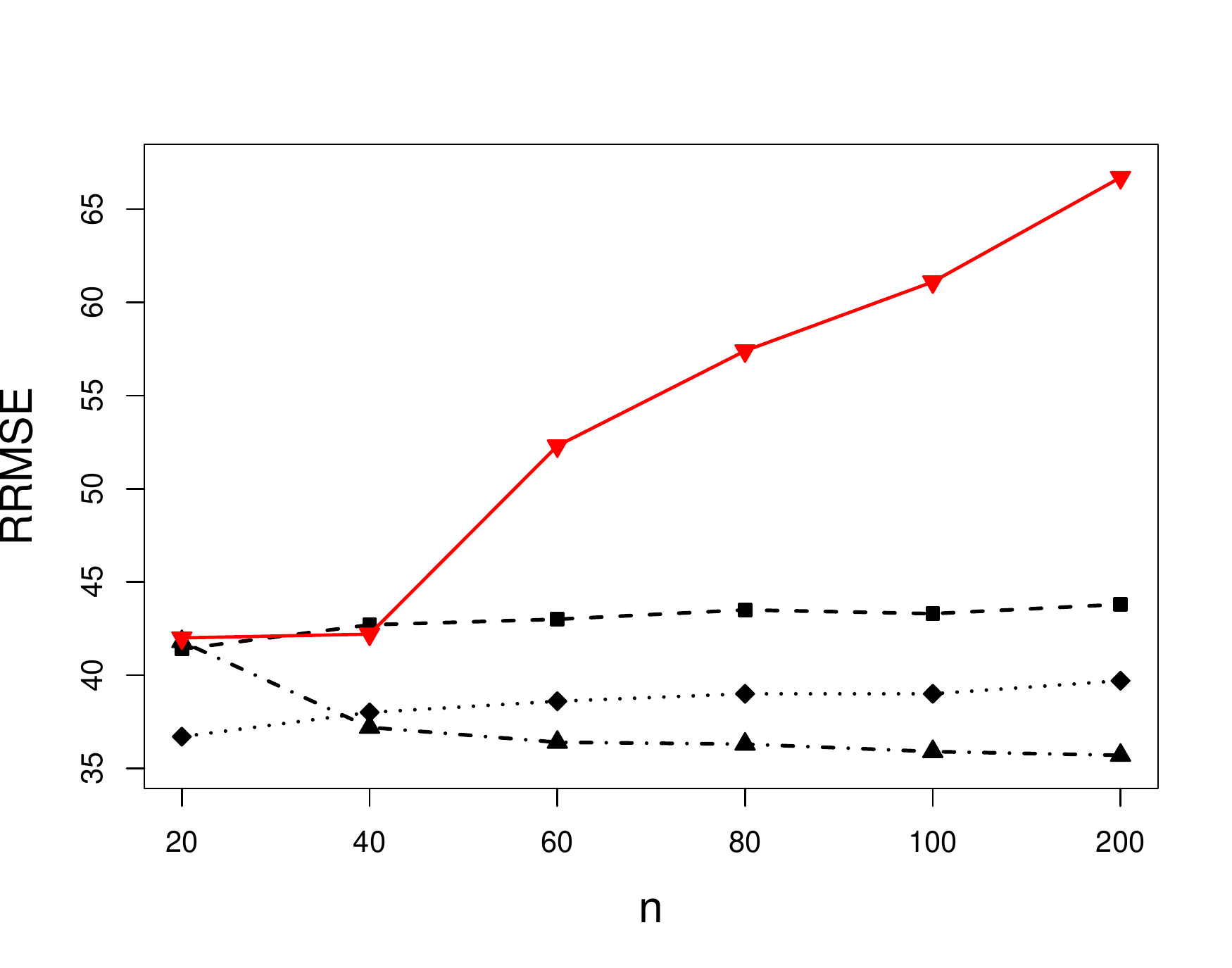}
	\caption{Performance of robust estimators for  $0.9 P(\alpha,1)+0.1 P(\alpha,1000)$ at ARE=94\%. Top left and right correspond to RB and RRMSE values  for $\alpha=1$. Bottom left and right correspond to RB and RRMSE values  for $\alpha=3$.}
	\label{fig:mix-90}
\end{figure}

As in \cite{Brzezinski2016}, the data sets are simulated from the Pareto distribution ${\rm Pareto}(1,1)$ and contaminated in two ways. In the first method of introducing outliers, we generate observations from the following mixture distribution
\begin{equation}
\label{e:F-cont-1}
F=(1-\varepsilon)\:{\rm Pareto}(\alpha,1)+\varepsilon\:{\rm Pareto}(\alpha,1000)
\end{equation}
for $\varepsilon \in (0,1)$ and $\alpha>0$. In the second method of contamination, $s$ proportion of the  observations is randomly selected  from ${\rm Pareto}(\alpha,1)$ and multiplied by a constant factor of 10.
\vspace{-5mm}
\begin{figure}[H]
	\centering
	\includegraphics[width=0.45\textwidth]{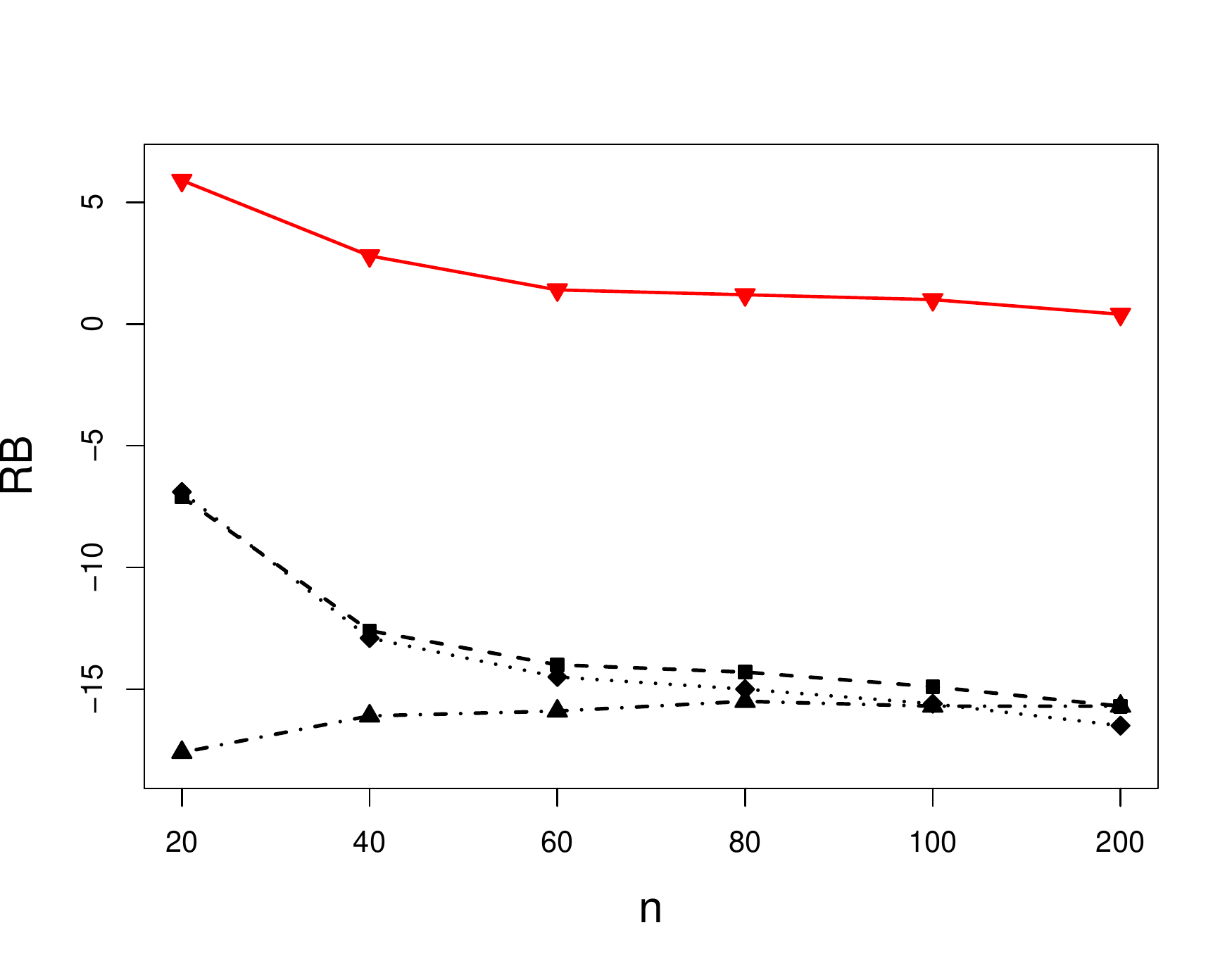}
	\includegraphics[width=0.45\textwidth]{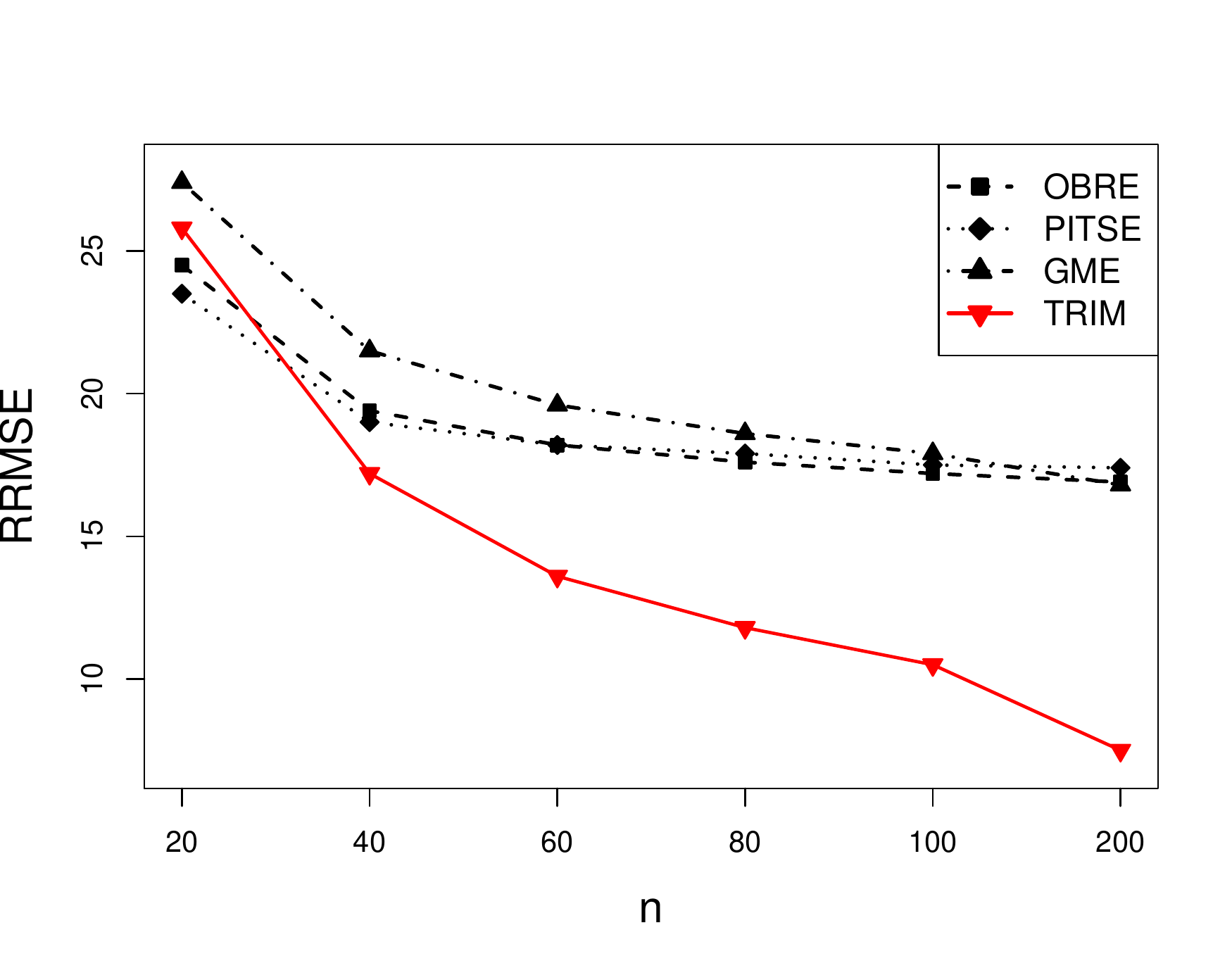}
\end{figure}
\vspace{-15mm}
\begin{figure}[H]
	\centering
	\includegraphics[width=0.45\textwidth]{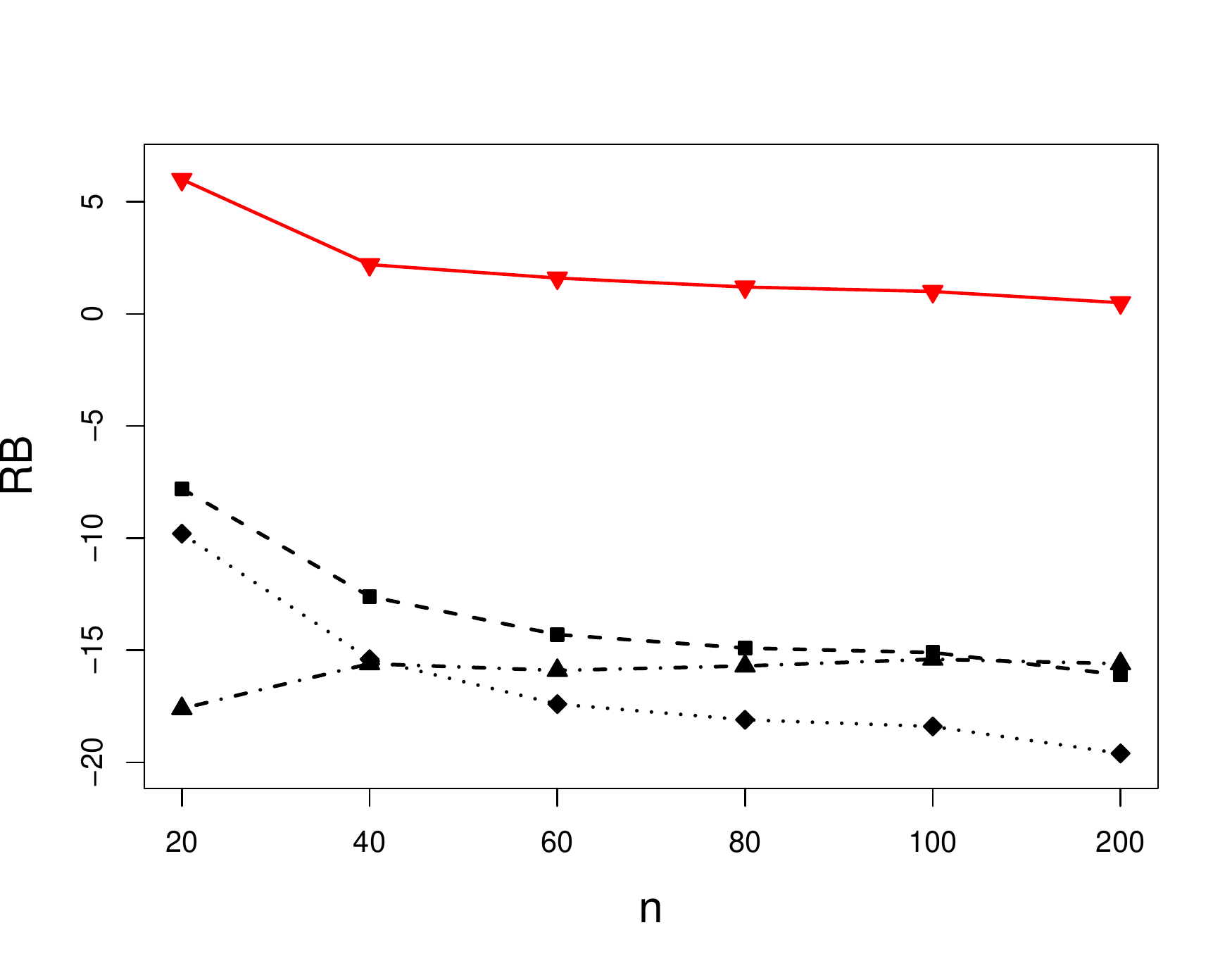}
	\includegraphics[width=0.45\textwidth]{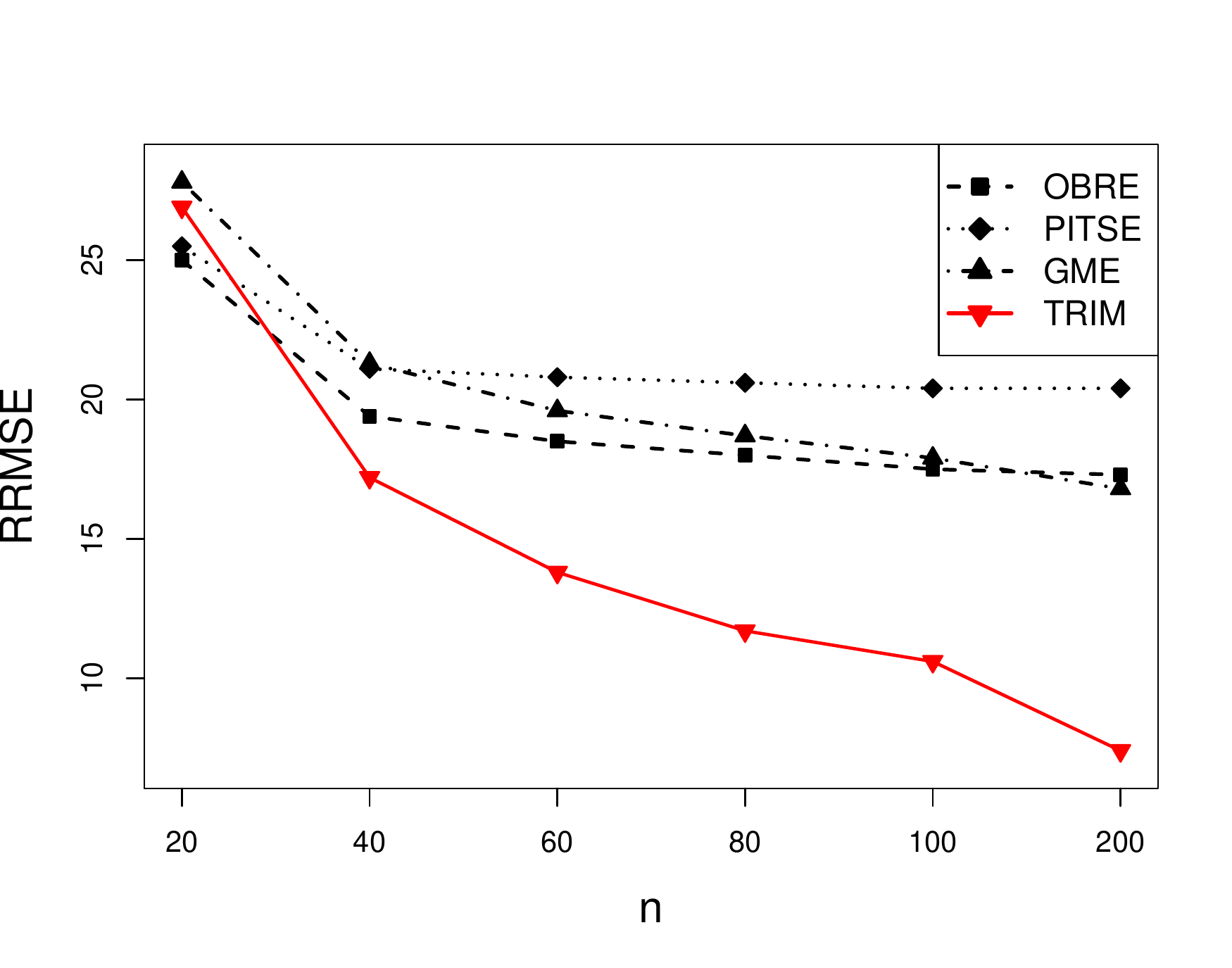}
	\caption{Performance of robust estimators for  $0.95 P(\alpha,1)+0.05 P(\alpha,1000)$ at ARE=94\%. Top left and right correspond to RB and RRMSE values  for $\alpha=1$. Bottom left and right correspond to RB and RRMSE values  for $\alpha=3$.}
	\label{fig:mix-95}
\end{figure}

For both methods of data contamination, we analyze performance of the four estimators viz OBRE, PITSE, GME and TRIM. We first fix the asymptotic relative efficiency for these estimators at $78\%$. Figure \ref{fig:mix-78} shows the  performance under the first method of data contamination with $\varepsilon=0.9$ and $\alpha=1$ and 3.  We observe that the performance of $\hat{\alpha}_{\rm TRIM}$ closely follows that of $\hat{\alpha}_{\rm OBRE}$, $\hat{\alpha}_{\rm PITSE}$ and $\hat{\alpha}_{\rm GME}$. In fact, all the estimators are relatively similar in this case and their difference is relatively small as the sample size $n$ grows.  Figure \ref{fig:scl-78} on the other hand,  shows the performance under the second method of data contamination with $s=0.05$ and $\alpha=1$ and 3. For this case, we observe the superior performance of $\hat{\alpha}_{\rm TRIM}$ in comparison to the estimators. This behavior is more apparent in larger sample sizes ($n=200$) where the trimmed estimator has more than 50\% lower RRMSE values than the rest.

We next fix the asymptotic relative efficiency for these estimators at $94\%$. Figure \ref{fig:mix-90} shows the  performance under the first method of data contamination with $\varepsilon=0.9$ and $\alpha=1$ and 3.  We observe that in this case the performance of $\hat{\alpha}_{\rm TRIM}$ is relatively poor when compared to that of $\hat{\alpha}_{\rm OBRE}$, $\hat{\alpha}_{\rm PITSE}$ and $\hat{\alpha}_{\rm GME}$ especially for larger sample sizes, $n$. However this phenomenon gets entirely reversed when $\varepsilon=0.95$ (see Figure \ref{fig:mix-95}). The performance of $\hat{\alpha}_{\rm TRIM}$ improves drastically with increase in sample size $n$ and surpasses the performance of all the other robust estimators. For $n=200$, the improvement is up to a factor 200\% in the RRMSE values. The surprising difference in the performance observed in Figures \ref{fig:mix-90} and \ref{fig:adap-94} can be explained as follows.

Since the ARE of $\hat{\alpha}_{\rm TRIM}$ is directly  related to the trimming value $k_0$ (see \eqref{e:alpha-trim}), large ARE or small $k_0$ values can control against small proportion of contamination ($1-\varepsilon=0.05$) but not against large proportions ($1-\varepsilon=0.1$). In scenario of Figure \ref{fig:mix-90}, setting the ARE as $94\%$ and contaminating $10\%$ of the data, our trimmed estimator is artificially forced to include outliers. This leads to the relatively poor performance of $\hat{\alpha}_{\rm TRIM}$. For other estimators, the link between ARE and robustness is not as direct which gives them an advantage. At $5\%$ contamination, our trimmed estimator picks up all the outliers at ARE level $94\%$ and hence outperforms the competitors (Figure \ref{fig:mix-95}).

In the following section, we illustrate an important advantage of our trimmed estimator when $k_0$ is estimated from the data. This allows us to adapt the degree of robustness to the proportion of outliers.

\subsection{Adaptive robustness}
\label{subsec:comp-adap}

\vspace{-5mm}

\begin{figure}[H]
	\centering
	\includegraphics[width=0.45\textwidth]{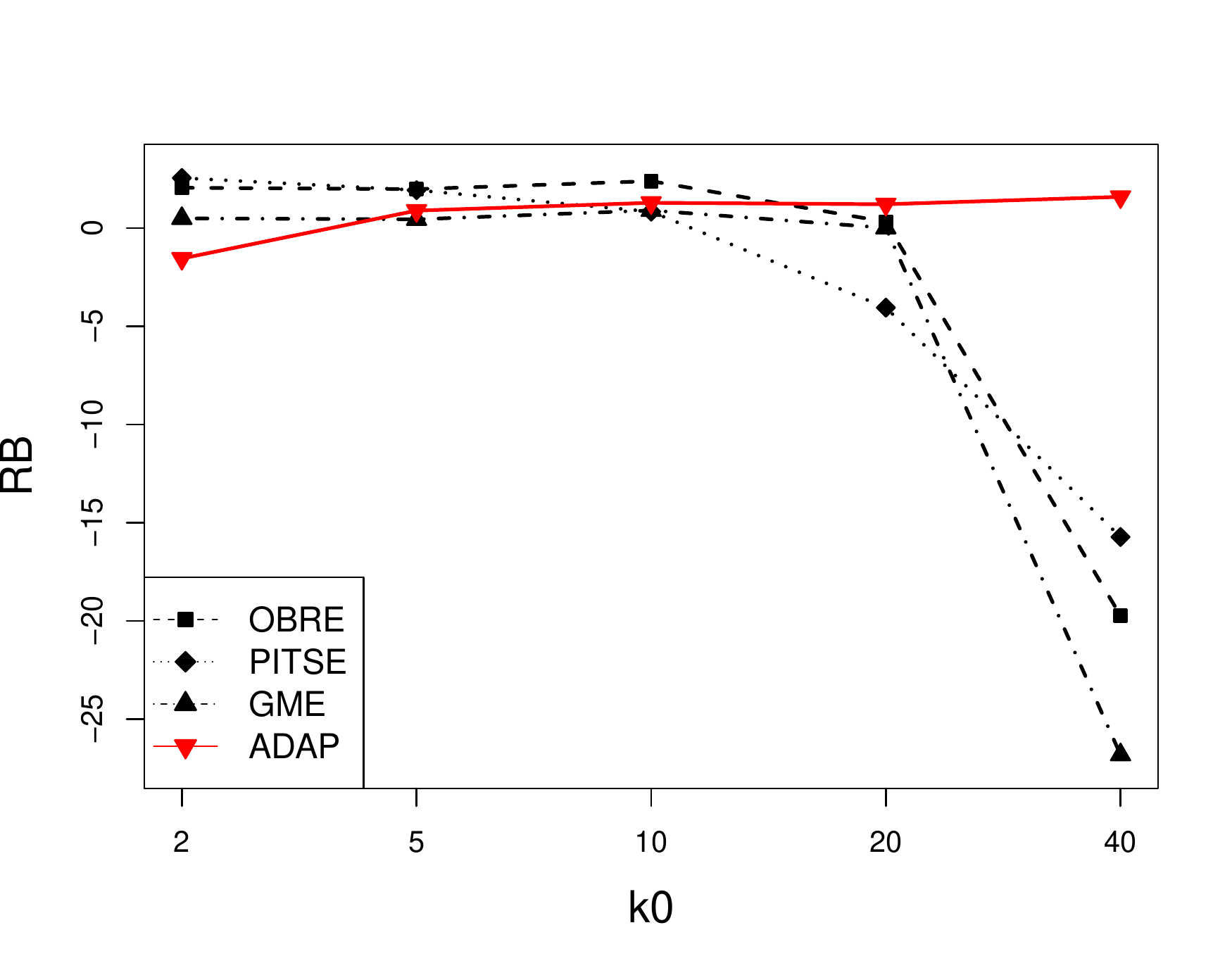}
	\includegraphics[width=0.45\textwidth]{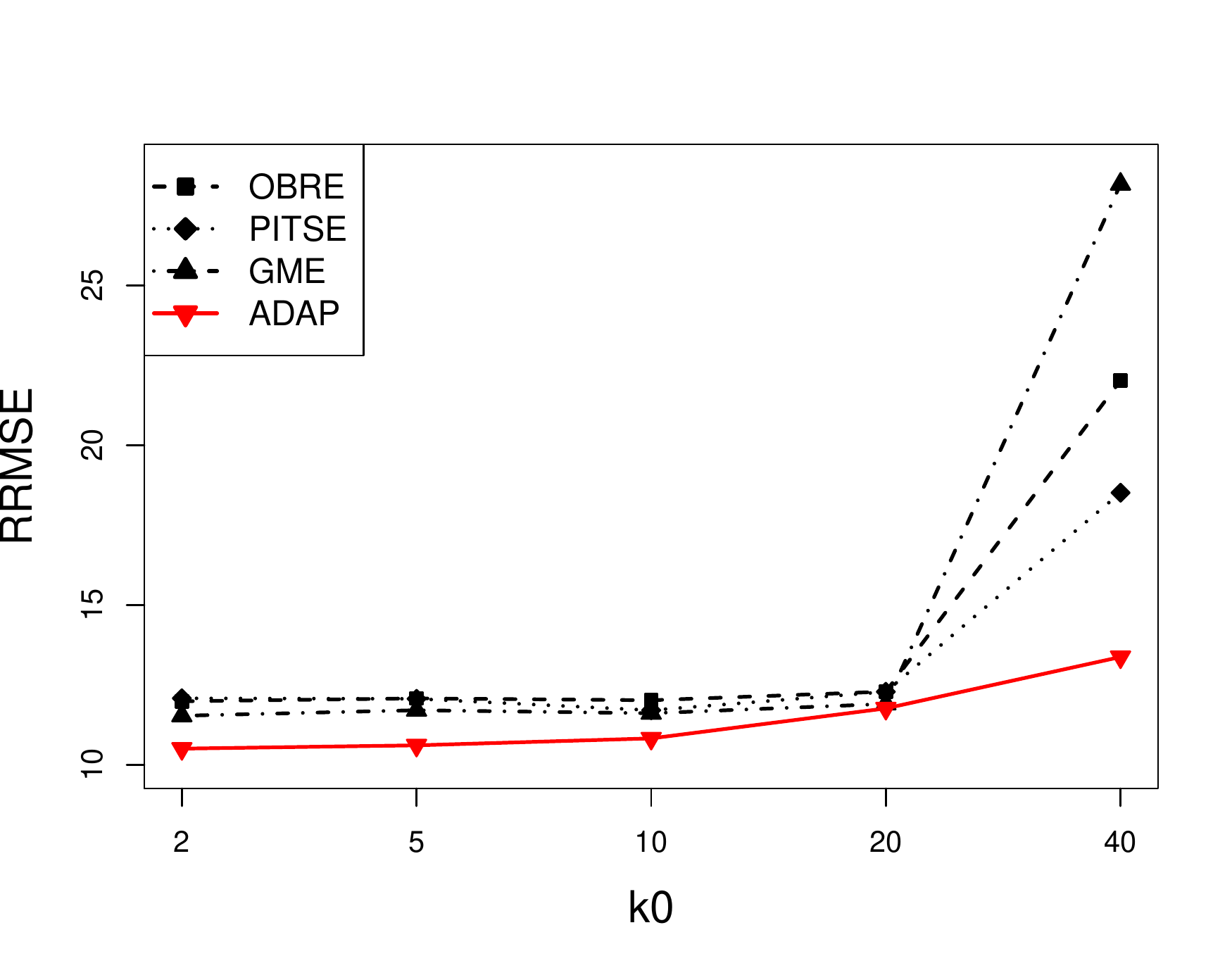}
	\caption{Performance of robust estimators at ARE=78\%. Top left and right correspond to RB and RRMSE values.}
	\label{fig:adap-78}
\end{figure}
In this section, we describe the superior performance of the adaptive trimmed Hill estimator (ADAP), $\widehat{\xi}_{\hat{k}_0,k}$, relative to several well known existing estimators when the degree of contamination is unknown. The performance of these existing robust estimators depends on the choice of parameters, which is directly related to their asymptotic relative efficiency.

For example, the optimal B-robust estimator (OBRE) requires a suitable choice of the parameter $c$ (see \cite{CJS:CJS247}) and the probability integral transform estimator (PITSE) requires a suitable choice of the parameter $t$ (see \cite{pitse})in order to allow for a given degree of robustness. Unless the degree of contamination is pre specified, it is impossible to accurately determine these parameters, which control the degree of robustness. Our estimator, on the other hand is adaptive in nature and automatically picks the trimming parameter, thereby producing a estimator of the tail index which can adapt to potentially unknown degree of contamination of the top order statistics.

We demonstrate the adaptive property of the proposed estimator, ADAP for the Pareto model where the outliers are injected as in \eqref{e:exp-trans-0}. For comparative purposes, we use the three best robust estimators, OBRE, PITSE and GME from \cite{Brzezinski2016} also described in Section \ref{subsec:comp-prev}. The comparison is made in terms of RRMSE and RB values as in \eqref{e:b-mse}. As in Section \ref{subsec:comp-prev}, we calibrate the parameters of the competing estimators by setting the ARE to be 78\% or 94\%.
\vspace{-5mm}

\begin{figure}[H]
	\centering
	\includegraphics[width=0.45\textwidth]{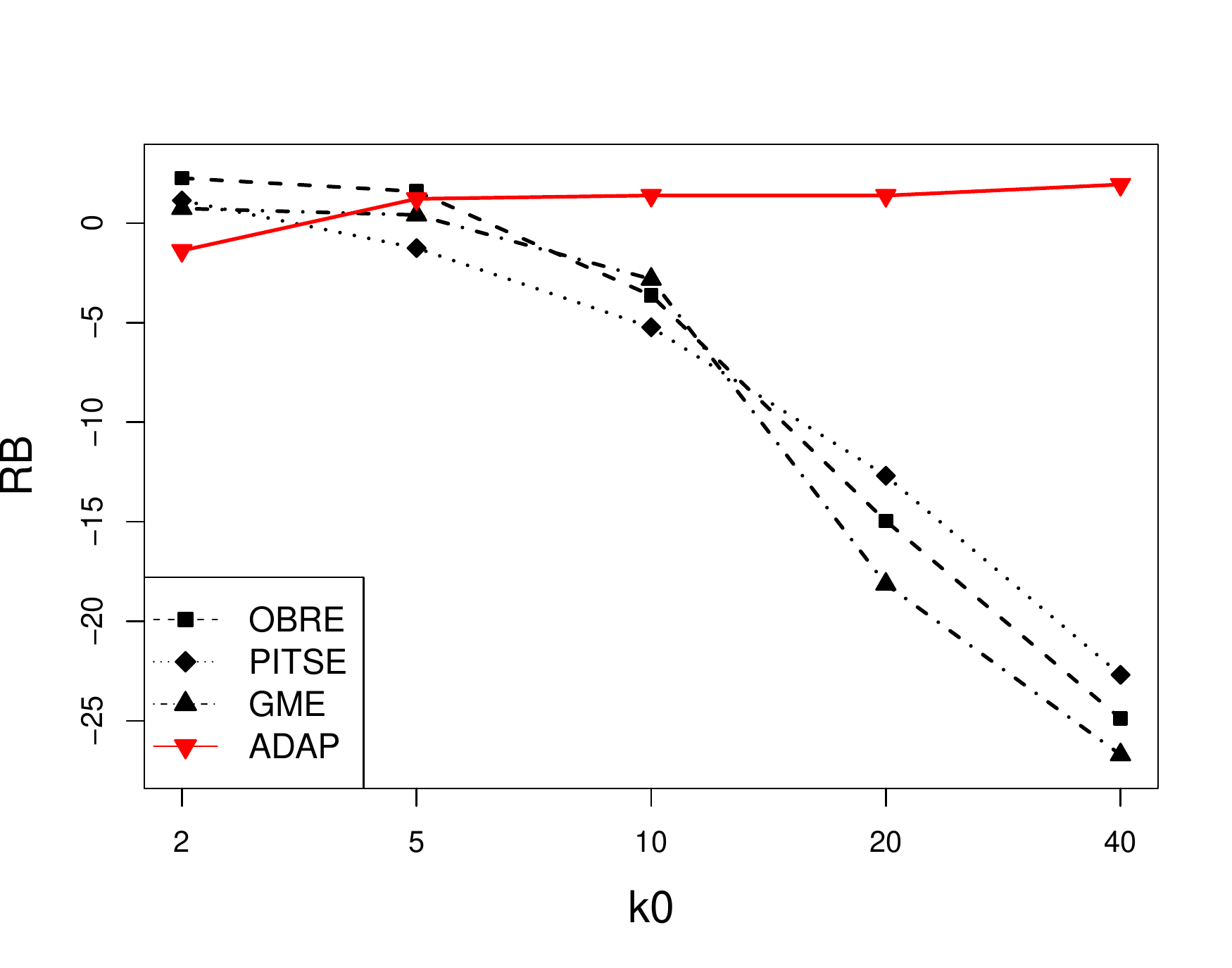}
	\includegraphics[width=0.45\textwidth]{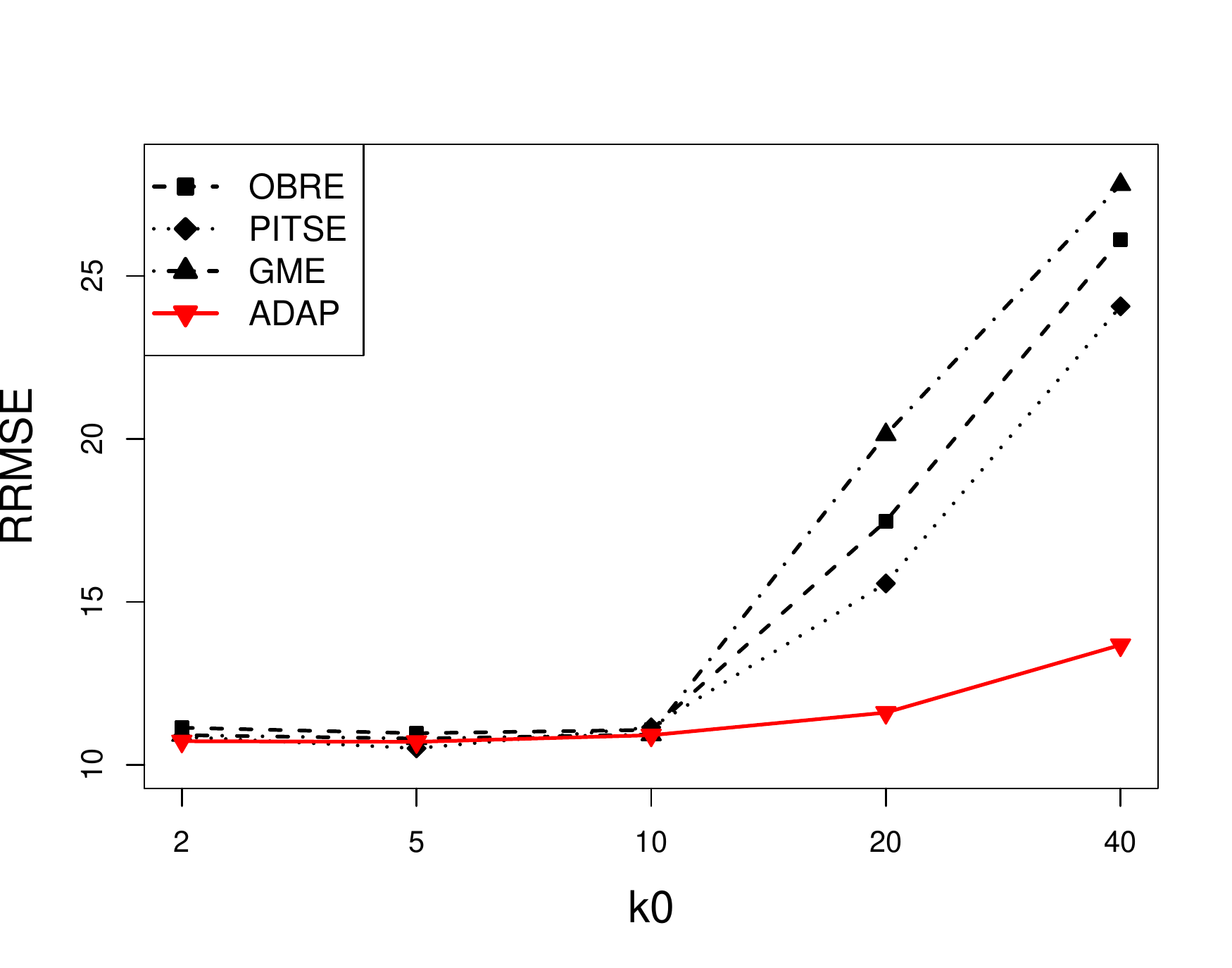}
	\caption{Performance of robust estimators at ARE=94\%. Top left and right correspond to RB and RRMSE values.}
	\label{fig:adap-94}
\end{figure}

Figure \ref{fig:adap-78} demonstrates the performance of ADAP against the three competitors at ARE=78\%. Observe that the competitors fail to adapt to the growing degree of contamination and essentially break down at $k_0/n=40\%$. On the other hand, apart from a mild loss in efficiency, our estimator is resilient to the degree of contamination and adapts itself even to higher values of $k_0/n$. This feature is even more prominent in Figure \ref{fig:adap-94} where the ARE for all estimators is fixed at 94\%. Even at contamination proportion as low as 10\%, ADAP outperforms all the competitors. This is expected since the performance of the competitors sensitive to the choice of ARE.  Large ARE values (94\%) allow for a smaller degree of robustness, hence the poor performance of the OBRE, PITSE and GME even at lower contamination levels. To the best of our knowledge, the remarkable adaptive robustness property inherent to our estimator is not present in any other estimator in the literature.

\section{Discussion}
\label{sec:discussion}

In this paper, we introduced the trimmed Hill estimator for the heavy-tail exponent $\xi$.  We established its finite-sample optimality in the ideal Pareto setting
and its asymptotic normality under a second order regular variation condition. In Section \ref{sec:minmax}, we established a uniform consistency result for the trimmed Hill estimator.  For the Hall
class of distributions, we argued that the trimmed Hill estimator attains the same minimax optimal rate as in the case of no outliers,
provided that $k_0 = o(n^{2\rho/(2\rho +1)})$, where $\rho>0$ is the second order regular variation
exponent.  One open problem is to establish the minimax optimal rate of the trimmed Hill estimator, in the case when 
the rate of contamination $k_0$ exceeds the minimax optimal rate.  

In Section \ref{subsec:k0-k-jt}, we develop a methodology for the joint selection of the parameters $k_0$ and $k$, based on the work of Drees and Kaufman \cite{drees}. 
We formulate an extension of their results when $k_0 = o(n^{2\rho/(2\rho+1)})$.  This leads to a practical method for the joint selection of $k_0$ and $k$.  This method is shown 
to work as well as the original method of Drees and Kaufman even if the top order statistics are contaminated.  As in the case of uncontaminated extremes, however, the main 
challenge is the accurate estimation of the second order exponent $\rho$.  In the future, perhaps other bootstrap-based methods for the joint 
estimation of $k_0$ and $k$ should be explored as in \cite{danielsson}.

Our key methodological contribution is the data--driven selection of the trimming parameter $k_0$ using weighted sequential testing.  
It leads to a robust estimator that adapts to the potentially unknown degree of contamination in the extremes.  This unique feature is not available in many 
other robust estimators, which require the selection of tuning parameters.  As demonstrated in Section \ref{subsec:comp-adap}, the adaptive trimmed Hill estimator has superior performance 
with practically no tuning.  As an added bonus, we obtain a method for the identification of suspect outliers in the extremes of the data, which can be used to perform forensics or detect 
anomalies \cite{amon}. 

Finally, we would like to advocate broadly for using robust methods for the estimation of the tail index.  Our experience
with extensive simulation studies (see e.g., Tables \ref{tab:ARE_exp_0} and \ref{tab:ARE_scl_0}) convinced us that contamination in small proportion of the extreme order statistics 
leads to severe bias in the non-robust estimation methods. Trimming and especially data--adaptive trimming provide good alternatives at the expense of little to no loss in efficiency 
in the case when no contamination is present.

\section{Appendix}
\label{sec:append}
\subsection{Auxiliary Lemmas}

\begin{lemma} 
	\label{lem:u-ind}
	Let $E_j \stackrel{i.i.d}{\sim}{\rm Exp}(1)$, $j=1,2, \cdots, n+1$ be standard exponential random variables. Then, the ${\rm Gamma}(i,1)$ random variables defined as
	\begin{equation}
	\label{e:gam-dist}
	\Gamma_i = \sum_{j=1}^iE_j \hspace{5mm} i=1, \cdots, n+1,
	\end{equation} 
	satisfy 
	\begin{equation}
	\label{e:gam-ind}
	\Big( \frac{\Gamma_1}{\Gamma_{n+1}},\cdots, \frac{\Gamma_n}{\Gamma_{n+1}}\Big)\:\textmd{  and  }\:\Gamma_{n+1}\textmd{ are independent.}\\
	\end{equation}and
	\begin{equation}
	\label{e:gam-u}
	\Big( \frac{\Gamma_1}{\Gamma_{n+1}},\cdots, \frac{\Gamma_n}{\Gamma_{n+1}}\Big)\stackrel{d}{=}( U_{(1,n)},\cdots,U_{(n,n)}) 
	\end{equation}
	
	where $U_{(1,n)}< \cdots< U_{(n,n)}$ are the order statistics of $n$ i.i.d. U(0,1) random variables.
\end{lemma}

For details on the proof see Example 4.6 on page 44 in \cite{MR3025012}. The next result, quoted from page 37 in \cite{Haan}, shall be used throughout the course of the paper to switch 
between order statistics of exponentials and i.i.d.\ exponential random variables.  

\begin{lemma}[R\'enyi, 1953]
	\label{lem:renyi}
	Let $E_1, E_2, \cdots, E_n$ be a sample of $n$ i.i.d. standard exponential random variables and $E_{(1,n)} \leq E_{(2,n)} \leq E_{(n,n)} $ be the order statistics. By R\'enyi's (1953) representation, we have for fixed $k\leq n$,
	\begin{equation}
	\label{e:renyi}
	(E_{(1,n)},\cdots, E_{(i,n)},\cdots,E_{(k,n)})\stackrel{d}{=}\Big(\frac{E_1^*}{n},\cdots,\sum_{j=1}^i \frac{E^*_j}{n-j+1}, \cdots, \sum_{j=1}^k \frac{E^*_j}{n-j+1}\Big) 
	\end{equation}
	where $E_1^*, \cdots, E_k^*$ are also i.i.d. standard exponentials.
\end{lemma}

\begin{lemma}
	\label{lem:ratio-conv}
	For $\Gamma_m=E_1+E_2+\cdots+E_m$ where the $E_i'$s are i.i.d. standard exponential random variables, for any $\rho$ 
	\begin{eqnarray}
	\label{e:gam-conv-0}
	\sup_{m\geq M}\Big|\Big(\frac{\Gamma_{m}}{m}\Big)^{-\rho}-1\Big| &\stackrel{a.s.}{\longrightarrow}& 0, \hspace{2mm}M \rightarrow \infty\\
	\label{e:gam-conv-1}
	\sup_{m,n\geq M}\Big|\Big(\frac{\Gamma_{m}/m}{\Gamma_{n}/n}\Big)^{-\rho}-1\Big| &\stackrel{a.s.}{\longrightarrow}& 0, \hspace{2mm}M \rightarrow \infty
	\end{eqnarray}
\end{lemma}

\begin{lemma}
	For all $\rho>0$, we have 
	\label{lem:gam-int}
	$$
	\sup_{m\geq M}\Big|\frac{1}{m}\sum_{i=1}^{m}{\Big(\frac{\Gamma_{i+1}}{\Gamma_{m+1}}\Big)}^{\rho}-\frac{1}{1-\rho}\Big|\stackrel{a.s.}{\longrightarrow}0, \hspace{2mm} M \rightarrow \infty
	$$
\end{lemma}
\begin{proof} It is equivalent to show that, as $m\to\infty$,
	\label{lem:sum-conv}
	\begin{equation}
	\label{e:sum-rho-conv}
	\Big|\frac{1}{m}\sum_{i=1}^{m}{\Big(\frac{\Gamma_{i+1}}{\Gamma_{m+1}}\Big)}^{\rho}-\frac{1}{1+\rho}\Big|\stackrel{a.s.}{\longrightarrow}0.
	\end{equation}
	For a fixed $\omega \in \Omega$, let us define the following sequence of functions
	$$f_m(x)=\sum_{i=1}^m {(\Gamma_{i+1}/\Gamma_{m+1})}^{\rho}(\omega){\mathbf{1}}_{(\frac{i-1}{m},\frac{i}{m}]}(x), \hspace{5mm}x>0$$
	Suppose $x \in ((i-1)/m,i/m]$, then 
	\begin{equation}
	\label{e:dct}
	f_m(x)={(\Gamma_{[mx]+1}/\Gamma_{m+1})}^{-\rho}(\omega)=\Big(\frac{[mx]+1}{m}\Big)^{-\rho}\Big(\frac{\Gamma_{[mx]+1}/([mx]+1)}{\Gamma_{m}/m}\Big)^{\rho}(\omega) \rightarrow x^{-\rho}
	\end{equation}
	where the convergence follows from \eqref{e:gam-conv-1}. Moreover since $\Gamma_{[mx]+1}<\Gamma_m$ and $\rho<0$, therefore $|f_m(x)|\leq 1$, for all $x>0$. Thus by dominated convergence theorem,
	\begin{equation}
	\label{e:dct-1}
	\int_{0}^{1}f_m(x)dx=\frac{1}{m}\sum_{i=1}^m {(\Gamma_{i+1}/\Gamma_{m+1})}^{-\rho}(\omega) \rightarrow \int_0^1 x^{-\rho}dx=\frac{1}{1-\rho}
	\end{equation}
	Since \eqref{e:dct} hold for all $\omega \in \Omega$ with $P[\Omega]=1$, so does \eqref{e:dct-1}. This completes the proof.
\end{proof}

\subsection{Proofs for Section \ref{sec:trim-hill}}
\begin{lemma}
	\label{lem:blue} 
	If $E_i$'s $i=1,\cdots,n$ are i.i.d. observations from ${\rm Exp}(\xi)$, the best linear unbiased estimator (BLUE) of $\xi$ based on the order statistics, $E_{(1,n)}<\cdots<E_{(r,n)}$ is given by
	\begin{equation*} 
	\hat{\xi}=\frac{1}{r}\sum_{i=1}^{r-1} E_{(i,n)} +\frac{n-r+1}{r}E_{(r,n)}\end{equation*}
\end{lemma}

\begin{proof}
	Let $\widehat{\xi}=\sum_{i=1}^{r} \gamma_i E_{(i,n)}$ denote the BLUE of $\xi$. By Relation \eqref{e:renyi} in Lemma \ref{lem:renyi}, the BLUE can then be expressed as
	\begin{equation}
	\label{e:z-y-1}
	\hat{\xi}=\sum_{i=1}^{r} \gamma_i\sum_{j=1}^i \frac{E^*_j}{(n-j+1)}=\sum_{j=1}^{r} E^*_j \sum_{i=j}^{r} \frac{\gamma_i}{(n-j+1)}=:\sum_{j=1}^{r} E^*_j \delta_j
	\end{equation} 
	where the $E_j^*$ are i.i.d. from ${\rm Exp}(\xi)$. 
	
	For i.i.d. observations from ${\rm Exp}(\xi)$, the sample mean is the uniformly minimum variance unbiased estimator for $\xi$ (see Lehmann Scheffe Theorem, Theorem 1.11, page 88 in~\cite{lehmann:casella}). Thus $\delta_j=1/r$ yields the required best linear unbiased estimator. 
	
	Using the fact that $\sum_{i=j}^r\gamma_i=\delta_j(n-j+1)=(n-j+1)/r$, we obtain
	$$\gamma_{i}=\begin{cases}
	\frac{n-{r}+1}{r} & \text{ }  i=r\\                                                                                                                                                                                                       \frac{1}{{r}}& \text{ } i<r
	\end{cases}$$
	This completes the proof.
\end{proof}

\begin{lemma}
	\label{lem:g-asy}
	Suppose $g$ is $-\rho$-varying for $\rho\geq 0$ and $Y_{(n-k,n)}$ is the $(k+1)^{th}$ order statistic for $n$ observations from ${\rm Pareto}(1,1)$, then 
	\begin{equation}
	\label{e:g-asy}
	\frac{g(Y_{(n-k,n)})}{g(n/k)}\stackrel{P}{\longrightarrow}1
	\end{equation}
	provided $k\rightarrow \infty$, $n \rightarrow \infty$ and $k/n \rightarrow \infty$.
\end{lemma}

\begin{proof}
	Since $g$ is $-\rho$ varying, $g$ may be expressed as $g(t)=t^{-\rho}l(t)$, for some slowly varying function $l(\cdot)$. Thus, we have
	$$\frac{g(Y_{(n-k,n)})}{g(n/k)}={\Big(\frac{Y_{(n-k,n)}}{n/k}\Big)}^{-\rho}\frac{l(Y_{(n-k,n)})}{l(n/k)}$$
	From \eqref{e:Pareto-order}, we have $Y_{(n-k,n)}\stackrel{d}{=}\Gamma_{n+1}/\Gamma_{k+1}$ and therefore, by weak law of large numbers, we have $Y_{(n-k,n)}/(n/k)\stackrel{P}{\longrightarrow}1$.
	
	Thus to prove \eqref{e:g-asy}, it suffices to show $l(Y_{(n-k,n)})/l(n/k)\stackrel{P}{\longrightarrow}1$. In this direction, observe that for some $\delta>0$, we have
	\begin{eqnarray*}
		\label{e:conv-p}
		P\Bigg[\Big|\frac{l(Y_{(n-k,n)})}{l(n/k)}-1\Big|>\varepsilon\Bigg]&\leq&P\Bigg[\Big|\frac{l(Y_{(n-k,n)})}{l(n/k)}-1\Big|>\varepsilon, \Big|\frac{Y_{(n-k,n)}}{n/k}-1\Big|\leq \delta\Bigg]+P\Bigg[\Big|\frac{Y_{(n-k,n)}}{n/k}-1\Big|>\delta\Bigg]\\
		&\leq& P\Bigg[ \sup_{\lambda \in [1-\delta,1+\delta]}\Big|\frac{l(\lambda n/k)}{l(n/k)}-1\Big|>\varepsilon\Bigg]+P\Bigg[\Big|\frac{Y_{(n-k,n)}}{n/k}-1\Big|>\delta\Bigg]
	\end{eqnarray*}
	The first term on the right hand side  goes to 0 by Theorem 1.5.2 on page 22 in \cite{bingham1989regular}. The second term goes to since $Y_{(n-k,n)}/(n/k)\stackrel{P}{\longrightarrow}1$.	
\end{proof}

\begin{proof}[ {\bf Proof of Proposition \ref{prop:xi-opt}}]
	Observe that $X_i$'s can be alternatively written as
	$$X_i=\sigma U_{i}^{-\xi}, \hspace{3mm}i=1, \cdots, n,$$
	where $U_i$'s are i.i.d. $U(0,1)$. Therefore by Relation \eqref{e:gam-u} in Lemma \ref{lem:u-ind}, we have
	\begin{equation}
	\label{e:Pareto-order}
	(X_{(n,n)}, \cdots, X_{(1,n)})=\sigma(U_{(1,n)}^{-\xi},\cdots,U_{(n,n)}^{-\xi})\stackrel{d}{=}\sigma\Bigg( {\Big(\frac{\Gamma_1}{\Gamma_{n+1}}\Big)}^{-\xi},\cdots, {\Big(\frac{\Gamma_n}{\Gamma_{n+1}}\Big)}^{-\xi}\Bigg)
	\end{equation}
	where $X_{(n,n)}>\cdots>X_{(1,n)}$ are the order statistics for the $X_i$'s. Hence, for all $1\leq k \leq n-1$, we have
	\begin{eqnarray}
	\label{e:log-u}
	\Bigg(\log \Big(\frac{X_{(n,n)}}{X_{(n-k,n)}} \Big),\cdots,\log \Big(\frac{X_{(k)}}{X_{(n-k,n)}} \Big)\Bigg)  &\stackrel{d}{=}& -\xi\Bigg(\log \Big(\frac{\Gamma_1}{\Gamma_{k+1}} \Big),\cdots,\log \Big(\frac{\Gamma_k}{\Gamma_{k+1}} \Big)\Bigg)\\\nonumber
	&\stackrel{d}{=}&-\xi(\log U_{(1,k)}, \cdots, \log U_{(k,k)}),
	\end{eqnarray}
	where the $U_{(i,k)}$'s are the order statistics for a sample of $k$ i.i.d. $U(0,1)$ and the last equality in \eqref{e:log-u} follows from Relation \eqref{e:gam-u} in Lemma \ref{lem:u-ind}. Since negative log transforms of $U(0,1)$ are standard exponentials, one can define $E_{(i,k)}$, $i=1,\cdots,k$ as
	\begin{equation}
	\label{e:log-exp}
	\Bigg(\log \Big(\frac{X_{(n,n)}}{X_{(n-k,n)}} \Big),\cdots,\log \Big(\frac{X_{(k)}}{X_{(n-k,n)}} \Big)\Bigg) =: (E_{(k,k)}, \cdots,  E_{(1,k)})
	\end{equation}
	such that the $E_{(i,k)}$'s are distributed as order statistics of $k$ i.i.d. exponentials with mean $\xi$, henceforth denoted by ${\rm Exp}(\xi)$.  One can thereby simplify $\widehat{\xi}^{\rm trim}_{k_0,k}$ in \eqref{e:xi-trimmed} as
	\begin{equation}
	\label{e:blue-1}
	\widehat{\xi}^{\rm trim}_{k_0,k} =  \sum_{i=k_0+1}^k c_{k_0,k}(n-i+1,n)E_{(k-i+1,k)}=\sum_{i=1}^{k-k_0}\delta_i E_{(i,k)}
	\end{equation}
	where $\delta_i=c_{k_0,k}(k-i+1)$. The optimal choice of weights $\delta_i$'s which produce the best linear unbiased estimator (BLUE) for $\xi$ is obtained using Lemma \ref{lem:blue} as:
	\begin{equation}
	\label{e:blue-2}
	\delta^{\rm opt}_i=
	\begin{cases}
	\frac{1}{k-k_0}\hspace{5mm} i=1,\cdots, k-k_0-1\\
	\frac{k_0+1}{k-k_0}\hspace{5mm} i=k-k_0
	\end{cases}
	\end{equation}
	
	Rewriting $E_{(i,k)}$'s in terms of $X_{(n-i+1,n)}$'s as in \eqref{e:log-exp} completes the proof.
\end{proof}

\begin{proof}[{\bf Proof of Proposition \ref{prop:xi-exp}}]
	From \eqref{e:blue-1} and \eqref{e:blue-2} in Proposition \ref{prop:xi-opt}, we have
	\begin{equation}
	\label{e:xi-jt}
	{\Big\{\widehat{\xi}_{k_0,k},\ k_0=0,\ldots,k-1 \Big\}}
	= {\Big\{\frac{1}{k-k_0}\sum_{i=1}^{k-k_0-1}E_{(i,k)}+\frac{k_0+1}{k-k_0}E_{(k-k_0,k)},\ k_0=0,\ldots,k-1 \Big\}}
	\end{equation}
	Using Relation \eqref{e:renyi} from Lemma \ref{lem:renyi}, for all $k_0=0, 1, \cdots, k-1$, we have
	\begin{equation}
	\label{e:xi-simp}
	\widehat{\xi}_{k_0,k}=\frac{1}{k-k_0}\sum_{i=1}^{k-k_0-1}\sum_{j=1}^{i}\frac{E^*_j}{(k-j+1)}+\frac{k_0+1}{k-k_0}\sum_{j=1}^{k-k_0}\frac{E^*_j}{(k-j+1)}\\
	\end{equation}
	Interchanging the order of summation in the first term in the right hand side of \eqref{e:xi-simp}, for $k_0=0,1, \cdots, k-1$, we obtain
	\begin{eqnarray*}
		\widehat{\xi}_{k_0,k}
		&=&\sum_{j=1}^{k-k_0-1}\frac{E^*_j}{k-j+1}\sum_{i=j}^{k-k_0-1}\frac{1}{k-k_0}+\frac{k_0+1}{k-k_0}\sum_{j=1}^{k-k_0}\frac{E^*_j}{(k-j+1)}\\
		&=&\sum_{j=1}^{k-k_0-1}\frac{E^*_j}{k-j+1}\left(\sum_{i=j}^{k-k_0-1}\frac{1}{k-k_0}+\frac{k_0+1}{k-k_0}\right)+\frac{E^*_{k-k_0}}{k-k_0}\\
		&=&\sum_{j=1}^{k-k_0-1}\frac{E^*_j}{k-j+1}\frac{(k-j+1)}{k-k_0}+\frac{E^*_{k-k_0}}{k-k_0}\\
		&=&\frac{1}{k-k_0}\sum_{j=1}^{k-k_0}E^*_j, \hspace{5mm} 
	\end{eqnarray*}
	Since $E_j^*$, $j=1, \cdots, k-k_0$ are rescaled i.i.d. standard exponentials, Relation \eqref{e:xi-jt-big} follows. 
	
	The covariance structure in \eqref{e:covar} readily follows from \eqref{e:xi-jt-big} and the fact that 
	$${\rm Cov}\Big(\frac{\Gamma_i}{i}, \frac{\Gamma_j}{j}\Big)=\frac{i \wedge j}{ij}=\frac{1}{i \vee j}, \hspace{5mm} i, j=0,1,\cdots, k$$
	where $\vee$ denotes the max operator. This completes the proof.
\end{proof}

\subsection{Proofs for Section \ref{sec:optimal}} \label{sec:proofs-sec3}

\begin{proof}[{\bf Proof of Theorem \ref{thm:umvue-p}}]
	Suppose for the moment $\sigma$ is known and consider the class of statistics:
	$$	{\cal{U}}^{\sigma}_{k_0}=\left\{T=T(X_{(n-k_0,n)},\cdots,X_{(1,n)}):\: \mathbb{E}(T)=\xi,\: X_1, \cdots, X_n \stackrel{i.i.d.}{\sim} {\rm Pareto}(\sigma,{\xi})\right\}.$$
	Since $\sigma$ is no longer a parameter, every statistic in ${\cal{U}}_{k_0}^{\sigma}$ can be equivalently written as a function of $\log(X_{(n-i+1,n)}/\sigma)$, $i=k_0+1, \cdots, n$.
	Therefore, the set of random variables in ${\cal{U}}_{k_0}^{\sigma}$ equals 
	\begin{equation*}
	\label{e:u-sigma-0}
	{\cal{U}}_{k_0}^{\sigma}=\left\{S=S\left(\log\Big(\frac{X_{(n-k_0,n)}}{\sigma}\Big),\cdots,\log\Big(\frac{X_{(1,n)}}{\sigma}\Big)\right): \: \mathbb{E} (S)=\xi, \: X_1, \cdots, X_n \stackrel{i.i.d.}{\sim} {\rm Pareto}(\sigma,{\xi})\right\}.
	\end{equation*}
	Since $X_i$'s follow ${\rm Pareto(\sigma,\xi)}$, we have $\log(X_i/\sigma) \sim {\rm Exp}(\xi)$ and therefore 
	\begin{equation*}
	\left(\log\Big(\frac{X_{(n-k_0,n)}}{\sigma}\Big),\cdots,\log\Big(\frac{X_{(1,n)}}{\sigma}\Big) \right) \stackrel{d}{=}\left(E_{(n-k_0,n)}, \cdots,E_{(1,n)}\right),
	\end{equation*} 
	where $E_{(1,n)} \leq \cdots \leq E_{(n,n)}$ are the order statistics of $n$ i.i.d. observations from ${\rm Exp}(\xi)$. Therefore
	\begin{equation}
	\label{e:u-sigma-1}
	{\cal{U}}_{k_0}^{\sigma}\stackrel{d}{=}\left\{S=S(E_{(n-k_0,n)}, \cdots,E_{(1,n)}):\: \mathbb{E}(S)=\xi, \:E_1, \cdots, E_n \stackrel{i.i.d.}{\sim} {\rm Exp}(\xi)\right\},
	\end{equation}
	where we observe that the distribution of the $E_i$'s does not depend on $\sigma$.
	
	Using Relation \eqref{e:renyi} from Lemma \ref{lem:renyi}, we have
	\begin{eqnarray*}
		S(E_{(n-k_0,n)}, \cdots,E_{(1,n)})&=&S\Big(\sum_{j=1}^{n-k_0} \frac{E^*_j}{n-j+1}, \cdots, \sum_{j=1}^{n-k} \frac{E^*_j}{n-j+1}\Big)\\
		&=&R(E^*_1, \cdots,E^*_{n-k_0})	
	\end{eqnarray*}
	Using this on the right hand side of \eqref{e:u-sigma-1}, we get
	\begin{equation}
	\label{e:umvue-1}
	{\cal{U}}_{k_0}^{\sigma}
	\stackrel{d}{=}{\cal{V}}_{k_0}:=\left\{R=R(E^*_1, \cdots,E^*_{n-k_0}): \: \mathbb{E}(R)=\xi,\: E^*_{1}, \cdots, E^*_{n-k_0} \stackrel{i.i.d.}{\sim} {\rm Exp}(\xi)\right\}.
	\end{equation}		
	where the first equality is in the sense of finite dimensional distributions. 
	
	Therefore, $L=\inf_{T \in {\cal{U}}_{k_0}^{\sigma}} {\rm Var}(T)=\inf_{R \in {\cal{V}}_{k_0}} {\rm Var}(R)$. The quantity $L$ can be easily obtained as
	\begin{equation}L={\rm Var}({\overline{E}}^*_{n-k_0})=\frac{\xi^2}{n-k_0}\end{equation}
	since the sample mean, ${\overline{E}}^*_{n-k_0}=\sum_{i=1}^{n-k_0} E^*_i/(n-k_0)$ is uniformly minimum variance estimator (UMVUE), for $\xi$ among the class described by ${\cal{V}}_{k_0}$. This follows from the fact that ${\overline{E}}^*_{n-k_0}$ is an unbiased and complete sufficient statistic for $\xi$ (see Lehmann Scheffe Theorem, Theorem 1.11, page 88 in~\cite{lehmann:casella}).
	
	To complete the proof, observe that every statistic, $T$ in ${\cal{U}}_{k_0}$ is an unbiased estimator of $\xi$ for any arbitrary choice of $\sigma$. This implies that $T \in {\cal{U}}_{k_0}^{\sigma}$ and therefore $L\leq\inf_{T \in {\cal{U}}_{k_0}} {\rm{Var}}(T)$, which yields the lower bound in \eqref{e:opt-var}.
	
	For the upper bound in \eqref{e:opt-var}, we observe that $\widehat{\xi}_{k_0,n-1} \in {\cal{U}}_{k_0}$, which in view of Proposition \ref{prop:xi-exp} implies $$\inf_{T \in {\cal{U}}_{k_0}} {\rm{Var}}(T) \leq {\rm Var}(\widehat{\xi}_{k_0,n-1})= \frac{\xi^2}{n-k_0-1}.$$
	This completes the proof.
\end{proof}

We shall next present the proof of Theorem \ref{prop:E-conv}. To begin with we state the following three lemmas which shall be used as a part of the proof.

\begin{lemma}
	\label{lem:r-def}
	\begin{equation}
	\label{e:R-conv-2}
	\max_{0 \leq k_0 <k}\Big|\frac{k-k_0}{kg(Y_{(n-k,n)})}S_{k_0,k}+\frac{c}{1+\rho}\left(\frac{k_0}{k}\right)^{1+\rho}-\frac{c}{1+\rho}\Big|\stackrel{P}{\longrightarrow}0.
	\end{equation}
	where $S_{k_0,k}$ is defined as
	\begin{equation}
	\label{e:s-def}
	S_{k_0,k}:=\frac{cg(Y_{(n-k,n)})}{k-k_0}\Big((k_0+1) \int_{1}^{Y_{(n-k_0,n)}/Y_{(n-k,n)}} \nu^{-\rho-1} d\nu +\sum_{i=k_0+2}^{k} \int_{1}^{Y_{(n-i+1,n)}/Y_{(n-k,n)}} \nu^{-\rho-1} d\nu\Big).
	\end{equation}
	where $Y_i$'s are n i.i.d observations from Pareto(1,1)
\end{lemma}

\begin{proof}
	The proof of \eqref{e:R-conv-2} involves two cases:  $\rho>0$ and $\rho=0$.\vspace{2mm}
	
	\noindent\textbf{Case $\rho>0$:} Using the expression of $S_{k_0,k}$ in \eqref{e:s-def}, we get
	
	\begin{eqnarray}\label{e:lem:r-def-1}
	\frac{k-k_0}{kg(Y_{(n-k,n)})} S_{k_0,k} &=&-\frac{c}{k\rho}\Bigg((k_0+1)\Big(\frac{Y_{(n-k_0,n)}}{Y_{(n-k,n)}}\Big)^{-\rho}+\sum_{i=k_0+2}^{k}\Big(\frac{Y_{(n-i+1,n)}}{Y_{(n-k,n)}}\Big)^{-\rho}-k\Bigg) \nonumber\\
	&=&\frac{c}{k\rho}\sum_{i=1}^{k_0}
	\left\{\Big(\frac{Y_{(n-i+1,n)}}{Y_{(n-k,n)}}\Big)^{-\rho}-\Big(\frac{Y_{(n-k_0,n)}}{Y_{(n-k,n)}}\Big)^{-\rho}\right\} \nonumber\\
	& & -\frac{c}{k\rho}\sum_{i=1}^{k}\left\{\Big(\frac{Y_{(n-i+1,n)}}{Y_{(n-k,n)}}\Big)^{-\rho}-1\right\}
	\end{eqnarray}
	Expressing the order statistics of Pareto in terms of Gamma random variables as in \eqref{e:Pareto-order}, we get
	\small
	\begin{equation*}
	\frac{k-k_0}{kg(Y_{(n-k,n)})} S_{k_0,k}+c_{1+\rho}{\left(\frac{k_0}{k}\right)}^{1+\rho}\stackrel{d}{=}
	\underbrace{c_{1+\rho}{\left(\frac{k_0}{k}\right)}^{1+\rho}+\frac{c_\rho}{k}\sum_{i=1}^{k_0}\Big\{\Big(\frac{\Gamma_{i+1}}{\Gamma_{k+1}}\Big)^{\rho}-\Big(\frac{\Gamma_{k_0+1}}{\Gamma_{k+1}}\Big)^{\rho}\Big\}}_{B_{k_0,k}}-\underbrace{\frac{c_\rho}{k}\sum_{i=1}^{k}\Big\{\Big(\frac{\Gamma_{i+1}}{\Gamma_{k+1}}\Big)^{\rho}-1\Big\}}_{A_k}
	\end{equation*}
	\normalsize with $c_{t}=c/t$. 
	
	\noindent To prove \eqref{e:R-conv-2}, we first show that $\max_{0 \leq k_0 <k}|A_k+c/(1+\rho)| \stackrel{a.s.}{\longrightarrow}0.$ For this \eqref{e:sum-rho-conv}, we have $|(1/k)\sum_{i=1}^{k}(\Gamma_{i+1}/\Gamma_{k+1})^{\rho}-1/(1+\rho)| \stackrel{a.s.}{\longrightarrow} 0 $. Therefore for any $\omega \in \Omega$ with $P[\Omega]=1$,  
	$$\Big|A_k(\omega)+\frac{c}{1+\rho}\Big|=\Big|\frac{c}{\rho k}\sum_{i=1}^k\Big(\frac{\Gamma_{i+1}}{\Gamma_{k+1}}\Big)^{\rho}(\omega)-\frac{c}{\rho}+\frac{c}{1+\rho}\Big|=\Big|\frac{c_{\rho}}{k}\sum_{i=1}^k\Big(\frac{\Gamma_{i+1}}{\Gamma_{k+1}}\Big)^{\rho}(\omega)-\frac{c_\rho}{1+\rho}\Big|{\rightarrow} 0$$
	We next show that $\max_{0 \leq k_0 <k}B_{k_0,k} \stackrel{a.s.}{\longrightarrow}0$.  For this observe that for any $\omega \in \Omega$,
	\begin{eqnarray}
	\label{e:b-1}
	\max_{0 \leq k_0 <M}B_{k_0,k}(\omega)&\leq& \max_{0 \leq k_0 <M} \left\{c_{1+\rho}{\left(\frac{k_0}{k}\right)}^{1+\rho}+\frac{c_\rho}{k}\sum_{i=1}^{k_0}\Big|\Big(\frac{\Gamma_{i+1}}{\Gamma_{k+1}}\Big)^{\rho}(\omega)-\Big(\frac{\Gamma_{k_0+1}}{\Gamma_{k+1}}\Big)^{\rho}(\omega)\Big|\right\}\\\nonumber
	&\leq&  \max_{0 \leq k_0 <M}\left\{c_{1+\rho}{\left(\frac{k_0}{k}\right)}^{1+\rho}+c_{\rho}\frac{2k_0}{k}\right\}\hspace{2mm}(\textmd{since ${(\Gamma_i /\Gamma_{k+1})}^{\rho} \leq 1$, $1\leq i\leq k$,  $\rho>0$})\\\nonumber
	&\leq& \frac{c^{1+\rho}M^{1+\rho}+2c_\rho M}{k} =\frac{B_{0M}}{k}
	\end{eqnarray}
	Additionally,
	\begin{eqnarray}
	\label{e:b-2}
	\max_{M \leq k_0 <k} B_{k_0,k}(\omega)&\leq& \max_{M \leq k_0 <k}  \left|c_{1+\rho}{\left(\frac{k_0}{k}\right)}^{1+\rho}+\frac{c_\rho}{k}\sum_{i=1}^{k_0}\Big\{\Big(\frac{\Gamma_{i+1}}{\Gamma_{k+1}}\Big)^{\rho}(\omega)-\Big(\frac{\Gamma_{k_0+1}}{\Gamma_{k+1}}\Big)^{\rho}(\omega)\Big\}\right|\\\nonumber
	&\leq&  \max_{M \leq k_0 <k}{\left(\frac{k_0}{k}\right)}^{1+\rho}\left|\frac{c_{1+\rho}}{c_\rho}+{\left(\frac{k_0}{k}\right)}^{\rho}\frac{1}{k_0}\sum_{i=1}^{k_0}\Big\{\Big(\frac{\Gamma_{i+1}}{\Gamma_{k+1}}\Big)^{\rho}(\omega)-\Big(\frac{\Gamma_{k_0+1}}{\Gamma_{k+1}}\Big)^{\rho}(\omega)\Big\}\right|\\\nonumber
	&\leq & \max_{M \leq k_0 <k}\Bigg|\frac{\rho}{1+\rho}+\Big(\frac{\Gamma_{k_0+1}/k_0}{\Gamma_{k+1}/k}\Big)^{\rho}(\omega)\Big\{\underbrace{\frac{1}{k_0}\sum_{i=1}^{k_0}\Big(\frac{\Gamma_{i+1}}{\Gamma_{k_0+1}}\Big)^{\rho}(\omega)}_{C_{k_0}(\omega)}-1\Big\}\Bigg|\\\nonumber
	&=&\max_{M \leq k_0 <k}\left|\frac{\rho}{1+\rho}+(C_{k_0}(\omega)-1)+(C_{k_0}(\omega)-1)\Big\{\Big(\frac{\Gamma_{k_0+1}/k_0}{\Gamma_{k+1}/k}\Big)^{\rho}-1\Big\}\right|\nonumber
	\end{eqnarray}

	Since $\Gamma_{i+1}<\Gamma_{k_0+1}$ and $\rho>0$, thereby $|C_{k_0}|<1$. This allows us to simplify \eqref{e:b-2} as
	\begin{eqnarray*}
		\max_{M \leq k_0 <k} B_{k_0,k}(\omega)&\leq&\underbrace{ \sup_{M \leq k_0 }\left|C_{k_0}(\omega)-\frac{1}{1+\rho}\right|}_{B_{1M}(\omega)}+2\underbrace{\sup_{M \leq k_0,k}\left|\Big(\frac{\Gamma_{k_0+1}/k_0}{\Gamma_{k+1}/k}\Big)^{\rho}(\omega)-1\right|}_{B_{2M}(\omega)}
	\end{eqnarray*}
	Thus $$\max_{0 \leq k_0 <k}B_{k_0,k}(\omega) \leq \frac{B_{0M}}{k}+B_{1M}(\omega)+B_{2M}(\omega).$$
	Taking $\limsup$ w.r.t to $k$ on both sides we get
	\begin{equation}
	\label{e:k-sup}
	\limsup_{k \rightarrow \infty}\max_{0 \leq k_0 <k}B_{k_0,k}(\omega)\leq B_{1M}(\omega)+B_{2M}(\omega)
	\end{equation}
	Using Lemmas  \ref{lem:sum-conv} and \ref{lem:ratio-conv} shows that $B_{1M}(\omega) \rightarrow 0$  and $B_{2M}(\omega) \rightarrow 0$ for all $\omega \in \Omega$ with $P[\Omega]=1$.
	
	Thus taking  $\limsup$ w.r.t $M$ on both sides of \eqref{e:k-sup} completes the proof for $\rho<0$.
	\vspace{2mm}
	
	\noindent\textbf{Case $\rho=0$:}  Using the expression of $S_{k_0,k}$ in \eqref{e:s-def}, we get
	\begin{eqnarray}\label{e:rho=0-case}
	\frac{k-k_0}{kg(Y_{(n-k,n)})} S_{k_0,k}+\frac{ck_0}{k} &=&\frac{c}{k}\Bigg((k_0+1)\log \Big(\frac{Y_{(n-k_0,n)}}{Y_{(n-k,n)}}\Big)+\sum_{i=k_0+2}^{k}\log\Big(\frac{Y_{(n-i+1,n)}}{Y_{(n-k,n)}}\Big)\Bigg)+\frac{ck_0}{k}\nonumber\\
	&\stackrel{d}{=}&\frac{c(k-k_0)}{k}\widehat{\xi}^{**}_{k_0,k}+\frac{ck_0}{k}\nonumber\\
	&\stackrel{d}{=}&c\Big(\frac{\Gamma_{k-k_0}}{k}-\frac{k-k_0}{k}+1\Big)
	\end{eqnarray}
	where $\widehat{\xi}^{**}_{k_0,k}$ is the trimmed Hill estimator in \eqref{e:xi-opt} with $X_i$'s replaced by the i.i.d. ${\rm Pareto}(1,1)$.
	
	Thus to prove \eqref{e:R-conv-2}, it suffices to show $\max_{0\leq k_0<k}|\Gamma_{k-k_0}-(k-k_0)|/k\stackrel{a.s.}{\longrightarrow}0$. For every $\omega$ in $\Omega$ with $P[\Omega]=1$, we have
	\begin{eqnarray}
	\label{e:max-gam}
	\max_{0\leq k_0<k}\frac{|\Gamma_{k-k_0}(\omega)-(k-k_0)|}{k}&=&\max_{0\leq k_0<k}\frac{(k-k_0)}{k}\Big|\frac{\Gamma_{k-k_0}}{k-k_0}(\omega)-1\Big|\\\nonumber
	&\leq&\frac{M}{k}\underbrace{\max_{0\leq k-k_0<M}\Big|\frac{\Gamma_{k-k_0}}{k-k_0}(\omega)-1\Big|}_{F_{1M}(\omega)}+\underbrace{\sup_{k-k_0\geq M}\Big|\frac{\Gamma_{k-k_0}}{k-k_0}(\omega)-1\Big|}_{F_{2M}(\omega)}
	\end{eqnarray} 
	Observe that by the SLLN, $|\Gamma_n/n-1| \stackrel{a.s.}{\longrightarrow}0$. Therefore  $\sup_{n}|\Gamma_n(\omega)/n-1|$ is bounded for all $\omega \in \Omega$ with $P[\Omega]=1$. This implies $F_{1M}(\omega) \leq  \sup_{n}|\Gamma_n(\omega)/n-1|$ is bounded. 
	
	Thus taking $\limsup$ with respect to $k$ on both sides of \eqref{e:max-gam} we get 
	$$\limsup_{k \rightarrow \infty}\max_{0\leq k_0<k}\frac{|\Gamma_{k-k_0}(\omega)-(k-k_0)|}{k}\leq F_{2M}(\omega)$$
	Taking $\lim_{M\rightarrow \infty}$ on both sides and using \eqref{e:gam-conv-1}, the proof follows.
\end{proof}

\begin{lemma}
	\label{lem:s-def}
	Assumption \eqref{e:SR2} imply
	\begin{equation}
	\label{e:rem-1}
	\max_{0\leq k_0 \leq k}\Big(	\frac{k-k_0}{kg(Y_{(n-k,n)})}|R_{k_0,k}-S_{k_0,k}|\Big)\stackrel{P}{\longrightarrow}0
	\end{equation}
	where $R_{k_0,k}$ and $S_{k_0,k}$ are defined in \eqref{e:r-def} and \eqref{e:s-def}, respectively.
\end{lemma}

\begin{proof}
	The proof of \eqref{e:rem-1} involves two cases:  $\rho>0$ and $\rho=0$.\vspace{2mm}
	
	\noindent\textbf{Case $\rho>0$:} Since $Y_{(n-i+1,n)}/Y_{(n-k,n)}>1$, $i=1,\cdots, k$,  over the event $\{Y_{(n-k,n)}>t_\varepsilon\}$, by \eqref{e:SR2}:
	\begin{eqnarray*}
		(k-k_0)|R_{k_0,k}-S_{k_0,k}|&\leq &(k_0+1)\left|\log \frac{L(Y_{(n-k_0,n)})}{L(Y_{(n-k,n)})}-cg(Y_{(n-k,n)} \int_{1}^{Y_{(n-k_0,n)}/Y_{(n-k,n)}} \nu^{-\rho-1} d\nu \right|\\
		&+&\sum_{i=k_0+2}^{k}\Bigg|\log \frac{L(Y_{(n-i+1,n)})}{L(Y_{(n-k,n)})}-cg(Y_{(n-k,n)} \int_{1}^{Y_{(n-i+1,n)}/Y_{(n-k,n)}} \nu^{-\rho-1} d\nu \Bigg|\\	
		&\leq&(k_0+1)g(Y_{(n-k,n)})\varepsilon+\sum_{i=k_0+2}^{k}g(Y_{(n-k,n)})\varepsilon =g(Y_{(n-k,n)})k\varepsilon.\vspace{-5mm}
	\end{eqnarray*} 
	Therefore over the event $\{Y_{(n-k,n)}>t_\varepsilon\}$  
	\begin{equation}\label{e:lem:s-def-1}
	\max_{0\leq k_0 \leq k}\Big(	\frac{k-k_0}{kg(Y_{(n-k,n)})}|R_{k_0,k}-S_{k_0,k}|\Big) \leq \varepsilon.
	\end{equation}
	From \eqref{e:Pareto-order} we get $Y_{(n-k,n)}\stackrel{d}{=}(\Gamma_{k+1}/\Gamma_{n+1})^{-1}$. By Lemma \ref{lem:ratio-conv}, we have 
	$$Y_{(n-k,n)}\stackrel{d}{=}\frac{n}{k}\Big(\frac{\Gamma_{k+1}/k}{\Gamma_{n+1}/n}\Big)^{-1} \stackrel{P}{\longrightarrow}\infty$$
	which implies 	$P[Y_{(n-k,n)} >t_{\varepsilon}]\rightarrow 1$ and hence completes the proof. \vspace{2mm}
	
	\noindent\textbf{Case $\rho=0$:} As in the previous case, over the event $\{Y_{(n-k,n)}>t_\varepsilon\}$, by \eqref{e:SR2} we have	\begin{eqnarray}
	\label{e:diff-0-1}
	(k-k_0)|R_{k_0,k}-S_{k_0,k}|&=&(k_0+1)\left|\log \frac{L(Y_{(n-k_0,n)})}{L(Y_{(n-k,n)})}-cg(Y_{(n-k,n)} \int_{1}^{Y_{(n-k_0,n)}/Y_{(n-k,n)}} \frac{d\nu}{\nu}  \right|\\\nonumber
	&+&\sum_{i=k_0+2}^{k}\left|\log \frac{L(Y_{(n-i+1,n)})}{L(Y_{(n-k,n)})}-cg(Y_{(n-k,n)} \int_{1}^{Y_{(n-i+1,n)}/Y_{(n-k,n)}} \frac{d\nu}{\nu} \right|\\\nonumber
	&\leq&\varepsilon\Bigg((k_0+1)g(Y_{(n-k,n)})\Big(\frac{Y_{(n-k_0,n)}}{Y_{(n-k,n)}}\Big)^{\varepsilon}+\sum_{i=k_0+2}^{k}g(Y_{(n-k,n)})\Big( \frac{Y_{(n-i+1,n)}}{Y_{(n-k,n)}}\Big)^{\varepsilon}\Bigg)\nonumber
	\end{eqnarray}
	Since $Y_{(n-i+1,n)}\geq Y_{(n-k_0,n)}$ for $i=1, \cdots, k_0+1$,  we further obtain
	\begin{equation}\label{e:rho=0-second-instance}
	\max_{0\leq k_0 \leq k}\Big(	\frac{(k-k_0)}{kg(Y_{(n-k,n)})}|R_{k_0,k}-S_{k_0,k}|\Big)\leq \frac{\varepsilon}{k} \sum_{i=1}^{k}\Big(\frac{Y_{(n-i+1,n)}}{Y_{(n-k,n)}}\Big)^{\varepsilon}\leq 2\varepsilon
	\end{equation}
	over the events $\{Y_{(n-k,n)}>t_\varepsilon\}$  and $\{(1/k) \sum_{i=1}^{k}(Y_{(n-i+1,n)}/Y_{(n-k,n)})^{\varepsilon}<2\}$. 
	
	For $\{(1/k) \sum_{i=1}^{k}(Y_{(n-i+1,n)}/Y_{(n-k,n)})^{\varepsilon}<2\}$, from \eqref{e:Pareto-order}, we observe that  
	$$\frac{1}{k} \sum_{i=1}^{k}\Big(\frac{Y_{(n-i+1,n)}}{Y_{(n-k,n)}}\Big)^{\varepsilon}\stackrel{d}{=}\frac{1}{k}\sum_{i=1}^{k}\Big(\frac{\Gamma_{i+1}}{\Gamma_{k+1}}\Big)^{-\varepsilon}=\frac{1}{k}\sum_{i=1}^{k}U_{i,k}^{-\varepsilon}\stackrel{P}{\longrightarrow}\frac{1}{1-\varepsilon}$$
	where the last convergence follows from weak law of large numbers.
	
	Thus $P[(1/k) \sum_{i=1}^{k}(Y_{(n-i+1,n)}/Y_{(n-k,n)})^{\varepsilon}<2] \rightarrow 1$ as long as $\varepsilon<0.5$. Since we already proved that $P[Y_{(n-k,n)} >t_{\varepsilon}]\rightarrow 1$, the proof for the case $\rho<0$ follows.	
\end{proof}

\begin{proof}[{\bf Proof of Theorem \ref{prop:E-conv}}]
	Using \eqref{e:tail-3}, we can rewrite \eqref{e:E-conv} as 
	\begin{equation}
	\label{e:E-conv-mod}
	k^\delta \max_{0\leq k_0<h(k)}\Big|R_{k_0,k}-\frac{k^{-\delta}cA}{(1+\rho)}\Big| \stackrel{P}{\longrightarrow}0
	\end{equation}
	In this direction, observe that
	\begin{equation*}
	\label{e:rs-decom}
	\Big|R_{k_0,k}-\frac{k^{-\delta}cA}{(1+\rho)}\Big|\leq|R_{k_0,k}-S_{k_0,k}|+\Big|S_{k_0,k}-\frac{k^{-\delta}cA}{ ({1+\rho})}\Big|
	\end{equation*}
	where $S_{k_0,k}$ is defined in \eqref{e:s-def}.
	
	\noindent To prove \eqref{e:E-conv-mod}, we first show that $k^\delta \max_{0\leq k_0<h(k)}|R_{k_0,k}-S_{k_0,k}|\stackrel{P}{\longrightarrow}0$. In this direction, we have
	\begin{eqnarray*}
		k^\delta \max_{0\leq k_0<h(k)}|R_{k_0,k}-S_{k_0,k}|&=&k^\delta \max_{0\leq k_0<h(k)}\frac{kg(Y_{n-k,n})}{k-k_0}\Big(\frac{k-k_0}{kg(Y_{(n-k,n)})}|R_{k_0,k}-S_{k_0,k}|\Big)\\\nonumber
		&\leq&\frac{\Delta_{1k}}{1-h(k)/k}\underbrace{\max_{0\leq k_0<h(k)}\Big(\frac{k-k_0}{kg(Y_{(n-k,n)})}|R_{k_0,k}-S_{k_0,k}|\Big)}_{\Delta_{2k}}
	\end{eqnarray*}
	where   $\Delta_{2k} \stackrel{P}{\longrightarrow}0$ by Lemma \ref{lem:s-def}. Since $h(k)=o(k)$, $1-h(k)/k\rightarrow 1$ and $\Delta_{1k}=k^\delta g(Y_{(n-k,n)})$
	\begin{equation}
	\label{e:del-1}
	\Delta_{1k}=k^\delta g(n/k)\frac{g(Y_{(n-k,n)})}{g(n/k)}\stackrel{P}{\longrightarrow}A
	\end{equation}
	where \eqref{e:del-1} follows from assumption \eqref{e:A-def} and Lemma \ref{lem:g-asy}.
	
	\noindent Towards the proof of \eqref{e:E-conv-mod}, we finally show that $k^\delta \max_{0\leq k_0<h(k)}|S_{k_0,k}-(k^{-\delta}cA)/ ({1+\rho})|\stackrel{P}{\longrightarrow}0$. In this direction, we have
	\begin{eqnarray*}
		k^\delta \max_{0\leq k_0<h(k)}\Big|S_{k_0,k}-\frac{k^{-\delta}cA}{(1+\rho)}\Big|&=&k^\delta \max_{0\leq k_0<h(k)}\frac{kg(Y_{n-k,n})}{k-k_0}\Big|\frac{k-k_0}{kg(Y_{(n-k,n)})}S_{k_0,k}-\frac{cA(k-k_0)}{k(1+\rho)\Delta_{1k}}\Big|\\\nonumber
		&\leq&\frac{\Delta_{1k}}{1-h(k)/k}\underbrace{\max_{0\leq k_0<h(k)}\Big|\frac{k-k_0}{kg(Y_{(n-k,n)})}S_{k_0,k}-\frac{cA(k-k_0)}{k(1+\rho)\Delta_{1k}}\Big|}_{\Delta_{3k}}
	\end{eqnarray*}
	where $\Delta_{1k} \stackrel{P}{\longrightarrow} A$ as in \eqref{e:del-1} and $1-h(k)/k \rightarrow 1$. $\Delta_{3k}$ can be further simplified as
	\small
	$$\Delta_{3k}\leq\underbrace{\max_{0\leq k_0<h(k)}\Big|\frac{k-k_0}{kg(Y_{(n-k,n)})}S_{k_0,k}+c{\Big(\frac{k_0}{k}\Big)}^{1+\rho}-\frac{c}{1+\rho}\Big|}_{\Delta_{4k}}+\underbrace{\max_{0\leq k_0<h(k)}\Big|\frac{c}{1+\rho}-c{\Big(\frac{k_0}{k}\Big)}^{1+\rho}-\frac{cA(k-k_0)}{k(1+\rho)\Delta_{1k}}\Big|}_{\Delta_{5k}}$$
	\normalsize where
	$\Delta_{4k}\stackrel{P}{\longrightarrow}0$ by Lemma \ref{lem:r-def}. Since $\max_{0\leq k_0<k}(k_0/k)^{1+\rho}\leq (h(k)/k)^{1+\rho}\rightarrow 0$, thus to prove $\Delta_{5k} \stackrel{P}{\longrightarrow}0$, it suffices to show that
	$$\max_{0\leq k_0<h(k)}\Big|\frac{c}{1+\rho}-\frac{cA(k-k_0)}{k(1+\rho)\Delta_{1k}}\Big|\stackrel{P}{\longrightarrow}0$$
	In this direction, we observe that
	\begin{eqnarray*}
		\max_{0\leq k_0 \leq h(k)}\Big|\frac{c}{1+\rho}-\frac{cA(k-k_0)}{k(1+\rho)\Delta_{1k}}\Big|&\leq& \frac{|c|}{1+\rho}\max_{0\leq k_0<h(k)}\Bigg(\Big|1-\frac{ A}{\Delta_{1k}}\Big|+\frac{ Ak_0}{k\Delta_{1k}}\Bigg)\\
		&\leq&\frac{|c|}{1+\rho}\Bigg(\Big|1-\frac{ A}{\Delta_{1k}}\Big|+\frac{ Ah(k)}{\Delta_{1k}k}\Bigg) \stackrel{P}{\longrightarrow}0
	\end{eqnarray*}
	since $h(k)/k \rightarrow 0$ and $A/\Delta_{1k} \stackrel{P}{\longrightarrow} 1$ as in \eqref{e:del-1}. This completes the proof.
\end{proof}

\subsection{Proofs for Section \ref{sec:aut-trim}} \label{sec:proofs-sec4}

\begin{proof}[{\bf Proof of Theorem \ref{prop:T-conv}}]
	From \eqref{e:T-i-k} we have
	\begin{eqnarray*}
		\label{e:T-conv-mod}
		k^{\delta}\max_{0\leq k_0<h(k)}|T_{k_0,k}-T^*_{k_0,k}|&=&k^\delta \max_{0\leq k_0 <h(k)}\frac{k-k_0-1}{k-k_0}\Big|\frac{\widehat{\xi}_{k_0+1,k}}{\widehat{\xi}_{k_0,k}}-\frac{\widehat{\xi}^*_{k_0+1,k}}{\widehat{\xi}^*_{k_0,k}}\Big|\\
		&\leq&\frac{k^{\delta}}{1-h(k)/k}\max_{0\leq k_0 <h(k)}\underbrace{\Big|\frac{\widehat{\xi}_{k_0+1,k}}{\widehat{\xi}_{k_0,k}}-\frac{\widehat{\xi}^*_{k_0+1,k}}{\widehat{\xi}^*_{k_0,k}}\Big|}_{W_{k_0,k}}
	\end{eqnarray*}
	Since $h(k)=o(k)$, to prove \eqref{e:T-conv-mod}, it remains to show $
	k^\delta \max_{0\leq k_0 <h(k)}W_{k_0,k} \stackrel{P}{\longrightarrow}0$.
	In this direction, we observe that
	\begin{eqnarray*}
		W_{k_0,k} &\leq&\Big|\frac{\widehat{\xi}_{k_0+1,k}}{\widehat{\xi}_{k_0,k}}-\frac{\widehat{\xi}^*_{k_0+1,k}}{\widehat{\xi}_{k_0,k}}-\frac{cAk^{-\delta}}{(1+\rho)\widehat{\xi}_{k_0,k}}\Big|+\frac{|c|Ak^{-\delta}}{(1+\rho)\widehat{\xi}_{k_0,k}}\Big|1-\frac{\widehat{\xi}^*_{k_0+1,k}}{\widehat{\xi}^*_{k_0,k}}\Big|\\
		&+&\Big|\frac{\widehat{\xi}^*_{k_0+1,k}}{\widehat{\xi}_{k_0,k}}-\frac{\widehat{\xi}^*_{k_0+1,k}}{\widehat{\xi}^*_{k_0,k}}+\frac{cAk^{-\delta}}{(1+\rho)\widehat{\xi}_{k_0,k}}\frac{\widehat{\xi}^*_{k_0+1,k}}{\widehat{\xi}^*_{k_0,k}}\Big|\\
		&=&\frac{1}{\widehat{\xi}_{k_0,k}}\Bigg(\Big|R_{k_0,k}-\frac{cAk^{-\delta}}{1+\rho}\Big|+\frac{|c|Ak^{-\delta}}{(1+\rho)}\Big|1-\frac{\widehat{\xi}^*_{k_0+1,k}}{\widehat{\xi}^*_{k_0,k}}\Big|+\frac{\widehat{\xi}^*_{k_0+1,k}}{\widehat{\xi}^*_{k_0,k}}\Big|\frac{cAk^{-\delta}}{(1+\rho)}-R_{k_0+1,k}\Big|\Bigg)\\
	\end{eqnarray*}
	\normalsize
	where $R_{k_0,k}$ is defined in \eqref{e:r-def}. Thus to show $
	k^\delta \max_{0\leq k_0 <h(k)}W_{k_0,k} \stackrel{P}{\longrightarrow}0$,  it 
	\begin{eqnarray*}
		\max_{0\leq k_0 <h(k)}k^\delta W_{k_0,k}&\leq& \Big(M_{1k}+\frac{|c|A}{(1+\rho)}\max_{0\leq k_0  h(k)}|1-B_{k_0,k}|+M_{1k}\max_{0\leq k_0 \leq h(k)}B_{k_0,k}\Big)\max_{0\leq k_0 <h(k)} \frac{1}{\widehat{\xi}_{k_0,k}}\\
		&=&\Big(M_{1k}\max_{0\leq k_0 \leq h(k)}(1+B_{k_0,k})+\frac{|c|A}{(1+\rho)}\max_{0\leq k_0 \leq h(k)}|1-B_{k_0,k}|\Big)\max_{0\leq k_0 <h(k)} \frac{1}{\widehat{\xi}_{k_0,k}}
	\end{eqnarray*}
	where $M_{1k}=k^\delta \max_{0\leq k_0<h(k)}|R_{k_0,k}-(k^{-\delta}cA)/(1+\rho)| \stackrel{P}{\longrightarrow}0$ is a direct consequence of Theorem \ref{prop:E-conv}.
	Using \eqref{e:xi-jt}, we next observe that 
	\begin{eqnarray}
	\label{e:B-conv}
	\max_{0\leq k_0 \leq h(k)}|1-B_{k_0,k}|&\stackrel{d}{=}&\max_{0\leq k_0 \leq h(k)}\Big|1-\frac{\Gamma_{k-k_0-1}/(k-k_0-1)}{\Gamma_{k-k_0}/(k-k_0)}\Big|\\\nonumber
	&\leq &\frac{1}{1-h(k)/k}\max_{k-h(k)\leq i \leq k}\Big|\frac{\Gamma_{i}/i}{\Gamma_{i+1}/(i+1)}-1\Big|\stackrel{a.s.}{\longrightarrow}0
	\end{eqnarray}  is a direct consequence of   \eqref{e:gam-conv-1} in Lemma \ref{lem:ratio-conv}. \eqref{e:B-conv} also proves that $\max_{0\leq k_0 \leq h(k)}(1+B_{k_0,k})$ is bounded in probabibilty. 
	
	Thus, to complete the proof of $
	k^\delta \max_{0\leq k_0 <h(k)}W_{k_0,k} \stackrel{P}{\longrightarrow}0$, we show that $\min_{0\leq k_0 <h(k) }|\widehat{\xi}_{k_0,k}|$ is bounded away from 0 in probability as follows:
	\begin{eqnarray}
	\label{e:xi-bound}
	\min_{0\leq k_0 <h(k) }\widehat{\xi}_{k_0,k}\geq \min_{0\leq k_0 <h(k) }\widehat{\xi}^*_{k_0,k}-\max_{0\leq k_0 <h(k) }|\widehat{\xi}_{k_0,k}-\widehat{\xi}^*_{k_0,k}|
	\end{eqnarray}
	For $\delta>0$,  Theorem \ref{prop:E-conv} implies $\max_{0\leq k_0 <h(k) }|\widehat{\xi}_{k_0,k}-\widehat{\xi}^*_{k_0,k}| \stackrel{P}{\longrightarrow}0$. Therefore $\min_{0 \leq k_0 <h(k)}\widehat{\xi}_{k_0,k}$ is bounded away from 0 as long as  $\min_{0 \leq k_0 <h(k)}\widehat{\xi}^*_{k_0,k}$ is bounded away from 0. This can be easily shown because
	$$\min_{0 \leq k_0 <h(k)}\widehat{\xi}^*_{k_0,k}\stackrel{d}{=}\min_{0 \leq k_0 <h(k)}\frac{\Gamma_{k-k_0}}{k-k_0}\geq 1-\max_{k-h(k) \leq i <k}\Big|\frac{\Gamma_{i}}{i}-1\Big|\stackrel{a.s.}{\longrightarrow}1$$
	where the last convergence is a direct consequence of \eqref{e:gam-conv-0} in Lemma \ref{lem:ratio-conv}. This completes the proof.
	
\end{proof}

\section{Supplement}
\label{sec:supple}

\begin{lemma}
	\label{lem:Hall}
	Assumption \eqref{e:D-new-def} implies there exist $M>0$ such that
	\begin{equation}
	\label{e:rem-hall}
	\inf_{F \in {\cal{D}}_{\xi}(B,\rho)}P_F\Big[\max_{0\leq k_0<k}\sqrt{k}|R_{k_0,k}|\leq M\Big]\rightarrow 1 \textmd{ as }h(k) \rightarrow \infty
	\end{equation}
	where $R_{k_0,k}$ is defined in \eqref{e:r-def} and $k=O(n^{2\rho/(1+2\rho)})$.
\end{lemma}

\begin{proof}
	By \eqref{e:D-new-def}, we have $1-Bx^{-\rho}\leq L(x)\leq 1+Bx^{-\rho}$. Therefore
	\begin{eqnarray}
	\label{e:r-ub}
	(k-k_0)R_{k_0,k}&\leq& (k_0+1) \log\frac{1+BY^{-\rho}_{(n-k_0,n)}}{1-BY^{-\rho}_{(n-k,n)}} +\sum_{i=k_0+2}^{k} \log\frac{1+BY^{-\rho}_{(n-i+1,n)}}{1-BY^{-\rho}_{(n-k,n)}} \\\nonumber
	&\leq& k\log\frac{1+BY^{-\rho}_{(n-k,n)}}{1-BY^{-\rho}_{(n-k,n)}}
	\end{eqnarray}
	since $Y^{-\rho}_{(n-k,n)}\geq Y^{-\rho}_{(n-i+1,n)}$ for $i=k_0+1, \cdots,k$.  Similarly, we also have \begin{equation}
	\label{e:r-lb}
	(k-k_0)R_{k_0,k}\geq k\log\frac{1-BY^{-\rho}_{(n-k,n)}}{1+BY^{-\rho}_{(n-k,n)}}
	\end{equation}and thus, \eqref{e:r-ub} and \eqref{e:r-lb} together imply
	\begin{equation}
	\label{e:r-mod-ub}
	\max_{0\leq k_0<h(k)}\sqrt{k}|R_{k_0,k}|\leq \frac{\sqrt{k}Y^{-\rho}_{(n-k,n)}}{1-h(k)/k}\max_{0\leq k_0<h(k)}\frac{1}{Y^{-\rho}_{(n-k,n)}}\log\frac{1+BY^{-\rho}_{(n-k,n)}}{1-BY^{-\rho}_{(n-k,n)}}\end{equation}
	where $h(k)=o(k)$ and 
	\begin{eqnarray*}
		\sqrt{k}Y^{-\rho}_{(n-k,n)}\frac{1}{Y^{-\rho}_{(n-k,n)}}\log\frac{1+BY^{-\rho}_{(n-k,n)}}{1-BY^{-\rho}_{(n-k,n)}}&
		\stackrel{d}{=}&\underbrace{\sqrt{k}(\Gamma_{k+1}/\Gamma_{n+1})^{\rho}}_{\Delta_{1k}}\underbrace{\frac{1}{(\Gamma_{k+1}/\Gamma_{n+1})^{\rho}}\log\frac{ 1+B(\Gamma_{k+1}/\Gamma_{n+1})^{\rho}}{1-B(\Gamma_{k+1}/\Gamma_{n+1})^{\rho}}}_{\Delta_{2k}}
	\end{eqnarray*}
	Now, by relation \eqref{e:gam-conv-0} in Lemma \ref{lem:ratio-conv}, we have $((\Gamma_{k+1}/k)/(\Gamma_{n+1}/n))^{\rho} \stackrel{a.s.}{\longrightarrow}1$. For $k=O(n^{2\rho/(1+2\rho)})$, $\Delta_{1k}$ is bounded almost surely and  $\Delta_{2k}\stackrel{a.s.}{\longrightarrow}2B$.
	Therefore there exist $M$  such that 	$$\inf_{F \in {\cal{D}}_{\xi}(B,\rho)}P_F\Big[\max_{0\leq k_0<k}\frac{k-k_0}{kY^{-\rho}_{(n-k,n)}}|R_{k_0,k}|\leq M\Big]\geq P[\Delta_{1k}\Delta_{2k}\leq M] \rightarrow 1$$
	This completes the proof.

\end{proof}
\begin{proof}[{\bf Proof of Theorem \ref{t:uniform-consistency}}] Let $P_n=\inf_{F \in {\cal{D}}_{\xi}(B,\rho)}P_F\Big[\max_{0\leq k_0<h(k)}|\widehat{\xi}_{k_0,k}-\xi|\leq a(n)\Big]$, then
	
	$$P_n=\inf_{{\cal{D}}_{\xi}(B,\rho)}P_F\Big[\underbrace{\max_{0\leq k_0<h(k)}\sqrt{k}|R_{k_0,k}|\leq (\sqrt{k}a(n))/2}_{A_{1n}}\cap\underbrace{\max_{0\leq k_0<h(k)}\sqrt{k}|\widehat{\xi}^*_{k_0,k}-\xi|\leq (\sqrt{k}a(n))/2}_{A_{2n}}\Big]$$
	Since $\sqrt{k}a(n) \rightarrow \infty$, by Lemma \ref{lem:Hall}, $\inf_{F \in {\cal{D}}_{\xi}(B,\rho)}P_F[A_{1n}] \rightarrow 1$. We also have that, $$\inf_{F \in {\cal{D}}_{\xi}(B,\rho)}P_F[A_{2n}]=P\Big[\max_{0\leq k_0<h(k)} \sqrt{k}| \widehat {\xi}_{k_0,k}^*(n) - \xi|\leq(\sqrt{k}a(n))/2\Big]$$ since $\widehat{\xi}^*_{k_0,k}$ does not depend on $F \in {\cal{D}}_{\xi}(B,\rho)$.

	By using Donsker's principle, we will show that 
	$$
	\max_{0\le k_0 <h(k)} | \widehat \xi_{k_0,k}^*(n) - \xi|  =o_P(a(n)),
	$$
	which will imply $P_F(A_{2n}) \rightarrow 1$.  Indeed, without loss of generality, suppose $\xi=1$ and let 
	$E_i,\ i=1,2,\dots$ be independent standard exponential random variables. For every $\epsilon\in (0,1)$, we have that
	\begin{equation}\label{e:Wkprocess}
	W_k =\{W_k(t),\ t\in [\epsilon,1]\} :=  \left\{\frac{\sqrt{k}}{[kt]} \sum_{i=1}^{[kt]} (E_i -1),\ t \in [0,1] \right\} \stackrel{d}{\to } \{B(t)/t, \ t\in [\epsilon,1]\},
	\end{equation}
	as $k\to\infty$, where $B = \{B(t),\ t\in [0,1]\}$ is the standard Brownian motion, and where the last convergence is in 
	the space of cadlag functions ${\mathbb D}[\epsilon,1]$ equipped with the Skorokhod $J_1$-topology.  (In fact, since the limit
	has continuous paths, the convergence is also valid in the uniform norm.)
	
	Recall that by \eqref{e:xi-jt-big}, we have 
	$$
	\{\widehat\xi_{k_0,k}^*(n),\ 0\le k_0 < k\}\stackrel{d}{=} \left\{\sum_{i=1}^{k-k_0} E_i/(k-k_0),\ 0\le k_0 < k\right\}.
	$$
	Thus, 
	\begin{equation}\label{e:t:uniform-consistency-1}
	\sqrt{k} \max_{0\le k_0<h(k)} | \widehat \xi_{k_0,k}^*(n) - \xi|  \stackrel{d}{= } \sup_{t \in [1-h(k)/k, 1]} |W_k(t)| \le
	\sup_{t \in [\epsilon,1]} |W_k(t)|, 
	\end{equation}
	where the last inequality holds for all sufficiently large $k$, since $1-h(k)/k\to1$, as $k\to\infty$.  
	Since the supremum is a continuous functional in $J_1$, the convergence in \eqref{e:Wkprocess} implies that the 
	right--hand side of \eqref{e:t:uniform-consistency-1} converges in distribution
	to $\sup_{t\in [\epsilon,1]} |B(t)/t| = O_P(1)$, which is finite with probability one. This, since $a(n)\sqrt{k(n)} \to \infty$, completes the proof.
\end{proof}

\begin{proof}[{\bf Proof of Theorem \ref{prop:U-conv}}]
	We first begin with the proof of \eqref{e:U-conv}. For this from \eqref{e:U-i-k} we have
	\begin{eqnarray*}
		\max_{0\leq k_0<h(k)}|U_{k_0,k}-U^*_{k_0,k}|&=&\max_{0\leq k_0<h(k)}2\Big||{(T_{k_0,k})}^{k-k_0-1}-0.5|-|{(T^*_{k_0,k})}^{k-k_0-1}-0.5|\Big|\\
		&\leq&2\max_{0\leq k_0<h(k)}\Big||{(T_{k_0,k})}^{k-k_0-1}-{(T^*_{k_0,k})}^{k-k_0-1}|\Big|\\
		&\leq&2\max_{0\leq k_0<h(k)}\Big|\Big(\frac{T_{k_0,k}}{T^*_{k_0,k}}\Big)^{k-k_0-1}-1\Big|
	\end{eqnarray*}
	where the last inequality holds since $T^*_{k_0,k}\leq 1$(see \ref{T-def}). Thus to prove \eqref{e:U-conv}, it suffices to show
	\begin{equation}
	\label{e:U-conv-mod}
	k^{\delta-1}\max_{0\leq k_0<h(k)}\Big|\Big(\frac{T_{k_0,k}}{T^*_{k_0,k}}\Big)^{k-k_0-1}-1\Big|\stackrel{P}{\longrightarrow}0
	\end{equation}
	To prove \eqref{e:U-conv-mod}, we begin by showing  \begin{equation} \label{e:TU-conv-mod}
	k^{\delta}\max_{0\leq k_0<h(k)}|\frac{T_{k_0,k}}{T^*_{k_0,k}}-1|\stackrel{P}{\longrightarrow}0.\end{equation}
	In this direction, from \eqref{e:T-conv}, we observe that
	$$k^{\delta}\max_{0\leq k_0<h(k)}\Big|\frac{T_{k_0,k}}{T^*_{k_0,k}}-1\Big|\leq \frac{1}{\min_{0\leq h(k)}T^*_{k_0,k}}\max_{0\leq k_0<h(k)}\underbrace{k^{\delta}|T_{k_0,k}-T^*_{k_0,k}|}_{\Delta_{k}}$$
	From \eqref{e:T-conv}, we have $\Delta_{k}\stackrel{P}{\longrightarrow}0$. Thus \eqref{e:TU-conv-mod} holds as long as $\min_{0\leq k_0< h(k)}T^*_{k_0,k}$ is bounded away from 0 in probability. This can be easily seen as follows
	$$\min_{0 \leq k_0 <h(k)}T^*_{k_0,k}\stackrel{d}{=}\min_{0 \leq k_0 <h(k)}\frac{\Gamma_{k-k_0-1}/(k-k_0-1)}{\Gamma_{k-k_0}/(k-k_0)}\geq 1-\max_{k-h(k) \leq i <k}\Big|\frac{\Gamma_{i}/i}{\Gamma_{i+1}/(i+1)}-1\Big|\stackrel{a.s.}{\longrightarrow}1$$
	where the last convergence is a direct consequence of \eqref{e:gam-conv-1} in Lemma \ref{lem:ratio-conv}.  
	
	In view of \eqref{e:T-conv-mod}, for a subsequence, $\{k_l\}$ there exists a further subsequence ${\widetilde{k}}$ such that
	$$\widetilde{k}^{\delta}\max_{0\leq k_0<h(\widetilde{k})}\Big|\frac{T_{k_0,\widetilde{k}}}{T^*_{k_0,\widetilde{k}}}-1\Big|\stackrel{a.s.}{\longrightarrow}0$$
	This implies there exists $M$ such that for every $\widetilde{k}\geq M$ and $0\leq k_0<h(\widetilde{k})$,
	\begin{equation} \label{e:ineq-1}
	1-\frac{\epsilon}{\widetilde{k}^\delta}\leq\Big(\frac{T_{k_0,\widetilde{k}}}{T^*_{k_0,\widetilde{k}}}\Big)(\omega)\leq 1+\frac{\epsilon}{\widetilde{k}^\delta}\end{equation}
	for all $\omega \in \Omega$ with $P[\Omega]=1$. \eqref{e:ineq-1} further implies
	\begin{eqnarray*}
		\underbrace{\widetilde{k}^{\delta-1}\Big (\Big(1-\frac{\epsilon}{\widetilde{k}^\delta}\Big)^{\widetilde{k}-h(\widetilde{k})-1}-1\Big)}_{-a_{\widetilde{k}}}\leq& \widetilde{k}^{\delta-1}\Big({\Big(\frac{T_{k_0,\widetilde{k}}}{T^*_{k_0,\widetilde{k}}}\Big)}^{\widetilde{k}-k_0-1}(\omega)-1\Big)&\leq \underbrace{\widetilde{k}^{\delta-1}\Big (\Big(1+\frac{\epsilon}{\widetilde{k}^\delta}\Big)^{\widetilde{k}-1}-1\Big)}_{b_{\widetilde{k}}}\end{eqnarray*}
	Therefore,
	\begin{equation}
	\label{e:T-lim-1}
	\widetilde{k}^{{\delta-1}}\max_{0\leq k_0<h(\widetilde{k})}\Big|\Big(\frac{T_{k_0,\widetilde{k}}}{T^*_{k_0,\widetilde{k}}}\Big)^{\widetilde{k}-k_0-1}(w)-1\Big|\leq a_{\widetilde{k}}\vee b_{\widetilde{k}}
	\end{equation}
	First observe that both the sequences $a_{\widetilde{k}}$ and $b_{\widetilde{k}}$ converge to $\epsilon$ as $\widetilde{k}\rightarrow \infty$. Thereby, taking limsup w.r.t $\widetilde{k}$ on both sides of \eqref{e:T-lim-1}, we get
	\begin{equation}
	\label{e:T-lim-2}\limsup_{\widetilde{k}\rightarrow \infty}\widetilde{k}^{{\delta-1}}\max_{0\leq k_0<h(\widetilde{k})}\Big|\Big(\frac{T_{k_0,\widetilde{k}}}{T^*_{k_0,\widetilde{k}}}\Big)^{\widetilde{k}-k_0-1}(w)-1\Big| \leq \epsilon
	\end{equation}
	Since \eqref{e:T-lim-2} holds for all $\epsilon>0$, we have
	$$\widetilde{k}^{{\delta-1}}\max_{0\leq k_0<h(\widetilde{k})}\Big|\Big(\frac{T_{k_0,\widetilde{k}}}{T^*_{k_0,\widetilde{k}}}\Big)^{\widetilde{k}-k_0-1}(w)-1\Big|\rightarrow 0$$
	This entails the proof of convergence in probability of \eqref{e:U-conv-mod}.
	
	We next begin with the proof \eqref{e:type-I-gen}. To this end, we have
	\begin{eqnarray*} 
		1-P_{H_0}[\widehat{k}_0=0]
		&=&1-P_{H_0}\Big[\underbrace{\bigcap_{i=0}^{f(k)}\{U_{i,k}<(1-q)^{ca^{k-i-1}}\}}_{A_k}\Big]
	\end{eqnarray*}where we shall show $P[A_k]\rightarrow 1-q$ as follows.
	\begin{eqnarray*}
		P_{H_0}\Big[A_k\Big]&\leq & P_{H_0}\Big[A_k \cap \underbrace{ \{k^{\delta-1}\max_{0\leq i< f(k)}(U_{i,k}-U^*_{i,k})<\epsilon\}}_{B_{1k}}\Big]+P[ \{k^{\delta-1}\max_{0\leq i< f(k)}(U_{i,k}-U^*_{i,k})>\epsilon\}]\\
		&\leq& P_{H_0}\Big[\underbrace{\bigcap_{i=0}^{f(k)}\{U^*_{i,k}<(1-q)^{ca^{k-i-1}}+\epsilon k^{1-\delta}\}}_{A^*_{1k}}\Big]+P[B_{1k}^c] \hspace{4mm}(\textmd{since $A_k \cap B_{1k} \implies A^*_{1k}$})
	\end{eqnarray*}
	where $P[B_{1k}^c] \rightarrow 0$ by \eqref{e:U-conv}. It remain to show $P_{H_0}[A^*_{1k}] \rightarrow 1-q$. In this direction, we observe that
	\begin{eqnarray*}
		P_{H_0}[A^*_{1k}] &=& \prod_{i=0}^{f(k)}(1-q)^{ca^{k-i-1}}\prod_{i=0}^{f(k)}\Big(1+\frac{\epsilon k^{1-\delta}}{(1-q)^{ca^{k-i-1}}}\Big)\\
		&\leq& \underbrace{\Big(1-q\Big)^{ \frac{a^{(k-1)}-a^{(k-f(k)-2})}{a^{(k-1)}-1}}}_{c_{0k}}\underbrace{\Big(1+\frac{\epsilon}{(1-q) k^{\delta-1}}\Big)^{f(k)}}_{c_{1k}}  \hspace{2mm}(\textmd{since $(1-q)^{ca^{k-i-1}}\geq (1-q)$})
	\end{eqnarray*}
	Now for $f(k) \rightarrow \infty$, $c_{0k} \rightarrow 1-q$. For $f(k)=O(k^{\delta-1})$, $\limsup_{k \rightarrow \infty}c_{1k} \leq(1+ M\epsilon/(1-q))$ for some $M>0$. Thus $$\limsup_{k \rightarrow \infty}P_{H_0}[A^*_{1k}] \leq (1-q)+M\epsilon$$
	holds for every $\epsilon >0$ which implies $\limsup_{k \rightarrow \infty}P_{H_0}[A^*_{1k}]\leq (1-q)$.
	Additionally,
	\begin{eqnarray*}
		P_{H_0}\Big[A_k\Big]&\geq& P_{H_0}\Big[A_k \cap \underbrace{\{k^{\delta-1}\max_{0\leq i< f(k)}(U_{i,k}-U^*_{i,k})>-\epsilon\}}_{B_{2k}}\Big]\\
		&\geq& P_{H_0}\Big[\underbrace{\bigcap_{i=0}^{f(k)}\{U^*_{i,k}<(1-q)^{ca^{k-i-1}}-\epsilon k^{1-\delta}\}}_{A^*_{2k}}  \Big]-P_{H_0}[B^c_{2k}]\hspace{4mm}(\textmd{since $A^*_{2k}\cap B_{2k} \implies A_k\cap B_{2k} $})\\
	\end{eqnarray*}
	where $P[B_{2k}^c] \rightarrow 0$ by \eqref{e:U-conv}. It remain to show $P_{H_0}[A^*_{2k}] \rightarrow 1-q$. In this direction, we observe that 
	\begin{eqnarray*}
		P_{H_0}[A^*_{2k}]&=&\prod_{i=0}^{f(k)}(1-q)^{ca^{k-i-1}}\prod_{i=0}^{f(k)}\Big(1-\frac{\epsilon k^{1-\delta}}{(1-q)^{ca^{k-i-1}}}\Big)\\
		&\geq& \underbrace{\Big(1-q\Big)^{ \frac{a^{(k-1)}-a^{(k-f(k)-2})}{a^{(k-1)}-1}}}_{c_{0k}}\underbrace{\Big(1-\frac{\epsilon}{(1-q) k^{\delta-1}}\Big)^{f(k)}}_{c_{2k}}  \hspace{2mm}(\textmd{since $(1-q)^{ca^{k-i-1}}\geq (1-q)$})
	\end{eqnarray*}	
	Now for $f(k) \rightarrow \infty$, $c_{0k} \rightarrow 1-q$.  For $f(k)=O(k^{\delta-1})$, $\limsup_{k \rightarrow \infty}c_{1k} \geq(1- M\epsilon/(1-q))$ for some $M>0$. Thus $$\limsup_{k \rightarrow \infty}P_{H_0}[A^*_{2k}] \geq (1-q)-M\epsilon$$
	holds for every $\epsilon >0$ which implies $\limsup_{k \rightarrow \infty}P_{H_0}[A^*_{2k}]\geq (1-q)$.
	
	Thus $\lim_{k \rightarrow \infty} P_{H_0}[A^*_{2k}]=1-q$. This completes the proof.
\end{proof}

\end{document}